%% file: Planted_Matching_Problem.tex
\documentclass[11pt]{article}
\usepackage{fullpage}
\usepackage{amsmath,amsthm,amsfonts,amssymb}
\usepackage{graphicx,enumerate}
\usepackage[title]{appendix}
\usepackage[
CJKbookmarks=true,
bookmarksnumbered=true,
bookmarksopen=true,
colorlinks=true,
citecolor=red,
linkcolor=blue,
anchorcolor=red,
urlcolor=blue
]{hyperref}

\usepackage{caption}
\usepackage{subcaption}
\usepackage{enumitem}

\makeatletter
\DeclareFontEncoding{LS1}{}{}
\DeclareFontSubstitution{LS1}{stix}{m}{n}
\DeclareMathAlphabet{\mathscr}{LS1}{stixscr}{m}{n}
\makeatother

\newcommand{\floor}[1]{\lfloor #1 \rfloor}

\newcommand{\expect}{\mathbb{E}}

\newcommand{\M}{\mathscr{M}}
\newcommand{\X}{\mathcal{X}}
\newcommand{\e}{\mathrm{e}}
\newcommand{\prob}{\mathbb{P}}

\DeclareMathOperator*{\argmin}{arg\,min}
\DeclareMathOperator*{\argmint}{{arg\,min}^{[2]}}
\DeclareMathOperator*{\mint}{{min}^{[2]}}
\newcommand{\lra}{\leftrightarrow}

\newcommand*\xbar[1]{\hbox{\vbox{
			\hrule height 0.5pt 
			\kern0.4ex
			\hbox{\kern-0.1em
				\ensuremath{#1}%
				\kern-0.1em
	}}}
}

\usepackage{tikz}
\usepackage{tikz-qtree}
\usepackage{mathtools}
\usetikzlibrary{calc}
\usetikzlibrary{decorations.pathreplacing,decorations.markings}
\tikzstyle{dot}=[circle,fill,black,inner sep=1pt]

\newtheorem{theorem}{Theorem}
\newtheorem{lemma}{Lemma}
\newtheorem{proposition}{Proposition}
\newtheorem{corollary}{Corollary}
\newtheorem{definition}{Definition}
\newtheorem{remark}{Remark} 
\newtheorem*{observation}{Observation}

\usepackage{color}
\usepackage{algorithm}
\usepackage{algorithmic}
\usepackage{xspace}
\newcommand{\ie}{i.e.\xspace}

\renewcommand{\hat}{\widehat}
\renewcommand{\tilde}{\widetilde}
\newcommand{\overlap}{\mathrm{overlap}}

\newcommand{\rplus}{\mathbb{R}_+}
\newcommand{\nplus}{\mathbb{N}}
\renewcommand{\root}{\textrm{\o}}
\newcommand{\lwc}{\xrightarrow{\textrm{loc}}}
\newcommand{\one}{{\boldsymbol{1}}}

\newcommand{\Mmin}{M_{\textrm{min}}}
\newcommand{\Mplanted}{M^*}
\newcommand{\opt}{\textrm{opt}}
\newcommand{\Mopt}{M_{\min}}
\newcommand{\Mnopt}{M_{n,\min}}
\newcommand{\Minfopt}{\mathscr{M}_{\infty,\opt}}
\newcommand{\concat}[2]{#1 #2}
\newcommand{\bs}[1]{\boldsymbol{#1}}

\newcommand{\parent}{\textrm{parent}}
\newcommand{\gen}{\textrm{gen}}
\newcommand{\symdiff}{\triangle}
\newcommand{\dx}{\mathrm{d}x}
\newcommand{\dy}{\mathrm{d}y}
\newcommand{\ds}{\mathrm{d}s}
\newcommand{\dz}{\mathrm{d}z}
\newcommand{\dF}{\mathrm{d}F}
\newcommand{\dG}{\mathrm{d}G}
\newcommand{\dV}{\mathrm{d}V}
\newcommand{\dW}{\mathrm{d}W}
\newcommand{\dU}{\mathrm{d}U}
\newcommand{\dt}{\mathrm{d}t}

\newcommand{\Gs}{G_\circ}
\newcommand{\Gsp}{G'_\circ}

\newcommand{\Ns}{N_\circ}
\newcommand{\Nsp}{N'_\circ}
\newcommand{\Nso}{N_{\circ\circ}}
\newcommand{\Gstar}{\mathcal{G}_*}

\newcommand{\Gstarstar}{\mathcal{G}_{**}}
\newcommand{\Gmadstarstar}{\widehat{\mathcal{G}}_{**}}

\title{The Planted Matching Problem:\\ 
	Phase Transitions and Exact Results}

\author{Mehrdad Moharrami, Cristopher Moore, and Jiaming Xu\thanks{
		M.\ Moharrami (Electrical and Computer Engineering, U. Michigan, Ann Arbor) is supported by the Rackham Predoctoral Fellowship, a departmental Graduate Student Instructor appointment and by NSF Grants AST-1443972/AST-1516075.
		\texttt{moharami@umich.edu}.
		C.\ Moore (Santa Fe Institute and currently a Visiting Researcher at Microsoft New England) is supported in part by NSF Grant IIS-1838251. \texttt{moore@santafe.edu}.
		J.\ Xu (The Fuqua School of Business, Duke University, Durham NC, USA) is supported by NSF Grants IIS-1838124, CCF-1850743, and CCF-1856424. \texttt{jx77@duke.edu}.
}}
\date{\today}

\begin{document}
	
\pgfdeclarelayer{background}
\pgfdeclarelayer{foreground}
\pgfsetlayers{background,main,foreground}

\maketitle

\begin{abstract}
We study the problem of recovering a planted matching in randomly weighted complete bipartite graphs $K_{n,n}$. For some unknown perfect matching $\Mplanted$, the weight of an edge is drawn from one distribution $P$ if $e \in \Mplanted$ and another distribution $Q$ if $e \notin \Mplanted$. Our goal is to infer $\Mplanted$, exactly or approximately, from the edge weights. In this paper we take $P=\exp(\lambda)$ and $Q=\exp(1/n)$, in which case the maximum-likelihood estimator of $\Mplanted$ is the minimum-weight matching $\Mmin$. We obtain precise results on the overlap between $\Mplanted$ and $\Mmin$, i.e., the fraction of edges they have in common. For $\lambda \ge 4$ we have almost perfect recovery, with overlap $1-o(1)$ with high probability. For $\lambda < 4$ the expected overlap is an explicit function $\alpha(\lambda) < 1$: we compute it by generalizing Aldous' celebrated proof of the $\zeta(2)$ conjecture for the un-planted model, using local weak convergence to relate $K_{n,n}$ to a type of weighted infinite tree, and then deriving a system of differential equations from a message-passing algorithm on this tree. 
\end{abstract}

\section{Introduction}\label{sec:intro}
\input{Sections/Introduction}

\section{Almost perfect recovery for $\lambda \ge 4$}\label{sec:first_moment_bound}
\input{Sections/Almostperfectrec}

\section{Exact results for the expected overlap when $\lambda < 4$}\label{sec:main-infinite-tree}
\input{Sections/ExactResult}

\section{Open questions} 
\input{Sections/Discussion}

\section*{Acknowledgements}
We are very grateful to Venkat Anantharam, Charles Bordenave, Jian Ding, David Gamarnik, Christopher Jones, Vijay Subramanian, Yihong Wu, and Lenka Zdeborov\'a for helpful conversations. C.M. is also grateful to Microsoft Research New England for their hospitality. We also thank an anonymous reviewer for helpful comments.

\bibliographystyle{plain}
\bibliography{bibliography}
\appendix
\section{Proof of Lemma~\ref{lmm:erlang-diff}}\label{app:erlang-diff}
\input{Sections/Proofs/erlangdiff}

\section{Analysis of system of ODEs}\label{app:ODE_analysis}
\input{Sections/odeanal}
\subsection{Basic properties of the solution}\label{app:basic_ode}
\input{Sections/basicprop}
\subsection{Monotonicity to the initial condition}\label{app:monotone_ode}
\input{Sections/initmonot}
\subsection{Limiting behavior of $(U,V,W)$}\label{app:limit_ode}
\input{Sections/limitode}
\subsection{Basins of attraction}\label{app:basin_ode}
\input{Sections/basinattract}
\subsection{Proof of Theorem \ref{thm:ode_unique_solution}}\label{app:finish_ode}
\input{Sections/finishode}

\section{Planted Networks and Local Weak Convergence}\label{app:gen_setup}
\input{Sections/GeneralSetup}

\section{The PWIT and Planted PWIT}\label{app:planted-PWIT}
\input{Sections/InfiniteObject}


\section{The Optimal Involution Invariant Matching on the Planted PWIT}
\label{app:optmatch}
\input{Sections/OptMatch}
\subsection{The Message-Passing Algorithm} 
\input{Sections/OptMatch-Heuristic}
\subsection{A Rigorous Construction of $\Minfopt$}
\label{app:OptMatch-Construct}
\input{Sections/OptMatch-Construct}
\subsection{Optimality of $\Minfopt$}
\input{Sections/OptMatch-Analysis}
\subsection{Uniqueness of The Solution of RDEs}
\input{Sections/OptMatch-RDEUnique}

\section{Convergence of the Minimum Matching on $(K_{n,n},\ell_{n})$ to $\Minfopt$}
\label{app:jointconv}
\input{Sections/JointConv}
\subsection{The Easy Half: A Simple Compactness Argument}
\input{Sections/EasyHalf_CompArg}
\subsection{The Harder Half}
\input{Sections/HardHalf_Intro}
\subsubsection{The Bi-infinite Planted PWIT} 
\label{app:BiinfTree}
\input{Sections/HardHalf_BiInf}
\subsubsection{The Unfolding Map}
\input{Sections/HardHalf_Unfold}
\subsubsection{Assigning Values to the edges of $(K_{n,n},\ell_n)$}
\input{Sections/HardHalf_AlMat}
\subsubsection{Construction of the Matching}
\input{Sections/MatchConst}

\section{Proof of Theorem~\ref{thm:lwc}}\label{app:proofoflwc}
\input{Sections/Proofs/Lwc}
\section{Proof of Proposition~\ref{prop:invol}}\label{app:proofofinvol}
\input{Sections/Proofs/Involinvar}
\section{Proof of Lemma~\ref{lem:BB}}
\label{app:proofofBB}
\input{Sections/Proofs/Ineqnoteq}

\end{document}

%% file: Sections/Introduction.tex
Consider a weighted complete bipartite graph $K_{n,n}$ with an unknown perfect matching $\Mplanted$, where for each edge $e$ the weight $w_e$ is independently distributed according to $P$ when $e \in \Mplanted$ and $Q$ when $e \notin \Mplanted$. The goal is to recover the ``hidden'' or ``planted'' matching $\Mplanted$ from the edge weights.

This problem is inspired by the long history of planted problems in computer science, where an instance of an optimization or constraint satisfaction problem is built around a planted solution in some random way. As we vary the parameters used to generate these instances, such as the size of a hidden clique or the density of communities in the stochastic block model of social networks, we encounter phase transitions in our ability to find this planted solution, exactly or approximately. In an inference problem, the instance corresponds to some noisy observation, such as a data set produced by a generative model, and the planted solution corresponds to the ground truth---the underlying structure we are trying to discover.

More concretely, we are motivated by the problem of tracking moving objects in a video, such as flocks of birds, motile cells, or particles in a fluid. Figure \ref{fig:particles}, taken from~\cite{Chertkov2010}, shows two frames of such a video, where each particle has moved from its original position by some amount. Our goal is then to find the most-likely matching between the two frames, assuming some probability distribution of these displacements.

\begin{figure}[H]
\begin{center}
\includegraphics[width=4in]{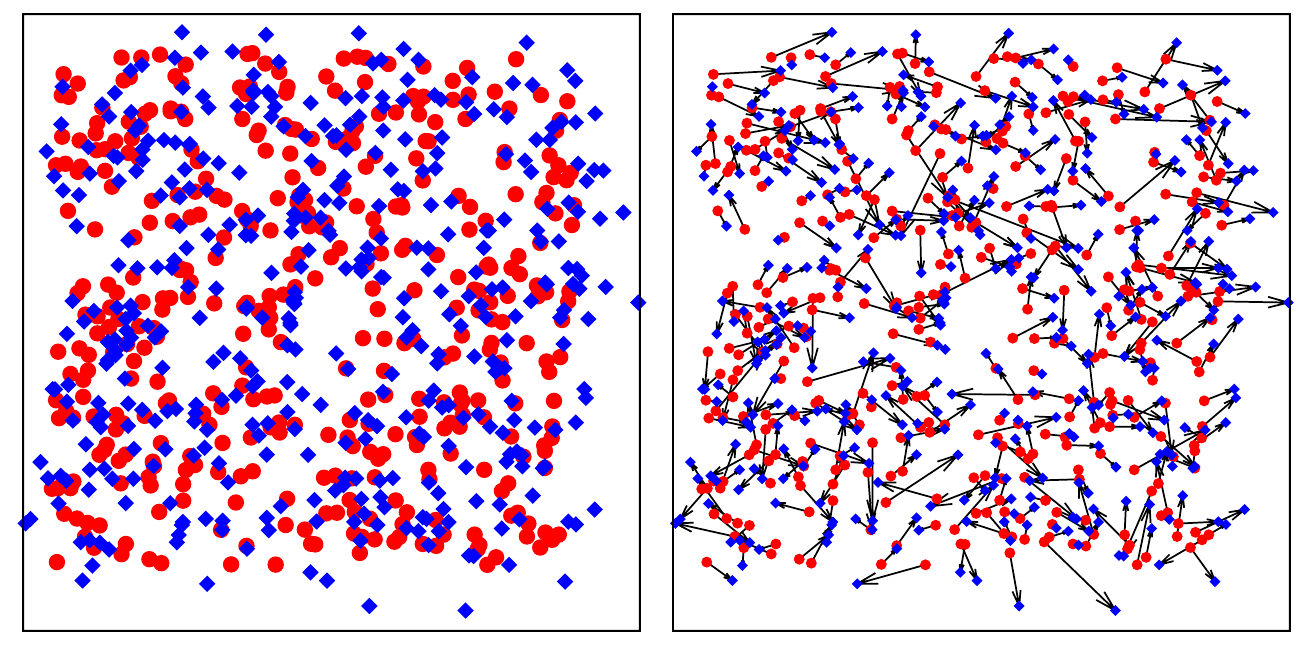}
\end{center}
\caption{On the left, the positions of particles in two frames of a video, with one frame red and the other blue. On the right, an inferred matching, hypothesizing how each particle has moved from one frame to the next. Taken from~\cite{Chertkov2010}.}
\label{fig:particles}
\end{figure}

For many planted problems such as Hidden Clique (e.g.~\cite{Barak2016}) or community detection in the stochastic block model (e.g.~\cite{Abbe2017,Moore2017}), there are two types of thresholds: information-theoretic and computational. When these are distinct, the region in between them has the interesting property that finding the planted solution, or at least approximating it better than chance, is information-theoretically possible but (conjecturally) computationally hard. These regions are also known as statistical-computational gaps. 

In the planted matching problem, one obvious estimator to try is the minimum weight matching (a.k.a.\ the linear assignment problem) which can be found in polynomial time. The natural question is then, as a function of the distributions $P$ and $Q$ on the planted and un-planted edges, how much the minimum matching $\Mmin$ has in common with the planted matching $\Mplanted$. In general, we define the \emph{overlap} of an estimator $M'$ with $\Mplanted$ as (assuming that $|M'|=n$)
\begin{equation}
\label{eq:overlap-def}
\overlap(\Mplanted,M') 
= 1 - \frac{1}{2n} | \Mplanted \symdiff M' | 
= \frac{1}{n} |\Mplanted \cap M'| \, . 
\end{equation} 
We say that $M'$ achieves \emph{almost perfect recovery} if $\expect[\overlap(\Mplanted,M')] = 1-o(1)$, or equivalently if $\overlap(\Mplanted,M') = 1-o(1)$ with high probability. We say that $M$ achieves \emph{partial recovery} if $\expect[\overlap(\Mplanted,M')] > 0$ as $n \to \infty$. 

Chertkov et al.~\cite{Chertkov2010} studied the case where $P=|\mathcal{N}(0,\kappa)|$ is a folded Gaussian and $Q$ is the uniform distribution over $[0,n]$. When $\kappa=O(1)$, the planted edges are competitive with the lightest un-planted edges at each vertex, which have expected weight $1$. This suggests a phase transition in this regime, and indeed they predicted a transition from almost perfect recovery to partial recovery at $\kappa \approx 0.17$ using the cavity method of statistical physics.

We focus on exponential weight distributions, $P=\exp(\lambda)$ and $Q=\exp(1/n)$, so that the planted and un-planted weights have expectation $1/\lambda$ and $n$ respectively. For this family of distributions we obtain exact results, proving a transition from almost perfect recovery to partial recovery at $\lambda = 4$, and determining the expected overlap between $\Mplanted$ and $\Mmin$ for $\lambda < 4$.

Many of our results apply more generally for any distribution of un-planted edge weights with density $Q'(0) = 1/n$, such as when $Q$ is uniform in the interval $[0,n]$. However, our assumption that the planted weights $P$ are exponentially distributed is important for two reasons. First, it makes possible to exactly analyze a message-passing algorithm, and obtain precise results for the expected overlap. Secondly, it has the pleasing consequence of making $\Mmin$ the maximum-likelihood estimator for $\Mplanted$. To see this, note that all $n!$ matchings are equally likely a priori. 
Let $G$ denote the observed complete bipartite graph with edge weights $W$. 
The posterior probability for a given matching $M'$, i.e., $\prob[\Mplanted=M' \mid G]$, is proportional to the density
\begin{equation}
\begin{aligned}
\prob[G \mid M'] 
= \prod_{e \in M'} P(w_e) \prod_{e \notin M'} Q(w_e) 
&\propto \prod_{e \in M'} \exp\!\left( -(\lambda-1/n) w_e \right)\\ 
&= \exp\!\left( -(\lambda-1/n) \sum_{e \in M'} w_e \right) \, \label{eq:density_planted}.
\end{aligned}
\end{equation}
Thus maximizing the likelihood is equivalent to minimizing the total weight of $M'$.

Our main results are as follows. 
\begin{itemize} 
\item In Theorem \ref{thm:first_moment_bound}, we show that the minimum matching $\Mmin$ achieves almost perfect recovery with high probability whenever $\lambda \ge 4$. This proof is a simple first-moment argument using the expected number of augmenting cycles of each length. 
\item In Theorem \ref{thm:main}, we compute the expected overlap between $\Mplanted$ and $\Mmin$ for $\lambda < 4$, showing that it is an explicit function $\alpha(\lambda)$ given by a system of differential equations. 
\end{itemize}

The proof of Theorem \ref{thm:main} takes up most of the paper.  Our proof is inspired by Aldous' analysis of the minimum matching in the un-planted case where all edges have the same weight distribution with $Q'(0)=1/n$. Using the machinery of local weak convergence~\cite{Aldous1992,Aldous2001,Aldous2004} Aldous gave a rigorous justification for the cavity method of statistical physics~\cite{Parisi1987}, modeling $K_{n,n}$ as a Poisson-weighted infinite tree (PWIT). The cost of matching a vertex with one of its children then follows a probability distribution which is the fixed point of a recursive distributional equation (RDE) which can then be transformed into an ordinary differential equation (ODE). Solving this ODE proves the conjecture of M\'ezard and Parisi~\cite{Parisi1987} that the expected cost per vertex is $\zeta(2)=\pi^2/6$. 

Generalizing Aldous' analysis to the planted case presents several challenges. We now have an infinite weighted tree we call the \emph{planted PWIT} with two types of edges and two types of vertices, since the partner of a vertex in $\Mplanted$ can be its parent or one of its children. The cost of matching a vertex with a child follows a pair of probability distributions fixed by a system of RDEs, which (when $P$ is exponential) we can transform into a system of four coupled ODEs. We use techniques from dynamical systems to show that this system has a unique solution consistent with its boundary conditions, and express the expected overlap $\alpha(\lambda)$ as an integral involving this solution. 

While we focus on the case where $P$ is exponential, we claim that a qualitatively similar picture to Theorems \ref{thm:first_moment_bound} and Theorem \ref{thm:main} holds for other distributions of planted weights. Indeed, much of our proof applies to any distribution $P$, including the general framework of a message-passing algorithm on the planted PWIT, and the resulting system of RDEs. Thus while the location of the threshold and the overlap would change, in any one-parameter family of distributions $P$ we expect there to be a phase transition from almost-perfect to partial recovery when $P$’s expectation crosses some critical value.

%% file: Sections/Almostperfectrec.tex
We start by proving that the minimum matching achieves almost perfect recovery whenever $\lambda \ge 4$. 

\begin{theorem}
\label{thm:first_moment_bound} 
For any $\lambda \ge 4$, we have $\expect[\overlap(\Mplanted,\Mmin)]= 1-o(1)$. In particular, $\expect[|\Mplanted \symdiff \Mmin|]$ is $O(1)$ for $\lambda > 4$ and $O(\sqrt{n})$ for $\lambda=4$.
\end{theorem}
\noindent 
To prove Theorem \ref{thm:first_moment_bound}, we use the following Chernoff-like bound on the probability that one Erlang random variable exceeds another. The proof is elementary and appears in Appendix \ref{app:erlang-diff}.

\begin{lemma}
\label{lmm:erlang-diff}
Suppose $X_1$ is the sum of $t$ independent exponential random variables with rate $\lambda_1$, 
and $X_2$ is the sum of $t$ independent exponential random variables with rate $\lambda_2$ (and independent of $X_1$) where $\lambda_1 > \lambda_2$. Then
\[
\prob[X_1 > X_2] 
\le \left( \frac{4 \lambda_1 \lambda_2}{(\lambda_1+\lambda_2)^2} \right)^{\!t} 
\le \left( \frac{4 \lambda_2}{\lambda_1} \right)^{\!t} .
\] 
\end{lemma}
\begin{proof}[Proof of Theorem \ref{thm:first_moment_bound}]
 An \emph{alternating cycle} is a cycle in $K_{n,n}$ that alternates between planted and un-planted edges, and an \emph{augmenting cycle} is an alternating cycle $C$ where the total weight of its planted edges $C \cap \Mplanted$ exceeds that of its un-planted edges $C \setminus \Mplanted$. 

Now recall that the symmetric difference $\Mplanted \symdiff \Mmin$ is a disjoint union of augmenting cycles. The number of cyclic permutations of $t$ things is $(t-1)!$. Thus the number of alternating cycles of length $2t$, i.e., containing $t$ planted edges and $t$ un-planted edges, is at most 
\begin{equation}
{n \choose t} (t-1)! 
= \frac{1}{t} \,n^{t} \left( 1-\frac{1}{n} \right) 
\left( 1-\frac{2}{n} \right) \cdots \left( 1-\frac{t-1}{n} \right)
\le \frac{1}{t} \,n^{t} \,\e^{-t (t-1) / (2n)} \, .
\label{eq:Celll-count}
\end{equation}
Applying Lemma \ref{lmm:erlang-diff} with $\lambda_1=\lambda$ and $\lambda_2=1/n$, the probability that a given alternating cycle of length $2t$ is augmenting is at most $(4/(\lambda n))^t$. 

Now the size of the symmetric difference $|\Mplanted \symdiff \Mmin|$ is at most the total length of all augmenting cycles. By the linearity of expectation, its expectation is bounded by
\[
\expect[|\Mplanted \symdiff \Mmin|]
\le \sum_{t=1}^n 2t \left( \frac{4}{\lambda n} \right)^{\!t} \frac{1}{t} \,n^{t} \,\e^{-t (t-1) / (2n)} 
\le 2 \e^{1/2} \sum_{t=1}^\infty \left( \frac{4}{\lambda} \right)^{\!t} \e^{-t^2 / (2n)} \, . 
\]
When $\lambda > 4$ the geometric sum $\sum_{t=1}^\infty (4/\lambda)^t$ converges, giving $\expect[|\Mplanted \symdiff \Mmin|] = O(1)$. When $\lambda = 4$, we have $\sum_{t=1}^\infty \e^{-t^2/(2n)} \le \int_0^\infty \e^{-t^2/(2n)} \mathrm{d}t = \sqrt{\pi n / 2}$, so $\expect[|\Mplanted \symdiff \Mmin|] = O(\sqrt{n})$. 

To complete the proof, let $\omega(1)$ be any function of $n$ that tends to infinity. By Markov's inequality, with high probability $|\Mplanted \symdiff \Mmin|$ is less than $\omega(1)$ times its expectation, and \eqref{eq:overlap-def} gives w.h.p. $\overlap(\Mplanted,\Mmin) 
= 1-o(1)$.
\end{proof}

We note that when $\lambda > 4$ is sufficiently large we have $\expect [|\Mplanted \symdiff \Mmin|] < 1$, implying that $\Mmin$ achieves perfect recovery, i.e., $\Mmin=\Mplanted$, with positive probability. We also note that a similar argument shows that, for $\lambda < 4$, the overlap is w.h.p. at least $1 - 2 \log \frac{4}{\lambda}$. But this bound is far from tight, and below we give much more precise results.

%% file: Sections/ExactResult.tex
In this section we provide a characterization of the asymptotic overlap of $\Mopt$, showing exactly how well $\Mopt$ achieves partial recovery when $\lambda < 4$. 

\begin{theorem}
\label{thm:main}
Suppose $0 < \lambda < 4$ is a fixed constant. Then the expected overlap between the minimum
matching and the planted one is
\[
 \lim_{n\to \infty} \frac{1}{n} \,\expect[\left| \Mmin \cap \Mplanted \right|] 
 = \alpha(\lambda) \, ,
\]
where 
\begin{equation}
\alpha(\lambda) =1 - 2 \int_0^\infty \left(1-F(x) \right) \left(1-G(x) \right) V(x) W(x) \,dx  < 1 \, , 
\label{eq:overlap_asym_limit}
\end{equation}
and where $(F, G, V, W)$ is the unique solution to the coupled system of ordinary differential equations \eqref{eq:ODE1}--\eqref{eq:ODE4} given below with boundary conditions~\eqref{eq:ODE_boundary}--\eqref{eq:ODE_regularity}.
\end{theorem}

Denote the weight of the minimum matching by $w(\Mopt) \triangleq \sum_{e\in\Mmin}w_e$,
where $w_e$ is the weight of the edge $e$.
We also derive the asymptotic value of $(1/n)\expect[w(\Mopt)]$ for $0<\lambda < 4$.
\begin{corollary}\label{cor:weight}
For $0<\lambda<4$, the weight of the minimum matching is
\begin{align*}
	&\lim_{n\to\infty} \frac{1}{n} \, \expect\big[w(\Mopt) \big] 
	= \beta_p(\lambda) + \beta_u(\alpha)
\end{align*}
where 
$\beta_p(\lambda)$ and $\beta_u(\lambda)$ are the contributions of planted and un-planted edges to the weight of $\Mmin$ respectively:
\begin{align*}
	&\beta_p(\lambda) = \int_{-\infty}^\infty \left(1-F(x)\right) \left(1-G(x) \right) V(x)\left[\int_{0}^{\infty}\lambda t\,  \e^{-\lambda t} \left(1-F(t-x)\right)\,\dt \right]\dx \\
	&\beta_u(\lambda) = \int_{-\infty}^\infty \left(1-F(y) \right) \big[ \left(1-G(y) \right) V(y)W(y)-\lambda (G(y) - W(y)) \big] \\
	&\qquad\qquad\times \left[\int_{0}^{\infty} t (1-F(t-y))W(t-y)\,\dt\right]\dy.
\end{align*}
\end{corollary}

\begin{proof}
We start by relating the planted model $(\ell_n,K_{n,n})$ where $\ell_n$ denotes the random edge weights, to a type of weighted infinite tree $(\ell_\infty, T_\infty)$ as Aldous did for the un-planted model~\cite{Aldous1992,Aldous2001}. This tree corresponds to the neighborhood of a uniformly random vertex, where ``local'' is defined in terms of shortest path length (sum of edge weights). While $K_{n,n}$ has plenty of short loops, this neighborhood is locally treelike since it is unlikely to have any short loops consisting entirely of low-weight edges. 

Starting at a root vertex $\root$, we define the tree $T_\infty$ shown in Figure \ref{fig:asympobject}. The root has a planted child, i.e., a child connected to it by a planted edge (bold in red), and a series of un-planted children (solid blue). We label these vertices with strings of integers as follows: the root is labeled with the empty string $\root$. Appending $0$ to a label indicates the planted child of that parent, if it has one---that is, if its partner in the planted matching is a child rather than its parent. We indicate the un-planted children by appending $i$ for $i \in \{1, 2, 3, \ldots\}$. 

\begin{figure} 
\begin{center}
	\includegraphics[scale=0.5]{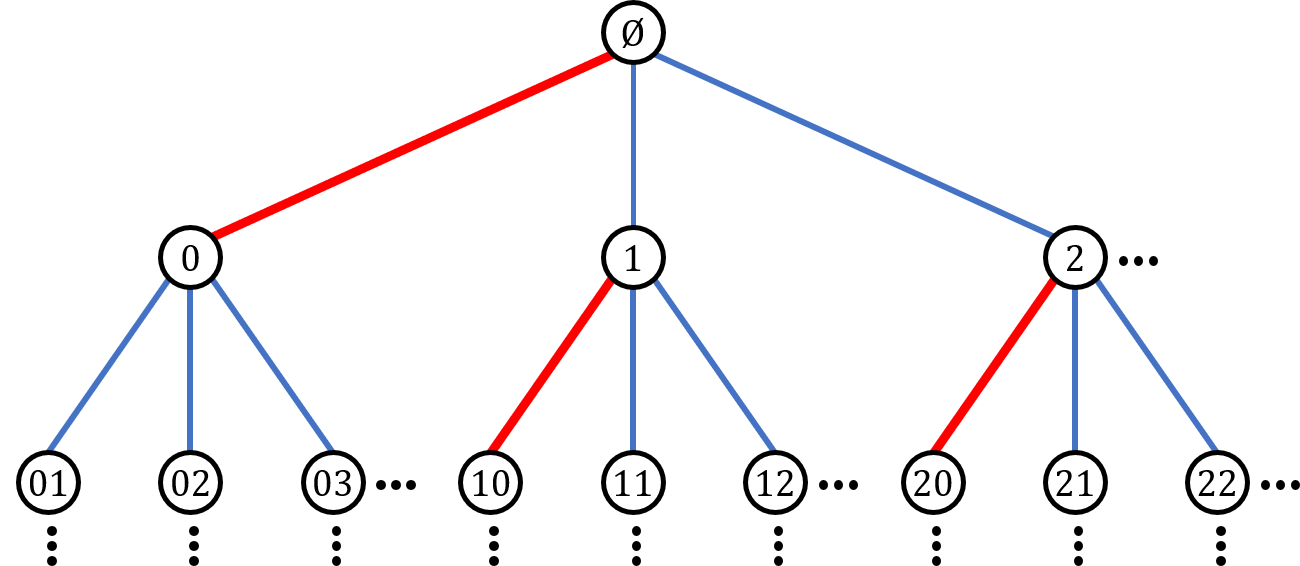}
\end{center}
\caption{The planted Poisson weighted infinite tree (planted PWIT) $(\ell_\infty, T_\infty)$ with the labeling scheme described in the text. The bold red edges are planted edges and the solid blue edges are un-planted. The root is the empty string $\root$. Appending $0$ to the label of a vertex indicates its planted child, if any, while appending $i \ge 1$ indicates its un-planted child with the $i$th lightest edge.
\label{fig:asympobject}}
\end{figure}

We sort the un-planted children of each vertex so that the one labeled with $i$ is the $i$th lightest, i.e., has the $i$th lightest edge. Since the distribution of un-planted weights has density $Q'(0)=1/n$ at $0$, these weights are asymptotically described by the arrivals of a Poisson process with rate $1$, while the weight of the planted edges are distributed as $\exp(\lambda)$. We call the resulting structure the planted Poisson weighted infinite tree, or planted PWIT, and use $\ell_\infty$ to denote its edge weights. We define all this formally in Section \ref{app:gen_setup} and Section \ref{app:planted-PWIT}, and prove that the finite planted model $(\ell_n, K_{n,n})$ weakly converges to $(\ell_\infty, T_\infty)$. 

Following Aldous~\cite{Aldous2001}, in Section \ref{app:optmatch} we then construct a matching $\Minfopt$ on the planted PWIT. Crucially, it has a symmetry property called \emph{involution invariance}, which roughly speaking means that it treats the root just like any other vertex in the tree. We prove that it is the unique involution invariant matching that minimizes the expected cost at the root. 

We define $\Minfopt$ in terms of the fixed point of a message-passing algorithm that computes, for each vertex $v$, the cost of matching $v$ with its best possible child. This cost is the minimum over $v$'s children $w$ of the weight of the edge between them, minus the analogous cost for $w$: 
\begin{align*}
X_v 
&= \min_{\text{children $w$ of $v$}} ( \ell_\infty(v,w) - X_w ) \, . 
\end{align*}
Now suppose that the $X_w$'s are independent, and our goal is to compute the distribution of $X_v$. Unlike the un-planted model, the two types of children will have their $X_w$ drawn from two different distributions. In the first case, $w$ is $v$'s planted child, and $w$'s children are all un-planted. In the second case, $w$ is an un-planted child of $v$, and has a planted child of its own. Let $X$ and $Y$ denote the distributions of $X_w$ in these two cases. Then assuming that $X_v$ obeys the appropriate distribution gives the following system of recursive distributional equations (RDE)s:
\begin{align}
X &\overset{d}{=} \min \{ \zeta_i - Y_i\}_{i=1}^\infty \label{eq:RDE2} \\ 
Y &\overset{d}{=} \min (\eta - X, \{\zeta_i - Y_i\}_{i=1}^\infty)
\overset{d}{=} \min (\eta - X, X' ) \, , \label{eq:RDE1}
\end{align}
where the $Y_i$'s are  i.i.d.,  $X$ and $X'$ are  i.i.d.,  $\eta \sim \exp(\lambda)$, and the $\zeta_i$ for $i=1,2,\ldots$ are jointly distributed as the arrivals of a Poisson process of rate 1. 

In general, analyzing recursive distributional equations (RDEs) is very challenging, since they act on the infinite-dimensional space of probability distributions over the reals. However, it is sometimes possible to ``collapse'' them into a finite-dimensional system of ordinary differential equations. For the un-planted case of the random matching problem, Aldous~\cite{Aldous2001} derived a single differential equation whose solution is the logistic distribution. 
For the planted case, we use a similar approach, but arrive at a more complicated system of four coupled ODEs. 

\begin{lemma}
\label{lem:rde-ode}
Let $f_X, f_Y, F_X(x) = \prob[X < x]$ and $F_Y(y) = \prob[Y < y]$ denote the probability density functions and cumulative distribution functions (CDFs) of $X$ and $Y$, and let $\bar{F}_X = 1-F_X$ and $\bar{F}_Y = 1-F_Y$. If \eqref{eq:RDE2}--\eqref{eq:RDE1} have a solution, then
\begin{equation}
\label{eq:RDE_ODE_1}
\frac{\dF_X(x)}{\dx} = \bar{F}_X(x) \bar{F}_X(-x) \,\expect[F_X(\eta+x)] .
\end{equation}
\end{lemma}

\begin{proof}
First note that \eqref{eq:RDE1} gives
\begin{equation}
\label{eq:fbary}
\bar{F}_Y(y) = 
\bar{F}_X(y) \,\expect[F_X(\eta-y)] \, . 
\end{equation}
Now the pairs $\{(\zeta_i,Y_i)\}$ in \eqref{eq:RDE2} form a two-dimensional Poisson point process $\{(z,y)\}$ on $\mathbb{R}_+ \times \mathbb{R}$ with density $f_Y(y) \,\dz \,\dy$. We have $X > x$ if and only if none of these points have $z-y < x$, so
\begin{equation}
\label{eq:fbarx}
\bar{F}_X(x) 
= \exp\!\left( - \iint_{z-y < x} f_Y(y) \,\dz \,\dy \right) 
= \exp\!\left( - \int_{z=-x}^\infty \bar{F}_Y(z) \,\dz \right) \, .
\end{equation}
Taking derivatives of both sides of this equation with respect to $x$ and using~\eqref{eq:fbary} gives
\begin{align*}
f_X(x) 
= \frac{\dF_X(x)}{\dx} 
= -\frac{\mathrm{d} \bar{F}_X(x)}{\dx} 
= \bar{F}_X(x) \bar{F}_Y(-x) 
= \bar{F}_X(x) \bar{F}_X(-x) \,\expect[F_X(\eta+x)] \, .
\end{align*}
\end{proof}

For the sake of simplicity, we omit the subscript $X$ in $F_X(\cdot)$ in the sequel. Define
\begin{equation}
\label{eq:GVW_def}
G(x) = F(-x), \quad V(x) = \expect[F(\eta+x)], \quad W(x) = V(-x). 
\end{equation}

\begin{lemma}\label{lmm:RDE_ODE_exp}
When $\eta \sim \exp(\lambda)$, $F$ is a solution to \eqref{eq:RDE_ODE_1} if and only if $(F,G,V,W)$ 
is a solution to the following four-dimensional system of ordinary differential equations (ODEs):
\begin{align}
&\frac{\dF}{\dx} = (1-F(x))(1-G(x))V(x)\allowdisplaybreaks \label{eq:ODE1} \\ 
&\frac{\dV}{\dx} = \lambda(V(x)-F(x)) \allowdisplaybreaks \label{eq:ODE2} \\ 
&\frac{\dG}{\dx} = -(1-F(x))(1-G(x))W(x)\allowdisplaybreaks \label{eq:ODE3} \\
&\frac{\dW}{\dx} = \lambda(G(x)-W(x))\allowdisplaybreaks \label{eq:ODE4}
\end{align}
with the boundary conditions
\begin{equation}\label{eq:ODE_boundary}
\begin{aligned}
&F(-\infty) = V(-\infty) = G(+\infty) = W(+\infty ) = 0 \\
&F(+\infty) = V(+\infty) = G(-\infty) = W(-\infty ) = 1.
\end{aligned}
\end{equation}
and 
\begin{align}
0 \le F ,G \le 1, \quad 0 < V, W \le 1 \, . \label{eq:ODE_regularity}
\end{align}
\end{lemma}

\begin{proof}
For one direction, suppose $F$ is a solution to \eqref{eq:RDE_ODE_1}. Then \eqref{eq:ODE1} and \eqref{eq:ODE3} directly follow from \eqref{eq:RDE_ODE_1} by plugging in the definition of $(F,G,V,W)$; thus they hold for any distribution of $\eta$. In contrast, \eqref{eq:ODE2} and \eqref{eq:ODE4} are derived via integration by parts under the assumption that $\eta\sim \exp(\lambda)$. The conditions \eqref{eq:ODE_boundary} and \eqref{eq:ODE_regularity} hold because $F$ must be a valid CDF. Note that $V(x),W(x)>0$ for any finite $x$ by definition, as  $\eta$ is larger than any fixed threshold with a positive probability. 

For the other direction, suppose $F$ is a solution to the system of ODEs \eqref{eq:ODE1}--\eqref{eq:ODE4} with conditions \eqref{eq:ODE_boundary}--\eqref{eq:ODE_regularity}.  Clearly $F$ satisfies \eqref{eq:RDE_ODE_1}. We only need to verify that $F$ is a valid CDF, which is equivalent to checking (1) $F$ is non-decreasing; (2) $F(+\infty)=1$ and $F(-\infty)=0$; and (3) $F$ is right continuous. All these properties  are satisfied automatically. 
\end{proof}

We comment that RDEs can be solved exactly for some other problems with random vertex or edge weights in the case of the exponential distribution, such as maximum weight independence sets and maximum weight matching in sparse random graphs~\cite{Gamarnik2004,Gamarnik2006,Gamarnik2008}. In some cases this is simply because the minimum of a set of exponential random variables is itself an exponential random variable. To our knowledge our situation involving integration by parts is more unusual.

An interesting consequence of \eqref{eq:ODE1}--\eqref{eq:ODE_regularity} is the following conservation law: 
\begin{equation}
F(x)W(x) + G(x)V(x) - V(x)W(x) = 0 \, . \label{eq:conservation_law}
\end{equation}
Since $F(0)=G(0)$ and $V(0)=W(0)$, this also implies that
\begin{equation}
\label{eq:v0f0}
V(0) = 2F(0) \, .
\end{equation}

Surprisingly, we find that the system~\eqref{eq:ODE1}--\eqref{eq:ODE4} exhibits a sharp phase transition at $\lambda=4$. On the one hand, when $\lambda\ge 4$, they have no solution consistent with~\eqref{eq:ODE_boundary}--\eqref{eq:ODE_regularity}, corresponding to Theorem~\ref{thm:first_moment_bound} that we have almost perfect recovery in that case. To see this, 
assume that $V(x) \neq 0$ and introduce a new function $U(x)$ as
\begin{align*}
 U(x) = \frac{F(x)}{V(x)} \, .
\end{align*}
Then $U(x)$ is differentiable and satisfies:
\begin{align}
\frac{\dU}{\dx} 
= - \lambda U(1-U) + (1-F) (1-G). \label{eq:U_diff}
\end{align}

\begin{lemma}
\label{lmm:ode_no_solution}
If $\lambda \ge 4$, then 
the system of ODEs \eqref{eq:ODE1}--\eqref{eq:ODE4} with conditions \eqref{eq:ODE_boundary}--\eqref{eq:ODE_regularity} has no solution.
\end{lemma}

\begin{proof}
We prove by contradiction. Suppose the system of ODEs \eqref{eq:ODE1}--\eqref{eq:ODE4} has a solution satisfying the conditions \eqref{eq:ODE_boundary}--\eqref{eq:ODE_regularity}. Then $U(x) \to 1$ as $x \to +\infty$. 
Since~\eqref{eq:v0f0} gives $U(0)=F(0)/V(0)=1/2$, this implies that there is some $x_0 \ge 0$ such that $U(x_0)=1/2$ and $U'(x_0) \geq 0$. Since $U(x_0) = F(x_0) / V(x_0)$ and $V(x_0) > 0$, we also have $F(x_0) = V(x_0)/2 > 0$, and we also have $G(x_0) \ge 0$. But then~\eqref{eq:U_diff} gives
\[
U'(x_0) = - \lambda/4 + \left( 1- F(x_0) \right) \left( 1- G(x_0) \right) < -\lambda /4 + 1 \le 0 \, .
\]
This contradicts $U'(x_0) \ge 0$, and shows that $U$ can never exceed its initial value of $1/2$ or tend to $1$ as $x \to +\infty$.
\end{proof}

On the other hand, Theorem \ref{thm:ode_unique_solution} in Section \ref{app:ODE_analysis} proves that for all $\lambda < 4$, there is a unique solution to \eqref{eq:ODE1}--\eqref{eq:ODE4} consistent with the conditions \eqref{eq:ODE_boundary}--\eqref{eq:ODE_regularity}, and hence giving the CDFs of $X$ and $Y$. The idea hinges on a dynamical fact, namely that the $(U,V,W) = (1,1,0)$ is a saddle point, and there is a unique initial condition that approaches it as $x \to \infty$ along its unstable manifold.

Along with Lemma \ref{lmm:RDE_ODE_exp}, this unique solution to the ODEs gives the unique solution to the RDEs \eqref{eq:RDE2} and \eqref{eq:RDE1}. Moreover Theorem \ref{thm:itallworks} in Section \ref{app:jointconv} tells us that the expected overlap of $\Mmin$ converges to that of $\Minfopt$, which in turn is the probability that the edge weight of a planted edge is less than the cost of matching its endpoints to other vertices: 
\[
\lim_{n\to \infty} \frac{1}{n} \,\expect[\left| \Mmin \cap \Mplanted \right|]
= \alpha(\lambda) = \prob[\eta < X + X' ] \, ,
\]
where $X$ and $X'$ are  i.i.d.\ with CDF given by $F$ and $\eta \sim \exp(\lambda)$ is independent. 
Finally, we compute $\prob[ \eta < X+ \hat{X} ]$ as follows:
\begin{align*} 
\prob[\eta < X+ \hat{X}]
&= 1 - \expect_\eta\left[\int_{-\infty}^{+\infty} f(x) F(\eta-x) \,\dx\right] \\
&= 1 - \int_{-\infty}^{+\infty} \frac{\dF(x)}{\dx}  \,\expect_\eta [F(\eta-x)] \,\dx \\
&= 1 - \int_{-\infty}^{+\infty} (1-F(x))(1-G(x)) V(x) W(x) \,\dx \\
&= 1 - 2 \int_{0}^{+\infty} (1-F(x))(1-G(x)) V(x) W(x) \,\dx \, .
\end{align*}
where in the last line we used the fact that the integrand is an even function of $x$. 

Similarly, by Theorem \ref{thm:itallworks}, Lemma \ref{lem:match} and Corollary \ref{cor:iidX}, we have
\begin{align*}
	\lim_{n\to\infty} \frac{1}{n} \,\expect\!\left[ \sum_{e\in\Mmin}\ell_n(e) \right] 
	&= \lim_{n \to \infty}\expect[\ell_n(1,\Mmin(1))] = \expect[\ell_\infty(\root,\Minfopt(\root))] \\
	&= \expect[\eta \,\one(\eta < X+X')] + \expect[\zeta \,\one(\zeta < Y+Y')],
	\end{align*}
where $X$ and $X'$ are  i.i.d.\ with CDF given by $F$, $Y$ and $Y'$ are  i.i.d.\ with CDF given by $1-(1-F)W$, $\eta \sim \exp(\lambda)$ is independent of $X$ and $X'$. Finally $\zeta$ is ``uniformly distributed'' over $\mathbb{R}_+$, i.e., with Lebesgue measure on $\mathbb{R}_+$, and is independent of $Y$ and $Y'$. Note that
\begin{align*}
	\beta_p(\lambda) &\triangleq \expect[\eta \,\one(\eta < X+X')] = \int_{0}^{\infty} \lambda t\, \e^{-\lambda t }\int_{-\infty}^{\infty}\frac{\dF(x)}{dx}(1-F(t-x))\,\dx\,\dt\\
	&=\int_{-\infty}^\infty \left(1-F(x) \right) \left(1-G(x) \right) V(x)\int_{0}^{\infty}\lambda t \, \e^{-\lambda t} \left(1-F(t-x)\right)\,\dt \,\dx,
\end{align*}
and
\begin{align*}
	\beta_u(\lambda) &\triangleq \expect[\zeta \,\one(\zeta < Y+Y')] \\
	&=  \int_{0}^{\infty} t  \int_{-\infty}^{\infty} \frac{\dF_Y(y)}{\dy} \left( 1- F_Y(t-y) \right) \,\dy \,\dt\\
	&=  \int_{0}^{\infty} t  \int_{-\infty}^{\infty} \left(\frac{\dF(y)}{\dy}W(y)-(1-F(y))\frac{\dW(y)}{\dy}\right) (1-F(t-y))W(t-y)\,\dy\,\dt\\
	&=\int_{-\infty}^\infty \left(1-F(y) \right) \left[ \left(1-G(y) \right) V(y)W(y)-\lambda (G(y) - W(y)) \right]\\
	&\qquad\qquad\times \left[\int_{0}^{\infty} t (1-F(t-y))W(t-y)\,\dt\right]\dy.
\end{align*}
This completes the proof of Theorem \ref{thm:main} and Corollary \ref{cor:weight}. 
\end{proof}

To illustrate our results, we plot the function $\alpha(\lambda)$ and $\beta_p(\lambda) + \beta_u(\lambda)$ in Figure \ref{fig:overlap_experiment}, and compare with experimental results from finite graphs with $n=1000$.
As the plot shows, when $\lambda \ge 4$ and the overlap is $1-o(1)$ w.h.p., $\beta_p(\lambda) + \beta_u(\lambda) =1/\lambda$, i.e., the expected weight of a planted edge. Conversely when $\lambda \to 0$ and the overlap is $o(1)$ w.h.p., $\beta_p(\lambda) + \beta_u(\lambda) \to \zeta(2)$ as in the un-planted model, since $\Mmin$ consists almost entirely of un-planted edges.

We comment that the connection between the finite planted model and the planted PWIT is an integral part of the above argument. We explore this connection in detail in Sections \ref{app:gen_setup}--\ref{app:jointconv}. The results presented in these sections are true for any distribution of un-planted edge weights with density $Q'(0) = 1/n$, and any distribution of planted edge weights $P$.

%% file: Sections/Discussion.tex
We conclude the discussion of our main results with some open questions.
\begin{enumerate}

\item We have computed the expected overlap and cost (per vertex) of the minimum matching. However, in Figure~\ref{fig:overlap_experiment} both quantities appear to concentrate on their expected values, and we conjecture that they both have variance $O(1/n)$. We give numerical evidence for this in Figure~\ref{fig:openproblem}, where we plot empirical values of these variances for fixed $\lambda=1$ and varying $n$. For the unplanted model with exponentially distributed edge weights, W{\"a}stlund~\cite{Wstlund2005,Wstlund2010} derived the precise asymptotic formula $(4\zeta(2)-4\zeta(3))/n + O(1/n^2)$ for $\mathrm{Var}\!\left[ \frac{1}{n}w(\Mmin) \right]$; it is not known if this holds for e.g., the uniform distribution. In the planted case, we can adapt the arguments in~\cite[Section 10]{talagrand1995concentration} from the unplanted model to show that
\[
\mathrm{Var}\!\left[ \frac{1}{n}w(\Mmin) \right] = O\!\left( \frac{(\log n)^4}{ N (\log \log n)^2 } \right) 
\]
via Talagrand's inequality.\footnote{
The only minor adaptation is to show \cite[Prop.10.3]{talagrand1995concentration} continues to hold 
for the planted model, that is, the bipartite graph $K_{n,n}$ after truncating edges with weights above 
$2u \log n$ is $(u\log n)$-expanding with high probability. This can be proved by first coupling the planted model to  the unplanted one in a way that every edge weight in the planted model is no larger and then directly invoking Prop.10.3.}
We do not know how to tighten this to $O(1/n)$, or to prove a similar bound for the overlap. We believe both would follow from correlation decay of messages in the planted PWIT, in which case distant pairs of edges in $K_{n,n}$ would be asymptotically independent.

\begin{figure}
\centering
\begin{subfigure}[!t]{.45\linewidth}
	\begin{center}
		\includegraphics[width=\linewidth]{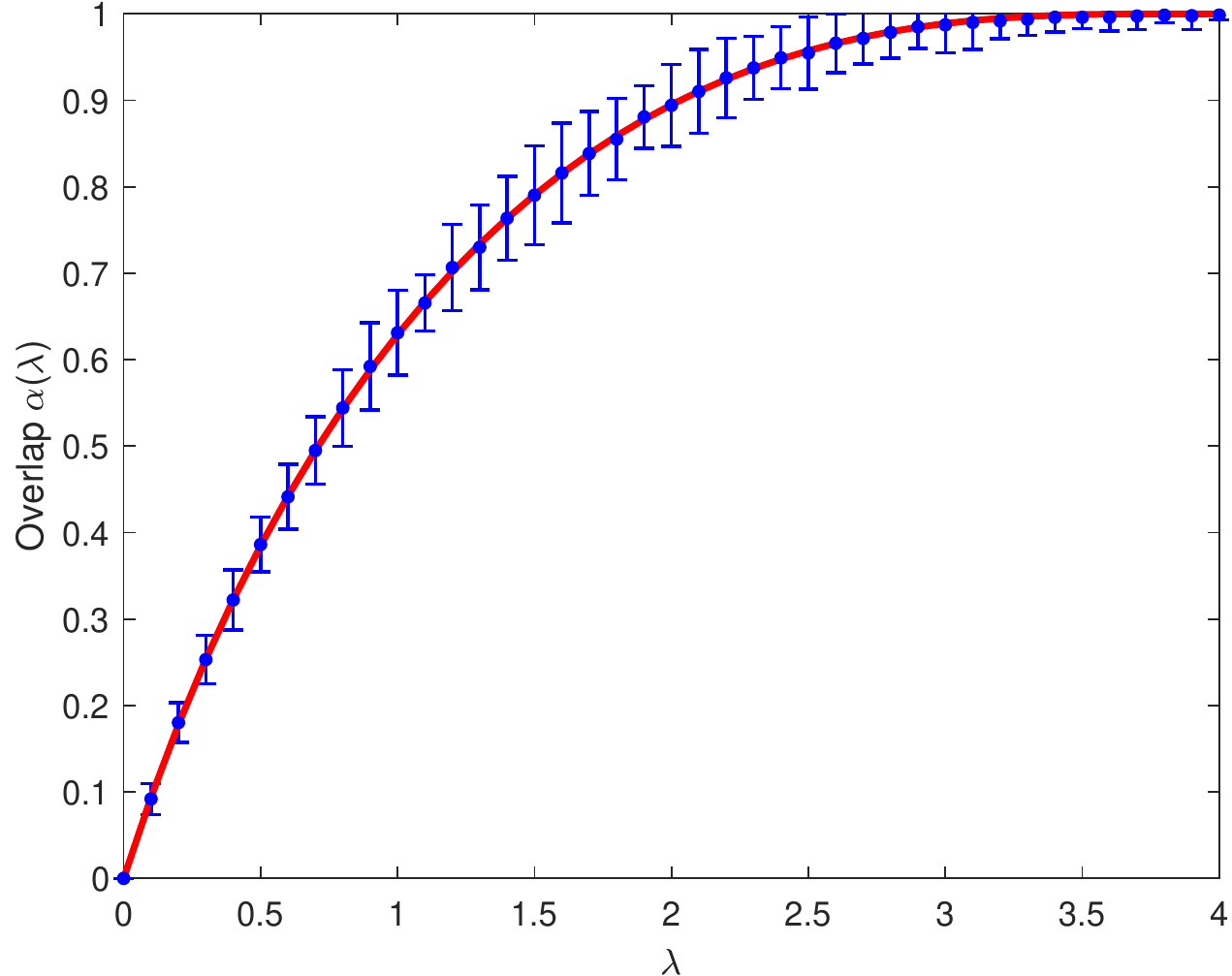}
		\caption*{Overlap of $\Mmin$ and $\Mplanted$.}
	\end{center}
\end{subfigure}
\hfill
\begin{subfigure}[!t]{.45\linewidth}
	\begin{center}
		\includegraphics[width=\linewidth]{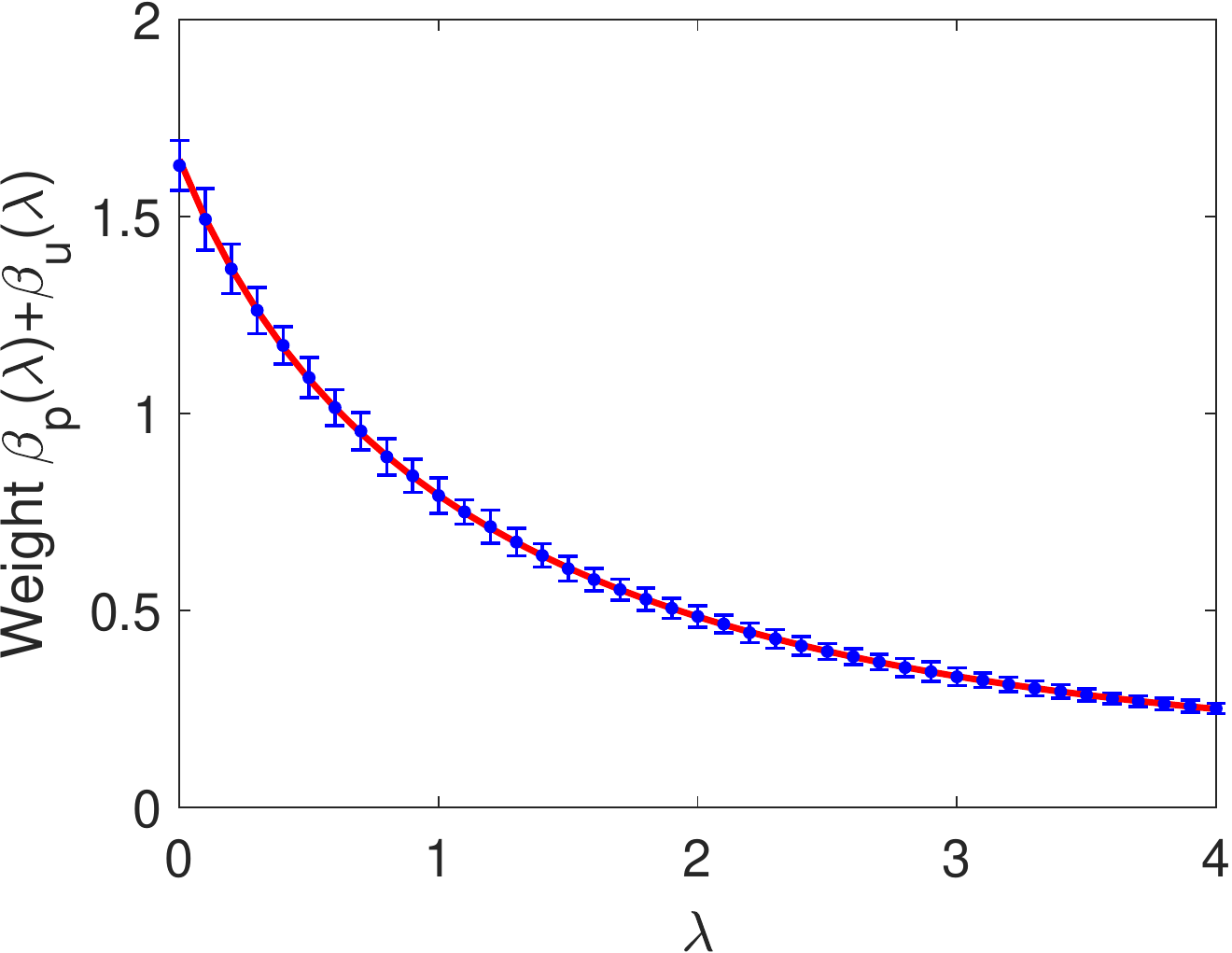}
		\caption*{Weight of $\Mmin$ divided by $n$.}
	\end{center}
\end{subfigure}
\caption{The solid red line is the theoretical value computed by numerically solving the system of ODEs \eqref{eq:ODE1}--\eqref{eq:ODE4}. The blue dots, is the empirical mean of the corresponding quantity on bipartite graphs generated by the planted model with $n=1000$. Each dot is the average of $100$ independent trials. The error bars show the 95\% confidence interval of the distribution, i.e., the range into which 95 out of 100 trials fell, suggesting that both quantities are concentrated around their expectations.}
\label{fig:overlap_experiment}
\end{figure}

\begin{figure}
	\centering
	\begin{subfigure}[!t]{.47\linewidth}
		\begin{center}
			\includegraphics[width=\linewidth]{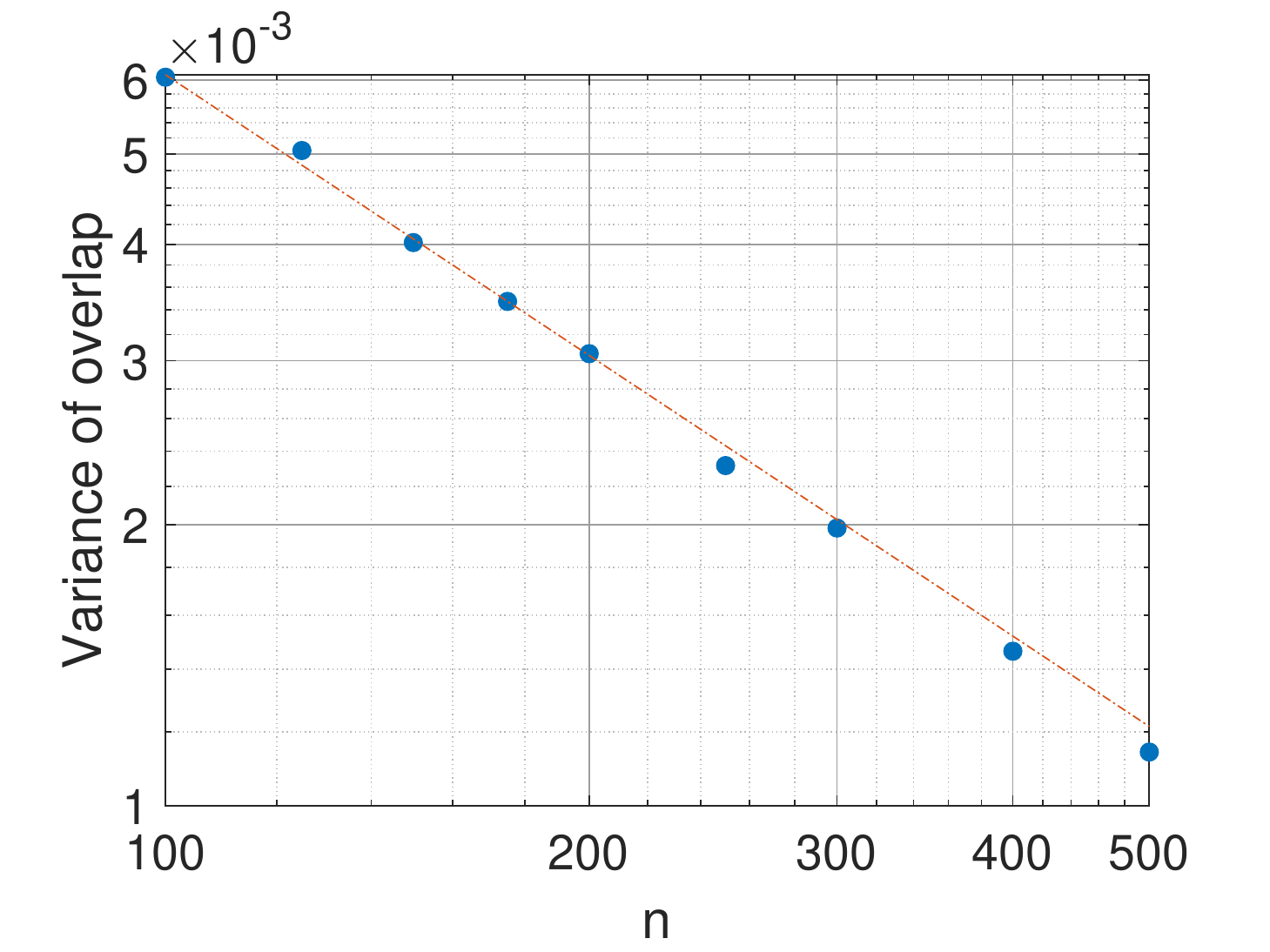}
			\caption*{$\mathrm{Var}\!\left[ \frac{1}{n} \,\left| \Mmin \cap \Mplanted \right| \right]$ for different values of $n$.}
		\end{center}
	\end{subfigure}
	\hfill
	\begin{subfigure}[!t]{.47\linewidth}
		\begin{center}
			\includegraphics[width=\linewidth]{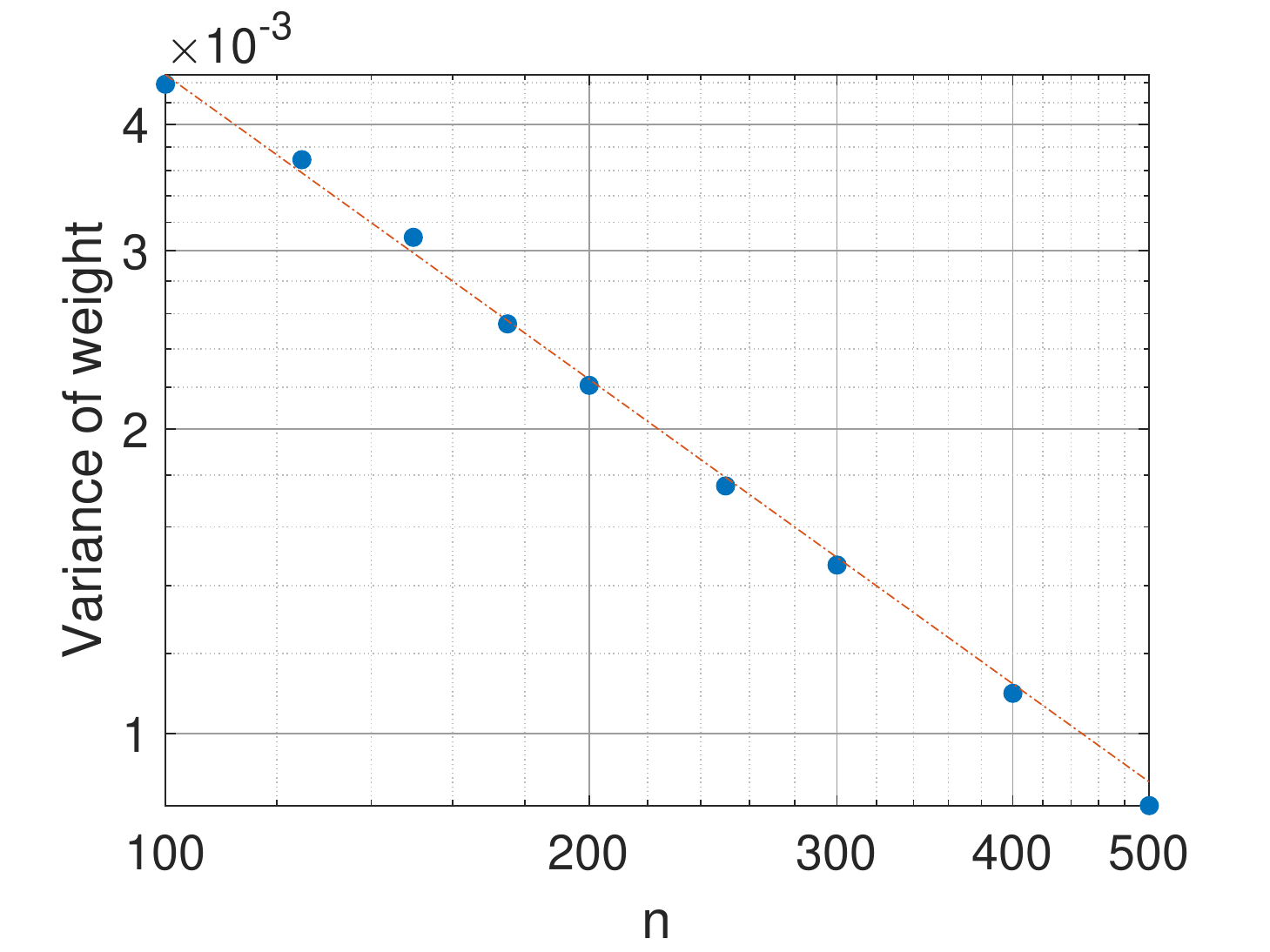}
			\caption*{$\mathrm{Var}\!\left[ \frac{1}{n} \,w(\Mmin) \right]$ for different values of $n$.}
		\end{center}
	\end{subfigure}
	\caption{The blue dots are computed by finding the variance of the corresponding quantity on bipartite graphs generated by the planted model with $\lambda = 1$. Each dot is the average of $1000$ independent trials. For both the overlap and the cost the variance appears to decrease as $O(1/n)$ as shown by the dashed lines with slope $-1$ on this log-log plot.}
	\label{fig:openproblem}
\end{figure}

\item We have focused here on the maximum-likelihood estimator, which for the exponential distribution is also the minimum-weight matching. In physical terms, we have considered this problem at zero temperature. In contrast, the posterior distribution $\prob[\Mplanted = M^\prime \mid G]$ given in \eqref{eq:density_planted} is a Gibbs distribution at nonzero temperature. The estimator with the largest expected overlap would then be the maximum marginal estimator, i.e., the set of edges $e$ for which $\prob[ e \in \Mplanted \mid G ] \ge \prob[ e \notin \Mplanted \mid G]$. This estimator is not generally a matching or even of size $n$; nevertheless, one can restrict to estimators which are perfect matchings while increasing the expected misclassification rate $| M^\prime \triangle \Mplanted |/(2n)$ at most by a factor of two. We leave for future work the problem of computing the expected overlap of this estimator. 
It is possible that it achieves almost-perfect recovery for some $\lambda < 4$, i.e., that the information-theoretic threshold for almost-perfect recovery is different from the threshold we have computed here, but we conjecture this is not the case.
	
\item In physics, a phase transition is called \textit{continuous} if the order parameter (in this case, the overlap) is continuous at the threshold, and as \emph{$p$th order} if its $(p-2)$th derivative is continuous. Although we have not proved this, $\alpha(\lambda)$ in Figure \ref{fig:overlap_experiment} appears to have zero derivative at $\lambda=4$. \label{question:1}
This suggests that the transition in the optimal overlap is of third or higher order, unlike other well-known problems in random graphs such as the emergence of the giant component (second order) \cite{Bollobas2001}, the stochastic block model with two groups (second order) or with four or more groups (first order)~\cite{Moore2017}, or the appearance of the $k$-core for $k \ge 3$ (first order)~\cite{Pittel1996}. Very recently, a non-rigorous argument that neglects small terms in the RDE~\cite{Semerjian2020} was given that suggests the transition is infinite order, i.e., with all derivatives continuous at $\lambda=4$. Proving this is an attractive open problem.

\item A related question is how the minimum matching changes when the graph undergoes a small perturbation. Aldous and Percus~\cite{Aldous2003} introduced this problem formally and classified combinatorial optimization problems based on how the cost of the optimal solution scales with the size of the perturbation. Using a cavity-based analysis and Monto Carlo simulation, they suggested that the minimum cost among all perfect matchings that differ from the minimum matching by at least $\delta n$ edges is larger than the cost of the minimum matching by $\Theta(\delta^3)$.
This framework has been studied rigorously in \cite{Aldous2008} and \cite{Aldous2009} for different combinatorial optimization problems. It would be interesting to explore this same kind of perturbation in the planted model, where we 
study the minimum cost among all matchings that differ from the planted matching by at least $\delta n $ edges.


\item 
Can Lemma~\ref{lmm:ode_no_solution} be turned into a proof of almost perfect recovery? More generally, when the RDEs~\eqref{eq:RDE2} and~\eqref{eq:RDE1} lack a solution supported on $\mathbb{R}$ (i.e., excluding weights in $\{\pm\infty\}$) does this imply almost perfect recovery?

\item We have given two proofs that that the RDEs have a unique solution if $\lambda < 4$. Theorem \ref{thm:ode_unique_solution} uses the dynamics of the ODEs, while Theorem \ref{thm:unique-from-pwit} uses the uniqueness of $\Minfopt$ on the planted PWIT. These two types of reasoning seem completely orthogonal, but they must be connected. When do the properties of the optimum involution invariant object on an appropriate type of infinite tree imply the dynamical fact that a system of RDEs has a unique fixed point?

\item What can we say about distributions $P(w)$ of planted weights other than exponential? For what distributions is it possible to collapse the RDEs into a finite-dimensional system of ODEs? As stated above, Chertkov et al.~\cite{Chertkov2010} studied the folded Gaussian distribution $P=|\mathcal{N}(0,\kappa)|$, but we have been unable to reduce the RDEs to ODEs in this case. 
Nevertheless, in the spirit of universality classes in physics, we expect any reasonable family of distributions $P$ to undergo a phase transition similar to what we have shown here for the exponential distribution, namely from almost-perfect to partial recovery at some critical value of $P$’s expectation (where this critical value may depend on the shape of the distribution $P$). Moreover, with respect to question \#\ref{question:1} above, we expect the order of this phase transition, and other scaling properties in its vicinity, to be robust as long as $P^\prime(0)>0$.

\item Finally, what about planted models with spatial structure, as in the original problem of particle tracking from~\cite{Chertkov2010}?

\end{enumerate}

%% file: Sections/Proofs/erlangdiff.tex
The moment generating function for an exponential random variable $Y$ with rate $\lambda$ is
\[
\expect[\e^{\mu Y}] = \frac{\lambda}{\lambda-\mu} \, . 
\]
Since $X_1$ and $X_2$ are independent, the exponential generating function for $X_1-X_2$ is
\begin{equation}
\label{eq:mgf-x1-x2}
\expect[\e^{\mu (X_1-X_2)}]
= \expect[\e^{\mu X_1}] \,\expect[\e^{-\mu X_2)}]
= \left( \frac{\lambda_1 \lambda_2}{(\lambda_1-\mu)(\lambda_2+\mu)} \right)^{\!t} \, ,
\end{equation}
By Markov's inequality
\[
\prob[X_1 > X_2] = \prob[\e^{\mu (X_1-X_2)} > 1] \le \expect[\e^{\mu (X_1-X_2)}] 
\]
for any $\mu > 0$. The right-hand side of~\eqref{eq:mgf-x1-x2} is minimized when $\mu = (\lambda_1-\lambda_2)/2$, giving the desired result.

%% file: Sections/odeanal.tex
In this section, we state and prove Theorem \ref{thm:ode_unique_solution}.
\begin{theorem}\label{thm:ode_unique_solution}
When $\lambda<4$, the system of ODEs \eqref{eq:ODE1}--\eqref{eq:ODE4} has a unique solution $(F,G, V, W)$ 
satisfying  conditions \eqref{eq:ODE_boundary}--\eqref{eq:ODE_regularity}. 
\end{theorem}

When $V(x) \neq 0$, recall that 
$$
U(x) = \frac{F(x)}{V(x)};
$$
Thus, when $V(x) \neq 0$ and $W(x) \neq 0$,  the conservation law $FW+GV-VW=0$ is equivalent to 
$$
\frac{G(x)}{W(x)} = 1- U(x).
$$
Also, recall that from the conservation law 
we have $V(0)=W(0)=2F(0)=2G(0)$.
Hence,  the previous $4$-dimensional system of ODEs \eqref{eq:ODE1}--\eqref{eq:ODE4} 
with conditions \eqref{eq:ODE_boundary}--\eqref{eq:ODE_regularity}
reduces to the following $3$-dimensional system of ODEs: 
\begin{equation}
\begin{aligned}
&\frac{\dU}{\dx}  = - \lambda U(1-U ) + (1- UV)  \left(1- (1-U) W \right)  \\
&\frac{\dV}{\dx}   = \lambda V (1-U)  \\
& \frac{\dW}{\dx} = - \lambda W U 
\end{aligned}\label{eq:ODE_UVW}
\end{equation}
with initial condition
\begin{align}
& U(0) = \frac{1}{2}, \quad V(0)=W(0)=\epsilon, \quad \epsilon \in [0,1] \label{eq:ODE_UVW_initial}.
\end{align}

Note that the partial derivatives of the right hand side of \eqref{eq:ODE_UVW} with respect to $(U,V, W)$ are continuous. Therefore, by the standard existence and uniqueness theorem for solutions of systems of ODEs (see e.g.~\cite[Theorem 2]{Pontryagin1962})
it follows that the system \eqref{eq:ODE_UVW}
with the initial condition \eqref{eq:ODE_UVW_initial}
has a unique solution in a neighborhood of $0$ for a fixed $\epsilon \in [0,1]$.
We write this unique solution as $U(x,\epsilon)$, $V(x,\epsilon)$, and $W(x,\epsilon)$,
which we abbreviate as $\left( U, V, W \right)$ whenever the context is clear. 
We extend the neighborhood to the maximum interval $I_\epsilon$ 
in which the solution $(U(x,\epsilon), V(x,\epsilon), W(x,\epsilon))$ is finite, so that  
the existence and uniqueness theorem applies to $I_\epsilon$.
In particular, if the solution $(U(x,\epsilon), V(x,\epsilon), W(x,\epsilon))$ is finite for all $x \in [0, \infty)$, then
$I_\epsilon=[0,\infty)$; otherwise, at least one of  $U, V, W$ converges to $\infty$ as $x$ converges to some finite $x_0$
and then $I_\epsilon=[0,x_0)$. In the latter case, we say the solution is $\pm \infty$ for $x \geq x_0$.

Therefore, to prove Theorem \ref{thm:ode_unique_solution},
it suffices to show that the system of ODEs \eqref{eq:ODE_UVW}
with the initial condition \eqref{eq:ODE_UVW_initial} has a solution $\left(U(x,\epsilon_0),V(x,\epsilon_0), W(x,\epsilon_0)\right)$
for $x \in [0, +\infty)$ satisfying the boundary condition 
$U(+\infty,\epsilon_0)=V(+\infty,\epsilon_0)=1$ and $W(+\infty,\epsilon_0)=0$ for a unique $\epsilon_0 \in [0,1]$. 
Geometrically speaking, this is due to the fact that $(U=1,V=1,W=0)$ is a \emph{saddle point}, and 
there is a unique choice of $\epsilon_0$ so that the trajectory $(U(x,\epsilon_0),V(x,\epsilon_0),W(x,\epsilon_0))$
falls into the stable manifold i.e., set of initial conditions $(U(0,\epsilon_0), V(0,\epsilon_0),W(0,\epsilon_0))$ such that
$(U(x,\epsilon_0),V(x,\epsilon_0),W(x,\epsilon_0)) \to (1,1,0)$ as $x\to +\infty$.
For any other choice of $\epsilon \neq \epsilon_0$, 
the trajectory $(U(x,\epsilon),V(x,\epsilon),W(x,\epsilon))$ veers away from $(1,1,0)$ to infinity. 

The outline of the proof is as follows. We first prove some basic properties satisfied by the
solution $(U,V,W)$ in Section \ref{app:basic_ode}. Then based on these properties, in Section \ref{app:monotone_ode} 
we prove that the solution satisfies
some monotonicity properties with respect to $\epsilon$ by studying the sensitivity of the solution to the initial
condition. Next, in Section \ref{app:limit_ode} we characterize the limiting behavior of the solution depending on whether it
hits $1$ or not. The monotonicity properties and the limiting behavior enable us to completely characterize the basins of attraction
in  Section \ref{app:basin_ode}. In particular, we show that the basin of attraction for $(U=1,V=1,W=0)$ is a singleton, \ie, 
there is a unique choice of $\epsilon_0 \in [0,1]$ such that 
$(U(x,\epsilon_0),V(x,\epsilon_0),W(x,\epsilon_0) \to (1,1,0)$ as $x\to +\infty$. 
Finally, we connect the 3-dimensional system of ODEs \eqref{eq:ODE_UVW}
back to the 4-dimensional system of ODEs \eqref{eq:ODE1}--\eqref{eq:ODE4}
and finish the proof of  Theorem \ref{thm:ode_unique_solution} in Section \ref{app:finish_ode}.

%% file: Sections/basicprop.tex
In the following lemma, we prove some basic properties of the solution. 
\begin{lemma}\label{lmm:basic0}
Fix any $\epsilon \in (0,1]$. 
Then for any $x \in [0, +\infty)$ such that the unique solution $(U,V,W)$ 
is well-defined (not equal to $\pm \infty$), 
it holds that
\begin{align*}
V(x)= W(x) \,\e^{\lambda x},  \quad UV <1, \quad (1-U) W <1 ,  \quad U, V, W > 0.
\end{align*}
\end{lemma}
\begin{proof}
It follows from \eqref{eq:ODE_UVW} that 
\begin{align}
V(x) &=\epsilon \exp \left(  \lambda \int_{0}^x \left(1- U(y) \right) \dy \right) >0 , 
\label{eq:V_expression}\\
W(x) & =\epsilon \exp \left(  -\lambda \int_{0}^x U(y) \dy \right) >0. 
\label{eq:W_expression}
\end{align}
Hence $V(x)=W(x) \,\e^{\lambda x}$. Thus the conservation law $FW+GV-VW=0$ implies that 
\[
V= F +G \e^{\lambda x} , \qquad W= F \e^{-\lambda x} + G.
\]
Recall that $F=UV$ and $G=(1-U)W$. Then 
\begin{equation}
\begin{aligned}
\frac{\dF}{\dx} & = \left(1- F \right) \left( 1- G \right) V  = \left(1- F \right) \left( 1- G \right)  \left(F +G \e^{\lambda x} \right)  \\
\frac{\dG}{\dx} & = - \left(1- F \right) \left( 1- G \right) W  = -\left(1- F \right) \left( 1- G \right) \left( F \e^{-\lambda x} + G \right). 
\end{aligned}\label{eq:ODE_init2}
\end{equation}
For the sake of contradiction, suppose $\max \{ F(x), G(x) \} \ge 1$ for some finite $x>0$. 
Since $F(x)$ and $G(x)$ are continuous in $x$ and $F(0)=G(0)<1$,  there is an $x_0>0$ such that $\max \{ F(x_0), G(x_0) \} = 1$. 
Define $\tilde{F}(x) \equiv F(x_0)$ and $\tilde{G}(x) \equiv G(x_0)$. Then
$(\tilde{F}(x), \tilde{G}(x) )$ is a solution to ODE \eqref{eq:ODE_init2} in 
$x\in[0, x_0]$ running backward with its initial value at $x = x_0$ 
given by $\left( F(x_0), G(x_0) \right)$. 
Note that $\left( F(x), G(x) \right)$ is also a solution to ODE \eqref{eq:ODE_init2} in the backward time $x\in[0, x_0]$ with its initial value at $x = x_0$ given by $\left( F(x_0), G(x_0) \right)$. Also note that the right hand side of ODE \eqref{eq:ODE_init2}
is continuous in $x$ and the partial derivatives with respect to $F$ and $G$
are continuous. By existence and uniqueness~\cite[Theorem 2]{Pontryagin1962} 
it follows that 
$F(x) \equiv F(x_0)$ and $G(x) \equiv G(x_0)$ for $x \in [0, x_0]$. 
Hence, $\max\{ F(0), G(0) \} = 1$, which contradicts that 
$F(0)=G(0)=\epsilon/2<1$. Thus, $F=UV<1$ and $G=(1-U)W<1$. 

Next, we argue that $U> 0$. Suppose not,  since $U(0) =1/2$, by the differentiability of $U$ in $x$,
there exists a finite $x_0>0$ such that  $U(x_0)=0$ and $U'(x_0) \le 0$. 
However, 
$$
\frac{\dU}{\dx} \bigg|_{x=x_0}= \left[ 1- U (x_0) V(x_0) \right] \left[1- (1-U(x_0) ) W(x_0) \right] = \left[ 1- F(x_0) \right] \left[ 1- G(x_0)\right]>0,
$$
which leads to a contradiction. 
\end{proof}

%% file: Sections/initmonot.tex
The key to our proof is to study how the solution of 
the system of ODEs \eqref{eq:ODE_UVW}
changes with respect to the initial condition \eqref{eq:ODE_UVW_initial}. 

Standard ODE theory (see~\cite[Theorem 15]{Pontryagin1962}) shows that $U(x,\epsilon)$ is differentiable in $\epsilon$ and 
the mixed partial derivatives satisfy 
$$
\frac{\partial^2 U(x,\epsilon)}{\partial x \partial \epsilon} =\frac{\partial^2 U(x,\epsilon)}{\partial \epsilon \partial x};
$$
similarly for $V$ and $W$. Moreover, the partial derivatives $(\partial U/\partial \epsilon,
\partial V/\partial \epsilon, \partial W/\partial \epsilon)$ satisfy the following system of equations:
\begin{equation}
\begin{aligned}
\frac{\partial }{\partial x} \frac{\partial U}{\partial \epsilon} & = 
\left[ - \lambda (1-2U )  - V \left(1- (1-U) W \right) + (1- UV) W 
\right] \frac{\partial U}{\partial \epsilon} \\
& \qquad  -U  \left(1- (1-U) W \right) \frac{\partial V}{\partial \epsilon}
- (1- UV) (1-U) \frac{\partial W}{\partial \epsilon} \\
\frac{\partial }{\partial x} \frac{\partial V}{\partial \epsilon}  & = - \lambda  V  \frac{\partial U}{\partial \epsilon}   + \lambda (1-U)
\frac{\partial V}{\partial \epsilon} \\
\frac{\partial }{\partial x} \frac{\partial W}{\partial \epsilon}  &=  - \lambda W \frac{\partial U}{\partial \epsilon} -\lambda U 
\frac{\partial W}{\partial \epsilon},
\end{aligned} \label{eq:ODE_initial}
\end{equation}
with initial condition
\begin{align}
 \frac{\partial U (0, \epsilon) }{\partial \epsilon} =0, \quad  \frac{\partial V(0,\epsilon)}{\partial \epsilon} =1, \quad  \frac{\partial W(0,\epsilon)}{\partial \epsilon} =1.
 \label{eq:ODE_sens_ini_cond}
\end{align}
The system of equations \eqref{eq:ODE_initial} is known as
the \emph{system of variational equations} and can be derived 
by differentiating \eqref{eq:ODE_UVW} with respect to $\epsilon$ and interchange
$\partial x$ and $\partial \epsilon$. The initial condition \eqref{eq:ODE_sens_ini_cond}
can be derived by differentiating \eqref{eq:ODE_UVW_initial} with respect to $\epsilon$.

The following key lemma shows that whenever $U(x,\epsilon)\le1$, 
$U(x,\epsilon)$ is decreasing in $\epsilon$,
while $V(x,\epsilon)$ and $W(x,\epsilon)$ are increasing in $\epsilon$.
\begin{lemma}\label{lmm:initial_diff}
Fix $\epsilon \in (0,1)$. Suppose $U(x,\epsilon) \le 1$ for $x \in (0, c]$ for a finite constant $c>0$.
Then for all $x \in (0, c]$, 
\begin{align}
\frac{\partial U(x, \epsilon)  }{\partial \epsilon} <0, \quad \text{ and } \quad \frac{\partial W (x, \epsilon)}{\partial \epsilon} >0. \label{eq:UW_initial_derivative}
\end{align} 
 Moreover, it follows that all $x \in (0, c]$,
\begin{align}
\frac{\partial V(x, \epsilon) }{\partial \epsilon} \ge \exp \left(  \lambda \int_{0}^x \left( 1- U(t,\epsilon) \right) d t \right) = \frac{V(\epsilon)}{\epsilon}\ge 1.
\label{eq:V_initial}
\end{align}
\end{lemma}
\begin{proof}
We first show that \eqref{eq:V_initial} holds whenever $\partial U(x, \epsilon)/\partial \epsilon <0$
for $x \in (0,c]$. Recall that in Lemma \ref{lmm:basic0} we have shown that $V>0$.
It follows from  \eqref{eq:ODE_initial} that for all $x \in (0, c]$
$$
\frac{\partial }{\partial x} \frac{\partial V}{\partial \epsilon}  \ge  \lambda (1-U) \frac{\partial V}{\partial \epsilon}.
$$
Thus for all $x \in (0,c]$
$$
\frac{\partial V(x, \epsilon)  }{\partial \epsilon} \ge \exp \left( \lambda \int_{0}^x \left(1-U(s) \right) ds \right)
=\frac{V(x,\epsilon)}{\epsilon} \ge 1,
 $$
 where the equality holds due to \eqref{eq:V_expression}.

Next we show \eqref{eq:UW_initial_derivative}.
For the sake of contradiction, suppose not, \ie, there exists a $x_0 \in (0,c]$ 
such that either $\frac{\partial U(x_0, \epsilon)  }{\partial \epsilon} \ge 0$ or 
$\frac{\partial W(x_0, \epsilon)  }{\partial \epsilon} \le 0$. 

Define
$$
a= \inf \left\{ x \in (0,c]: \frac{\partial U(x, \epsilon)  }{\partial \epsilon}  \ge 0 \right\}
$$
and 
$$
b=\inf \left\{ x \in (0,c]: \frac{\partial W(x, \epsilon)  }{\partial \epsilon}  \le 0 \right\},
$$
with the convention that the infimum of an empty set is $+\infty$. 
Then $\min \{a, b\} \le x_0 \le c$. 

{\bf Case 1:} Suppose $a \le b$. 
Due to the initial condition \eqref{eq:ODE_sens_ini_cond} and 
the initial condition \eqref{eq:ODE_UVW_initial}, we have that
$$
 \frac{\partial U (0, \epsilon) }{\partial \epsilon} =0, 
 \quad \frac{\partial }{\partial x} \frac{\partial U (0, \epsilon) }{\partial \epsilon} 
 =  - \left( 1- \frac{\epsilon}{2} \right) <0. 
 $$
 Then we have $a>0$. Moreover, by the differentiability of $\frac{\partial U(x, \epsilon)  }{\partial \epsilon}$ in $x$
and the definition of $a$,  we have
$$
\frac{\partial U(x, \epsilon)  }{\partial \epsilon} <0,   \quad \forall x \in (0, a), 
\quad \frac{\partial U(a, \epsilon)  }{\partial \epsilon}  =0, \quad \text{ and } \quad \frac{\partial }{\partial x}  
\frac{\partial U(a, \epsilon)  }{\partial \epsilon}  \ge 0.
$$
It follows from our previous argument for proving  \eqref{eq:V_initial} that $\partial V(x,\epsilon)/\partial \epsilon \ge 1$ for all
$x \in (0,a]$.
Since $a \le b$, we also have that 
$$
 \frac{\partial W (a, \epsilon) }{\partial \epsilon} \ge 0.
$$

Recall that in Lemma \ref{lmm:basic0} we have shown that $UV<1$, $(1-U)W<1$ and $U>0$.
Moreover, by assumption we have $U \le 1$. Thus we get from ODE \eqref{eq:ODE_initial} that
$$
\frac{\partial }{\partial x}  
\frac{\partial U(a, \epsilon)  }{\partial \epsilon} =
-U  \left(1- (1-U) W \right) \frac{\partial V(a, \epsilon) }{\partial \epsilon}
- (1- UV) (1-U) \frac{\partial W(a, \epsilon) }{\partial \epsilon}
<0,
$$
which contradicts $\frac{\partial }{\partial x}  
\frac{\partial U(a, \epsilon)  }{\partial \epsilon}  \ge 0$. 

{\bf Case 2}: Suppose $a>b$. Due to the initial condition \eqref{eq:ODE_initial}, we have that
$$
\frac{\partial W(0,\epsilon)}{\partial \epsilon} =1.
$$
Thus $b>0$. 
By the differentiability of $\frac{\partial W(x, \epsilon)  }{\partial \epsilon}$ in $x$,  we have that 
$$
\frac{\partial W(b, \epsilon)  }{\partial \epsilon}  =0, \quad \text{ and } \quad \frac{\partial }{\partial x}  
\frac{\partial W(b, \epsilon)  }{\partial \epsilon}  \le 0, \quad \text{ and } \quad  \frac{\partial U (b, \epsilon) }{\partial \epsilon} < 0.
$$
Recall that in Lemma \ref{lmm:basic0} we have shown that  $W>0$. It follows from ODE \eqref{eq:ODE_initial} that
$$
\frac{\partial }{\partial x} \frac{\partial W(b,\epsilon)}{\partial \epsilon}  
=  - \lambda W \frac{\partial U(b,\epsilon) }{\partial \epsilon} >0,
$$
which contradicts $\frac{\partial }{\partial x}  
\frac{\partial W(b, \epsilon)  }{\partial \epsilon}  \le 0$.
\end{proof}

Based on  Lemma \ref{lmm:initial_diff}, we prove another ``monotonicity'' lemma, showing that if  
$U(x,\epsilon_0) < 1$ for all $x\ge 0$ and some $\epsilon_0 \in (0,1)$, 
then $U(x,\epsilon) <1$ for all $x \ge 0$ and all $\epsilon \in (\epsilon_0,1)$.
\begin{lemma}\label{lmm:U_monotone}
Suppose $U(x,\epsilon_0) < 1$ 
for all $x\ge 0$ and some $\epsilon_0 \in (0,1)$.  
Then $U(x,\epsilon) < 1$ for all $\epsilon \in (\epsilon_0,1)$ and all $x\ge 0$.
\end{lemma}
\begin{proof}
Fix an arbitrary but finite $x_0>0$. 
We claim that $U(x_0,\epsilon)  < 1$ for all $\epsilon \in (\epsilon_0, 1)$.
Suppose not. Then define 
$$
\epsilon_1 \triangleq \inf \left\{ \epsilon \in (\epsilon_0, 1):  
1 \le U(x_0,\epsilon) <+\infty \right\}
$$ 
Note that by assumption, $U(x_0,\epsilon_0) < 1$.
By the definition of $\epsilon_1$ and the differentiability of 
$U(x_0,\epsilon)$ in $\epsilon$, we have
$$
U(x_0,\epsilon_1) =1, \quad \frac{\partial U(x_0, \epsilon_1)  }{\partial \epsilon} \ge 0.
$$
We claim that $U(x, \epsilon_1)<1$ for all $x \in (0, x_0)$.  If not, then there exists an
$x_1 \in (0, x_0)$ such that $U(x_1,\epsilon_1)=1$. 
Note that $\frac{\dU(x,\epsilon_1) }{\dx} >0$ if $U(x) \ge 1$. 
Thus $U(x,\epsilon_1)>1$ for all $x> x_1$, 
which contradicts the fact that 
$U(x_0,\epsilon_1)=1$. Therefore,  we can apply Lemma \ref{lmm:initial_diff} with $c=x_0$ and
get that
$$
\frac{\partial U(x_0, \epsilon_1)  }{\partial \epsilon}  < 0,
$$
which contradicts the fact that $\frac{\partial U(x_0, \epsilon_1)  }{\partial \epsilon} \ge 0$. 
Since $x_0$ is arbitrarily chosen, we conclude that 
$U(x,\epsilon)  < 1$ for all $\epsilon \in (\epsilon_0, 1)$ and all $x>0$. 
\end{proof}

%% file: Sections/limitode.tex
In this section, we characterize the limiting 
behavior of $(U,V,W)$, depending on whether $U$ or $V$ hit $1$.

First, we state a simple lemma, showing that if 
both $U$ and $V$ do not hit $1$ in finite time, 
then they converge to $1$ as $x \to \infty$.
\begin{lemma}\label{lmm:UV_limit}
If
$U(x,\epsilon_0) < 1$ and $V(x,\epsilon_0) < 1$ for all $x \ge 0$ and some $\epsilon_0 \in (0,1)$, then 
$U(x,\epsilon_0) \to 1$, $V(x,\epsilon_0) \to 1$, and $W(x, \epsilon_0) \to 0$ as $x \to \infty$.
\end{lemma}
\begin{proof}
By Lemma \ref{lmm:basic0}, we have $W(x) = V(x) \,\e^{-\lambda x} \to 0$ as $x \to \infty$. 
Recall that according to~\eqref{eq:ODE_UVW},
$$
\frac{\dV}{\dx} = \lambda V  \left(1-U \right) > 0 \, .
$$
Since $V(x)  < 1$ for all $ x \ge 0$, it follows that $\frac{\dV}{\dx} \to 0$ and hence $U(x) \to 1$
as $x \to +\infty$.  Thus, as $x \to +\infty$, 
$$
\frac{\dU}{\dx} = -\lambda U (1-U) + \left(1- UV \right) \left( 1- (1-U) W  \right) \to 0 \, ,
$$
which implies that $V(x) \to 1$ as $x \to +\infty$. 
\end{proof}

The next lemma shows the behavior of $U$ and $V$ if they hit $1$ for finite $x$.

\begin{lemma}
\label{lmm:UV_blowup}
Let $x_0>0$ be finite. 
\begin{itemize}
\item If $V(x_0)=1$, then $V(x)$ monotonically increases to $+\infty$ 
and $U(x) \to 0$ for $x \ge x_0$.
\item If $U(x_0)=1$, then $U(x)$ monotonically increases to $+\infty$ and $V(x) \to 0$ for $x \ge x_0$.
\end{itemize}
\end{lemma}
\begin{proof}
Suppose $V(x_0)=1$. 
Recall that in Lemma \ref{lmm:basic0}, we have shown that $UV<1$. 
According to ODE \eqref{eq:ODE_UVW},
we get that $dV/dx>0$ if $V\ge 1$ as $UV<1$. 
Thus $V(x)$ monotonically increases to $+\infty$ for $x \ge x_0$.
Moreover $U(x) \to 0$ for $x \ge x_0$.

Suppose $U(x_0)=1$. Recall that in Lemma \ref{lmm:basic0}, we have shown that 
$(1-U)WV<1$. According to ODE \eqref{eq:ODE_UVW},
we get that $dU/dx>0$ if $U \ge 1$. Hence,  $U(x)$ monotonically increases to $+\infty$
for $x \ge x_0$. As $UV<1$, it further follows that $V(x) \to 0$ for $x \ge x_0$.
\end{proof}

%% file: Sections/basinattract.tex
In view of Lemma \ref{lmm:UV_limit} and Lemma \ref{lmm:UV_blowup}, define the basin of attraction for $( U=0, V=+\infty)$ as
$$
S_1= \left\{ \epsilon \in [0,1]:  V(x,\epsilon ) \ge 1 \text{ for some finite $x>0$ } \right\};
$$
the basin of attraction for $(U=+\infty, V=0)$  as 
$$
S_2=  \left\{ \epsilon \in [0,1]: U(x,\epsilon) \ge 1 \text{ for some finite $x>0$ }   \right\};
$$
and the basin of attraction for $(U=1, V=1)$ as 
$$
S_0=  \left\{ \epsilon \in [0,1]: U(x,\epsilon) <1  \text{ and }  V(x,\epsilon ) <1  \text{ for all finite $x>0$ }  \right\}.
$$

When $\epsilon$ is either $0$ or $1$, we have the following simple characterizations of the solution. 
\begin{lemma}\label{lmm:sol_eps_01}
Suppose $\lambda<4$. 
\begin{itemize}
\item If $\epsilon=0$, then $V(x)\equiv 0$, $W(X) \equiv 0$ and 
$U(x)$ monotonically increases to $+\infty$.
\item If $\epsilon=1$, then $V(x)$ monotonically increases to $+\infty$ and $U(x) \to 0$.
\end{itemize}
\end{lemma}
\begin{proof}
First, consider the case $\epsilon=0$. Then according to 
the system of ODEs \eqref{eq:ODE_UVW}, we immediately get that 
$V(x)\equiv 0$, $W(X) \equiv 0$. Thus 
$$
\frac{\dU}{\dx}  = - \lambda U(1-U ) + 1 >0,
$$
where the last inequality holds due to $\lambda<4$. 
Hence, $U(x)$ monotonically increases to $+\infty$. 

The conclusion in the case $\epsilon=1$ simply follows from Lemma \ref{lmm:UV_blowup}.
\end{proof}

Now, we are ready to prove a lemma, which completely characterizes the 
basins of attraction $S_0,$ $S_1,$ and $S_2.$

\begin{lemma}\label{lmm:unique_epsilon}
Suppose $\lambda<4$. Then there exists a unique $\epsilon_0 \in (0,1)$ such that
\begin{align}
S_0  = \{\epsilon_0\}, \quad S_1 = (\epsilon_0, 1], \quad S_2 = [0, \epsilon_0). \label{eq:S012}
\end{align} 
\end{lemma}
\begin{proof}
Lemma \ref{lmm:sol_eps_01} implies that $1 \in S_1$ and $0 \in S_2$. 
Note that $UV<1$ by Lemma \ref{lmm:basic0}.
Thus it follows from Lemma \ref{lmm:UV_blowup} that $S_1$ and $S_2$ are disjoint.

We first prove that $S_1$ is left open. Fix any $\epsilon \in S_1$. 
Since $V(x,\epsilon) \ge 1$ for some finite $x$, it follows from Lemma \ref{lmm:UV_blowup} that 
there exists an $x_0$ such that $V(x_0, \epsilon)>1$. 
By the continuity of $V(x_0,\epsilon)$ in $\epsilon$, there exists a $\delta>0$ such that for all
$ \epsilon' \in [\epsilon-\delta, \epsilon] $, $V(x_0, \epsilon')>1$,
and thus $V(x, \epsilon') \to +\infty$ and 
$U(x,\epsilon')\to 0$ as $x \to +\infty$. 
Hence, $[\epsilon-\delta, \epsilon]  \subset S_1$. Thus $S_1$ is left open.

Analogously, we can prove that $S_2$ is right open. 
Note that $S_0=[0,1] \setminus (S_1 \cup S_2)$, and $S_1$ and $S_2$ are disjoint. 
It follows that $S_0$ is non-empty. Let $\epsilon_0$ be any point in $S_0.$
Next we prove \eqref{eq:S012}.

We first fix any $\epsilon \in (\epsilon_0,1)$. 
Since $\epsilon_0 \in S_0$, it follows that $U(x,\epsilon_0) < 1$ and $V(x,\epsilon_0) < 1$ for all $x\ge0$. 
In view of Lemma \ref{lmm:U_monotone}, we have that $U(x,\epsilon)  < 1$ for all $\epsilon \in (\epsilon_0, 1)$ and all $x>0$. 
It follows from Lemma \ref{lmm:initial_diff} that $\partial V(x, \epsilon)/\partial \epsilon \ge 1$
for all $x>0$ and all $\epsilon \in (\epsilon_0,1)$.
Thus for all $x \ge 0$, 
$$
V(x, \epsilon) = V(x, \epsilon_0) + \int_{\epsilon_0}^{\epsilon}  \frac{\partial V(x, \eta) }{\partial \eta} d \eta
\ge V(x, \epsilon_0) + (\epsilon-\epsilon_0). 
$$
Since $V(x,\epsilon_0) \to 1$ as $x \to +\infty$, 
there exists an $x_0$ such that for all $x \ge x_0$,
$$
V(x, \epsilon_0) \ge 1- (\epsilon-\epsilon_0)/2.
$$ 
Combining the last two displayed equation gives that for all $x \ge x_0$, 
$$
V(x,\epsilon) \ge 1+ (\epsilon-\epsilon_0)/2>1. 
$$
We conclude that $\epsilon \in S_1$ and thus $(\epsilon_0,1] \subset S_1$.

Next we fix any $\epsilon \in (0, \epsilon_0)$ and show that $\epsilon \in S_2.$ 
Suppose not. Then there exists an 
$\epsilon_1 \in (0,\epsilon_0)$ such that $U(x,\epsilon_1)<1$  for all $x\ge0$. 
By Lemma \ref{lmm:U_monotone}, we have that 
$U(x,\epsilon)  < 1$ for all $\epsilon \in (\epsilon_1, 1)$ and all $x>0$. 
In view of Lemma \ref{lmm:initial_diff}, it immediately follows that 
\eqref{eq:V_initial} holds for all $x>0$ and all $\epsilon \in (\epsilon_1,1)$. 
Thus, 
$$
V(x, \epsilon_0) = V(x, \epsilon_1) + \int_{\epsilon_1}^{\epsilon_0}  \frac{\partial V(x, \epsilon) }{\partial \epsilon} d \epsilon
\ge V(x, \epsilon_1) + (\epsilon_0-\epsilon_1) .
$$
Note that since $\epsilon_0 \in S_0$, $V(x, \epsilon_0)<1$ for all $x \ge 0$, it follows that for all $x \ge 0$,
$$
V(x, \epsilon_1)  < 1 -  ( \epsilon_0-\epsilon_1),
$$
which contradicts the conclusion of Lemma \ref{lmm:UV_limit}. 
Thus we conclude that $\epsilon \in S_2$ and thus $[0,\epsilon_0) \subset S_2$. 

Since $S_0$, $S_1$, and $S_2$ are all disjoint,  the desired \eqref{eq:S012} readily follows.
\end{proof}

%% file: Sections/finishode.tex
We are now ready to prove Theorem \ref{thm:ode_unique_solution}. 
Let $S_0=\{\epsilon_0\}$, and let $(U(x,\epsilon_0), V(x,\epsilon_0),W(x,\epsilon_0))$ 
be the unique solution to the system of ODEs \eqref{eq:ODE_UVW}
with the initial condition \eqref{eq:ODE_UVW_initial}. 
For $x \in [0, +\infty)$, define
 \begin{align*}
 F(x) = U(x,\epsilon_0) V(x, \epsilon_0) \, , \quad & F(-x)= \left(1-U(x,\epsilon_0) \right) W(x,\epsilon_0)  \\
 V(x)= V(x,\epsilon_0) \, , \quad & V(-x) = W(x, \epsilon_0) \\
 G(x)=F(-x) \, , \quad & G(-x)=F(x) \\
 W(x)=V(-x) \, , \quad & W(-x)=V(x) \, .
 \end{align*}

We show that $(F,G,V,W)$ is a solution to
the system of ODEs \eqref{eq:ODE1}--\eqref{eq:ODE4}
with conditions \eqref{eq:ODE_boundary}--\eqref{eq:ODE_regularity}.
First, by construction $(F,G,V,W)$ satisfy the system of ODEs \eqref{eq:ODE1}--\eqref{eq:ODE4}.
In particular, since $U(0,\epsilon_0) = 1/2$ and $V(0,\epsilon_0)=W(0,\epsilon_0)=\epsilon_0$, we have
\begin{align*}
	\lim_{x\downarrow 0} \frac{\dV}{\dx} (x, \epsilon) &= \lambda V(0,\epsilon_0)  \left(1-U(0,\epsilon_0) \right) = \frac{\lambda \epsilon_0}{2} \\
	  \lim_{x\uparrow 0} \frac{\dV}{\dx} (x, \epsilon)&= - \lim_{x\downarrow 0} \frac{\dW}{\dx}(x, \epsilon)  =  \lambda W(0,\epsilon_0)  U(0,\epsilon_0) =  \frac{\lambda \epsilon_0}{2} \, .
\end{align*}
Therefore, 
$V$ is differentiable at $0$. 
Analogously, we can verify that $F$, $G$, and $W$ are differentiable at $0$. 

Second, since $\epsilon_0 \in S_0$, by definition
$U(x,\epsilon_0)<1$ and $V(x,\epsilon_0)<1$ for all $x\ge0$.
Thus it follows from Lemma \ref{lmm:UV_limit} that as $x \to +\infty$, $U(x,\epsilon_0) \to 1$, $V(x,\epsilon_0)\to 1$, 
and $W(x,\epsilon_0) \to 0$. Hence, $(F,G,V,W)$ satisfy condition \eqref{eq:ODE_boundary}.
Thirdly, in view of Lemma \ref{lmm:basic0}, we have that 
$U(x,\epsilon_0), V(x,\epsilon_0),W(x,\epsilon_0)>0$, $W(x,\epsilon_0)<1$,
$U(x,\epsilon_0) V(x,\epsilon_0)<1$, 
and $\left( 1- U(x,\epsilon_0) \right) W(x,\epsilon_0)<1$.
Therefore, $0<V, W<1$ and $0<F, G<1$, satisfying condition \eqref{eq:ODE_regularity}.

Next, we show that the solution $(F,G,V,W)$ is unique. Let $(\tilde{F},\tilde{G},\tilde{V},\tilde{W})$
denote another solution to system of ODEs \eqref{eq:ODE1}--\eqref{eq:ODE4}
with conditions \eqref{eq:ODE_boundary}--\eqref{eq:ODE_regularity}.
Let $\tilde{U}=\tilde{F}/\tilde{V}$. Then $(\tilde{U}, \tilde{V},\tilde{W})$ 
is a solution to the system of ODEs \eqref{eq:ODE_UVW},
satisfying the initial condition \eqref{eq:ODE_UVW_initial} with $\epsilon=\tilde{V}(0)=\tilde{W}(0)$.
Moreover, $\tilde{U}(x)<1$ and $\tilde{V}(x)<1$ for all $x\ge0$, because otherwise
by Lemma \ref{lmm:UV_blowup}, either $\tilde{U}(x) \to +\infty$ or $\tilde{V}(x) \to +\infty$,
violating that $\tilde{F}(x), \tilde{V}(x) \to 1$. As a consequence, $\tilde{V}(0) \in S_0$. 
It follows from Lemma \ref{lmm:unique_epsilon} that $\epsilon_0=\tilde{V}(0)$. 
By the uniqueness of the solution to system of ODEs \eqref{eq:ODE_UVW}
with the initial condition \eqref{eq:ODE_UVW_initial}, we have
$\tilde{U}(x) \equiv U(x,\epsilon_0)$, $\tilde{V} (x) \equiv V(x,\epsilon_0)$,
$\tilde{W} (x) \equiv V(x,\epsilon_0)$. Thus, $(\tilde{F},\tilde{G},\tilde{V},\tilde{W})=(F,G,V,W).$

%% file: Sections/GeneralSetup.tex
In this section and the succeeding ones we define the planted Poisson Weighted Infinite Tree (planted PWIT), define a matching $\Minfopt$ on it, prove that it is optimal and unique, and prove that the minimum weight matching $\Mnopt$ on $K_{n,n}$ converges to it in the local weak sense. We follow the strategy of Aldous' celebrated proof of the $\pi^2/6$ conjecture in the un-planted model~\cite{Aldous1992,Aldous2001}, and in a few places the review article of Aldous and Steele~\cite{Aldous2004}. There are some places where we can simply re-use the steps of that proof, and others where the planted model requires a nontrivial generalization or modification, but throughout we try to keep our proof as self-contained as possible.

In this section we lay out our notation, and formally define local weak convergence. We apologize to the reader in advance for the notational complications they are about to endure: there are far too many superscripts, subscripts, diacritical marks, and general doodads on these symbols. But some level of this seems to be unavoidable if we want to carefully define the various objects and spaces we need to work with. 

First off, the exponential distribution with rate $\nu$ is denoted by $\exp(\nu)$. Its cumulative distribution function is $\prob[x > t]=\e^{-\nu t}$ and its mean is $1/\nu$. For a Borel space $(S,\mathcal{S})$ consisting of a set $S$ and a $\sigma$-algebra $\mathcal{S}$, $\mathcal{P}(S)$ is the set of all Borel probability measures defined on $S$.

We use $\mathbb{Z}$ for the integers, $\mathbb{N}_0=\{0,1,2,\ldots\}$ and $\mathbb{N}_+=\{1,2,\ldots\}$ for the natural numbers, and $\mathbb{R}_+$ for the set of non-negative real numbers. The number of elements of a set $A$ is denoted by $|A|$. Random variables are denoted by capital letters; when we need to refer to a specific realization we sometimes use small letters.

Our graphs will be simple and undirected unless otherwise specified. Given an undirected graph $G=(V,E)$, a \emph{(perfect) matching} $M \subset E$ is a set of edges where every vertex $v \in V$ is incident to exactly one edge in $M$.  For each $v$, we refer to the unique $v'$ such that $\{ v,v' \} \in M$ as the \emph{partner} to $v$, and will sometimes denote it as $M(v)$; then $M(M(v))=v$. In a bipartite graph, a matching defines a one-to-one correspondence between the vertices on the left and those on the right. In a forgivable abuse of notation, will often write $M(v,v')=1$ if $\{v,v'\} \in M$ and $0$ otherwise.

A \emph{rooted graph} $G_{\circ} = (V,E,\root)$ is a graph $G= (V,E)$ with a distinguished vertex $\root \in V$. The \emph{height} of a vertex $v \in V$ in a rooted graph $G_{\circ} = (V,E,\root)$ is the shortest-path distance from $\root$ to $v$, i.e., the minimum number of edges among all paths from $\root$ to $v$. 

A \emph{planted graph} $G = (V,E,\Mplanted)$ is a graph $(V,E)$ together with a \emph{planted matching} $\Mplanted \subset E$. 
Similarly a \emph{rooted planted graph} $\Gs = (V,E,\Mplanted,\root)$ is a planted graph with a distinguished vertex $\root$. We refer to the edges in $\Mplanted$ and $E \setminus \Mplanted$ as the planted edges and un-planted edges respectively. 

Two planted graphs $G =(V,E,\Mplanted)$ and $G' = (V',E',{\Mplanted}')$ are said to be isomorphic if there exists a bijection $\gamma:V \to V'$ such that $\{v_1,v_2\} \in E$ if and only if $\gamma(\{v_1,v_2\}) \coloneqq \{ \gamma(v_1), \gamma(v_2) \} \in E'$, and $\{v_1,v_2\} \in \Mplanted$ if and only if $\gamma(\{v_1,v_2\}) = \{ \gamma(v_1), \gamma(v_2) \} \in {\Mplanted}'$. Thus the isomorphism $\gamma$ preserves the planted and un-planted edges. A rooted isomorphism from $G_\circ = (V,E,\Mplanted,\root)$ to $G'_\circ = (V',E',{\Mplanted}',\root')$ is an isomorphism between $G = (V,E,\Mplanted)$ and $G'=(V',E',{\Mplanted}')$ such that $\gamma(\root)=\root'$.

Next we endow a planted graph with a weight function. A \emph{planted network} $N = (G,\ell)$ is a planted graph $G = (V,E,\Mplanted)$ together with a function $\ell: E \to \mathbb{R}_+$ that assigns weights to the edges. For the sake of brevity, we write $\ell(v,w)$ instead of $\ell(\{v,w\})$. 

Now let $K_{n,n} = (V_n,E_n,\Mplanted_n)$ denote a complete bipartite graph together with a planted matching. We use $[n]$ to denote the set of integers $\{1,2,\ldots,n\}$. We label the vertices on the left-hand side of $K_{n,n}$ as $\{1,2,\ldots,n\}$, and the vertices on the right-hand side as $\{1',2',\ldots,n'\}$. In a slight abuse of notation, we denote these sets of labels $[n]$ and $[n']$ respectively. 
Thus $V_n = [n] \cup [n']$, $E_n = \{\{i,j'\} : i \in [n] \text{ and } j' \in [n'] \}$, and $\Mplanted_n =\{\{i,i'\} : i\in[n] \text{ and } i'\in[n'] \}$.

Let $\ell_n$ denote a random function that assigns weights to the edges of $K_{n,n}$ as follows: if $e \in \Mplanted$, then $\ell_n(e) \sim \exp(\lambda)$, and if $e \notin \Mplanted_n$ then $\ell_n(e) \sim \exp(1/n)$. We denote the resulting planted network as $(K_{n,n},\ell_n)$. We denote the minimum matching on $(K_{n,n},\ell_n)$ as $\Mnopt$. Figure \ref{fig:pmod} illustrates a realization of the planted model.

\begin{figure}[t!]
	\begin{center}
		\includegraphics[scale=0.6]{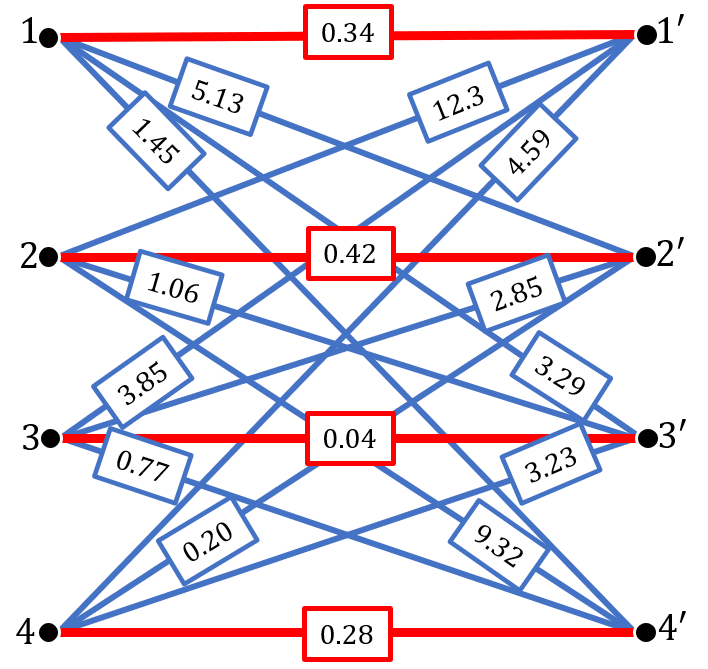}
		\caption{A realization of $(K_{4,4},\ell_4)$ for $n=4$ and $\lambda = 1$. Red bold edges are in $\Mplanted_n$ (planted edges), and solid blue edges are in $E_n\setminus \Mplanted_n$.}\label{fig:pmod}
	\end{center}
\end{figure}

We want to define a metric on planted networks, or rather on their isomorphism classes. Two planted networks $N = (G,\ell)$ and $N' = (G',\ell')$ are isomorphic if there is an isomorphism $\gamma$ between $G$ and $G'$ that preserves the length of the edges, i.e., if $\ell(v_1,v_2) = \ell'(\gamma(v_1),\gamma(v_2))$. A rooted planted network $\Ns = (\Gs,\ell)$ is a rooted planted graph $\Gs$ together with a weight function $\ell$, and we define rooted isomorphism as before. 
Let $[\Ns]$ denote the class of rooted planted networks that are isomorphic to $\Ns$. Henceforth, we use $\Ns$ to denote a typical member of $[\Ns]$.

Next, we define a distance function 
$d_\ell(v,v')$ as the shortest-path weighted distance between vertices but treating planted edges as if they have zero weight. That is,
\begin{equation}
\label{eq:d-ell}
d_\ell(v,v') \coloneqq \inf_{\text{paths $p$ from $v$ to $v'$}} \sum_{e \in p \setminus \Mplanted} \ell(e) \, . 
\end{equation}
For any vertex $v \in V$ and any $\rho \in \rplus$, we can consider the neighborhood $N_\rho(v) = \{ v' : d_\ell(v,v') \le \rho \}$. A network is \emph{locally finite} if $|N_\rho(v)|$ is finite for all $v$ and all $\rho$.


Now let $\Gstar$ denote the set of all isomorphism classes $[\Ns]$, where $\Ns$ ranges over all connected locally finite rooted planted networks. There is a natural way to equip $\Gstar$ with a metric. 
Consider a connected locally finite rooted planted network $\Ns = (\Gs,\ell)$. Now, for $\rho \in \rplus$, we can turn the neighborhood $N_\rho(\root)$ into a rooted subgraph $(\Gs)_\rho$. To be precise,
$(\Gs)_\rho = (V_\rho,E_\rho,\Mplanted_\rho,\root)$ is given as follows:
\begin{enumerate}[]
	\item \emph{Vertex set}: $V_\rho = N_\rho(\root) = \{v : d_\ell(v,\root) \leq \rho \}$.
	\item \emph{Edge set}: $e \in E_\rho$ if $e \in p$ for some path $p$ starting from $\root$ such that $\sum_{e \in p\setminus \Mplanted} \ell(e) \leq \rho$.
	\item \emph{Planted matching}: $\Mplanted_\rho = \Mplanted \cap E_\rho$.
\end{enumerate}
Given this definition, for any $[\Ns],[\Nsp] \in \Gstar$ a natural way to define a distance is 
\begin{align*}
d([\Ns],[\Nsp]) = \frac{1}{R+1} ,
\end{align*}
where $R$ is the largest $\rho$ at which the corresponding rooted subnetworks $((\Gs)_\rho, \ell) $ and $((\Gsp)_\rho,\ell')$ cease to be approximately isomorphic in the following sense:
\begin{equation}
\label{eq:gstar-metric}
R = \sup\left\{\rho \geq 0: \begin{minipage}{0.55\textwidth}
there exists a rooted isomorphism $\gamma_\rho : (\Gs)_\rho \to (\Gsp)_\rho$ 
such that $\forall e \in E_\rho$, 
$|\ell(e)-\ell'(\gamma_\rho(e))| < 1/\rho$ 
\end{minipage} \right\} \, .
\end{equation} 
(Note that this isomorphism is between the rooted subgraphs $(\Gs)_\rho$ and $(\Gsp)_\rho$, not the corresponding rooted networks, so it is not required to preserve the weights exactly.) In other words, $\Ns$ and $\Nsp$ are close whenever there is a large neighborhood around $\root$ where the edge weights are approximately the same, up to isomorphism. In particular, a continuous function is one that we can approximate arbitrarily well by looking at larger and larger neighborhoods of the root.

Equipped with this distance, we say that a sequence $([N_{n,\circ} = (G_{n,\circ},\ell_n)])_{n=1}^\infty$ \emph{converges locally to} $[N_{\infty,\circ}]$, and write $[N_{n,\circ}] \lwc [N_{\infty,\circ}]$,  
if the following holds: for all $\rho \in \rplus$ such that 
$N_{\infty,\circ}$ does not have a vertex at a distance exactly $\rho$ from the root $\root$, there is an $n_\rho \in \mathbb{N}_0$ such that for all $n > n_\rho$ there is a rooted isomorphism $\gamma_{n,\rho}:(G_{n,\circ})_\rho \to (G_{\infty,\circ})_\rho$ such that $\ell_n(\gamma^{-1}_{n,\rho}(e)) \to \ell(e)$ for all $e \in E_\rho$ where $E_\rho$ is defined from $G_{\infty,\circ}$ as above. That is, as $n$ increases, $N_{n,\circ}$ becomes arbitrarily close to $N_{\infty,\circ}$ on arbitrarily large neighborhoods. 

It is easy to check that $d$ defines a metric on $\Gstar$. Moreover, $\Gstar$ equipped with this metric is a Polish space: a complete metric space which is \emph{separable}, i.e., it has a countable dense subset. Hence, we can use the usual tools in the theory of weak convergence to study sequence of probability measures on $\Gstar$. More precisely, define $\mathcal{P}(\Gstar)$ as the set of all probability measures on $\Gstar$ and endow this space with the topology of weak convergence: a sequence $\mu_n\in \mathcal{P}(\Gstar)$ converges weakly to $\mu_{\infty}$, denoted by $\mu_n\xrightarrow{w}\mu$, if for any continuous bounded function $f:\Gstar \to \mathbb{R}$,
\begin{align*}
\int_{\Gstar} f\,d\mu_n \to \int_{\Gstar} f\,d\mu_{\infty}.
\end{align*}
Since $\Gstar$ is a Polish space, $\mathcal{P}(\Gstar)$ is a Polish space as well with the L\'evy-Prokhorov metric~\cite[pp. 394--395, Thm. 11.3.1 and Thm. 11.3.3]{Dudley2002}. Also, Skorokhod's theorem~\cite[p. 79, Thm. 4.30]{Kallenberg2002} implies that $\mu_n$ converges weakly to $\mu_{\infty}$ if an only if there are random variables $[N_{n,\circ}]$ and $[N_{\infty,\circ}]$ defined over $\Gstar$ such that $[N_{n,\circ}]\sim \mu_n$, $[N_{\infty,\circ}] \sim \mu_{\infty}$, and $[N_{n,\circ}]  \lwc [N_{\infty,\circ}]$ almost surely.

This notion of convergence in $\Gstar$ was first discussed by Aldous and Steele in~\cite{Aldous2004}. It is called {\em local weak convergence} to emphasize the fact that this notion of convergence only informs us about the local properties of measure around the root. We are going to use this framework to study the asymptotics of a sequence of finite planted networks. This methodology is known as \emph{the objective method}~\cite{Aldous2004} and has been used to analyze combinatorial optimization problems in a variety of random structures (e.g.~\cite{Aldous2004,Salez2009,Sundaresan2014,Gamarnik2006,Gamarnik2008,Gamarnik2004}).

In order to apply this machinery to random finite planted networks, consider a finite planted network $N = (G,\ell)$. For a vertex $v \in V$, let $\Ns(v)$ denote the planted network rooted at $v$ consisting of $v$'s connected component. Then we can define a measure $U(N) \in \mathcal{P}(\Gstar)$ as follows,
\begin{equation}
\label{eq:local-uniform}
U(N) = \frac{1}{|V|} \sum_{v\in V} \delta_{[\Ns(v)]} \, ,
\end{equation}
where $\delta_{[\Ns(v)]}\in\mathcal{P}(\Gstar)$ is the Dirac measure that assigns $1$ to $[\Ns(v)] \in \Gstar$ and $0$ to to any other member of $\Gstar$.	In other words, $U(N)$ is the law of $[\Ns(\root)]$ where $\root$ is picked uniformly from $V$. Now, to study the local behavior of a sequence of finite networks $( N_n )_n$, the objective method suggests studying the weak limit of the sequence of measures $( U(N_n) )_n$.
\begin{definition}(\textbf{Random Weak Limit})
\label{def:random-weak-limit}
	A sequence of finite planted networks $( N_n )_{n=1}^{\infty}$ has a random weak limit $\mu \in \mathcal{P}(\Gstar)$ if $U(N_n) \xrightarrow{w} \mu$. 
\end{definition}
If $N_n$ is a random planted network, we replace $U(N_n)$ in the above definition with $\expect U(N_n)$, where
\begin{equation}
\label{eq:local-uniform-expectation}
\expect U(N)\, (A) \coloneqq \expect[U(N)(A)] \qquad \text{for all Borel sets $A \subseteq \Gstar$,}
\end{equation}
and the expectation is taken with respect to the randomness of $N$. For us, in both $(K_{n,n},\ell_n)$ and the  weighted infinite tree $(N_{\infty,\circ},\ell_\infty)$ we define below, the only source of randomness is the edge weights. It is easy to see that if $N$ is vertex transitive, so that every vertex has the same distribution of neighborhoods, then $\expect U(N)$ is the law of $[N(\root)]$ (or of $[N(v)]$ for any vertex $v$). In many settings, e.g. sparse Erd\H{o}s-R\'enyi graphs, $U(N)$ converges in distribution to $N(\root)$, since averaging over all possible root vertices effectively averages over $N$ as well. But taking the expectation over $N$ as we do here avoids having to prove this.

Not all probability measures $\mu \in \mathcal{P}(\Gstar)$ can be random weak limits. The uniform rooting in the measure associated with finite networks implies a modest symmetry property on the asymptotic measure.  One necessary condition for a probability measure to be a random weak limit is called \emph{unimodularity}~\cite{Aldous2007}. 

To define unimodularity, let $\Gstarstar$ denote the set of all isomorphism classes $[\Nso]$, where $\Nso$ ranges over all connected locally finite \emph{doubly-rooted} planted networks---that is, networks with an ordered pair of distinguished vertices. We define $\Gstarstar$ as the set of equivalence classes under isomorphisms that preserve both roots, and equip it with a metric analogous to~\eqref{eq:gstar-metric} to make it complete and separable. A continuous function $f([\Nso(\root,v)])$ is then one which we can approximate arbitrarily well by looking at neighborhoods of increasing size that contain both $\root$ and $v$.

Then we can define unimodularity as follows:
\begin{definition}(\textbf{Unimodularity})
A probability measure $\mu \in \mathcal{P}(\Gstar)$ is \emph{unimodular} if for all Borel functions $f:\Gstarstar \to \mathbb{R}_+$,
	\begin{align}\label{eq:unimod}
	\int_{\Gstar} \sum_{v\in V} f([\Nso(\root,v)]) \,d\mu([\Ns(\root)]) = \int_{\Gstar} \sum_{v\in V} f([\Nso(v,\root)]) \,d\mu([\Ns(\root)]) \,  .
	\end{align}
\end{definition}
In other words, the expectation over $\mu$ of the sum (either finite or $+\infty$) over all $v$ of $f([\Nso(\root,v)])$ remains the same if we swap $\root$ and $v$. Since in a connected graph we can swap any vertex $v$ with $\root$ by a sequence of swaps between $\root$ and its neighbors, each of which moves $v$ closer to the root, this definition is equivalent to one where we restrict $f$ to Borel functions with support on $\{ [\Nso(\root,v)] \mid  \text{$\root$ and $v$ are neighbors} \}$. With this restriction, unimodularity is known as \emph{involution invariance}~\cite[Prop. 2.2]{Aldous2007}: 
\begin{lemma}(\textbf{Involution Invariance})
	A probability measure $\mu \in \mathcal{P}(\Gstar)$ is unimodular if and only if \eqref{eq:unimod} holds for all Borel functions $f:\Gstarstar \to \mathbb{R}_+$ such that $f([\Nso(\root,v)]) = 0$ unless $\{\root, v\}\in E$.
\end{lemma}

Aldous in \cite{Aldous2004} uses another characterization of involution invariance. Given a probability measure $\mu \in \mathcal{P}(\Gstar)$, define a measure $\widetilde{\mu}$ on $\Gstarstar$ as the product measure of $\mu$ and the counting measure on the neighbors of the root, i.e.,
\begin{align}\label{eq:involinvardef}
\widetilde{\mu}(\cdot) \coloneqq \int_{\Gstar} \sum_{v : \{\root, v\} \in E} \mathbf{1}([\Nso(\root,v)] \in \cdot) \,d\mu([\Ns(\root)]),
\end{align}
where $\one$ is the indicator function.  
Like $\mu$, $\widetilde{\mu}$ is a $\sigma$-finite measure. Throughout the following sections, we use the $\widetilde{\textrm{tilde}}$ 
to distinguish a measure associated with doubly-rooted planted networks from the corresponding measure associated with singly-rooted ones.

Then Aldous' definition of involution invariance in~\cite{Aldous2004} is as follows.
\begin{definition}(\textbf{Involution Invariance, again}) \label{def:spinvar}
	A probability measure $\mu \in \mathcal{P}(\Gstar)$ is said to be involution invariant if the induced measure $\widetilde{\mu}$ on $\Gstarstar$ is invariant under the involution map $\iota:\Gstarstar\to\Gstarstar$, i.e.,
	\begin{align*}
	\widetilde{\mu}(A) = \widetilde{\mu}(\iota^{-1}(A)) \qquad \text{for all Borel sets $A \subseteq \Gstarstar$,}
	\end{align*}
	where $\iota([\Nso(\root,v)]) = [\Nso(v,\root)]$.
\end{definition}
Crucially, unimodularity and involution invariance are preserved under local weak convergence. Any random weak limit satisfies unimodularity and is involution invariant~\cite{Aldous2007,Aldous2004} (although the converse is an open problem). 

The theory of local weak convergence is a powerful tool for studying random combinatorial problems. In the succeeding sections we will prove a series of propositions analogous to~\cite{Aldous1992,Aldous2001} showing local weak convergence between our planted model of randomly weighted graphs $K_{n,n}$ and a kind of infinite tree $N_{\infty,\circ}$. These propositions make a rigorous connection between the minimum matching on $K_{n,n}$ and the minimum involution invariant matching $\Minfopt$ on $N_{\infty,\circ}$. Finally, we analyze $\Minfopt$ using the RDEs that we solved with differential equations above. 

%% file: Sections/InfiniteObject.tex
In this section we define the planted Poisson Weighted Infinite Tree, and show that it is the weak limit of the planted model $(K_{n,n},\ell_n)$. 

Let us ignore the planted matching for the moment and assume that $\ell_n(e) \sim \exp(n)$ for all $e\in E_n$. The problem of finding the minimum matching on this un-planted network is known as the random assignment problem. Kurtzberg~\cite{Kurtzberg1962} introduced this problem with i.i.d.\ uniform edge lengths on $[0,n]$, and Walkup~\cite{Walkup1979} proved that the expected cost of the minimum matching is bounded and is independent of $n$. In the succeeding years, many researchers tightened the bound for $\expect[X_n]$ (e.g.~\cite{Karp1987,Lazarus1993,Goemans1993}). Using powerful but non-rigorous methods from statistical physics, Me\'zard and Parisi \cite{Parisi1987} conjectured that $\expect[X_n]$ has the limiting value $\zeta(2)=\pi^2/6$ as $n \to \infty$. Aldous first proved~\cite{Aldous1992} that $\expect[X_n]$ indeed has a limit, and then~\cite{Aldous2001} proved the $\pi^2/6$ conjecture, using the local weak convergence approach we follow here. 

Other methods have been introduced to study this problem~\cite{Nair2003,Linusson2004,Wastlund2009}, including the marvelous fact that for finite $n$, the expected cost of the minimum matching is the sum of the first $n$ terms of the Riemann series for $\zeta(2)$, namely $1+1/4+1/9+\cdots+1/n^2$. But these methods rely heavily on the specifics of the matching problem, and we will not discuss them here.

As the first step in applying local weak convergence to the planted problem, we are going to identify the weak limit of the planted model according to Definition~\ref{def:random-weak-limit}: that is, the kind of infinite randomly weighted tree that corresponds to $K_{n,n}$ with weights drawn from our model. To be more precise, we are interested in a probability measure $\mu_\infty \in \mathcal{P}(\Gstar)$ that 
$\mu_n = \expect U(N_n) \in \mathcal{P}(\Gstar)$ 
converges to in the local weak sense, where $N_n = (K_{n,n},\ell_n)$ is the planted model, $U(N_n)$ is the random measure defined in~\eqref{eq:local-uniform} by rooting $N_n$ at a uniformly random vertex, and $\expect U(N_n)$ is the measure defined in~\eqref{eq:local-uniform-expectation}. Since every neighborhood has the same distribution of neighborhoods in the planted model, the root might as well be at vertex $1$, so $\mu_n$ is simply the distribution of $[N_{n,\circ}(1)]$. Thus
\begin{equation}
\label{eq:rootis1}
\begin{aligned}
\mu_n(A) &= 
\expect U(N_n) (A) \\ 
&= \frac{1}{2n} \sum_{v\in V_n} \expect[ \delta_{[N_n(v)]}(A)]
= \prob[[N_{n,\circ}(1)] \in A], \qquad \text{for all Borel sets $A \subseteq \Gstar$.}
\end{aligned}
\end{equation}

In the un-planted model studied by Aldous and others, the weak limit of the random matching problem is the Poisson Weighted Infinite Tree (PWIT). The planted case is similar but more elaborate: the weights of the un-planted edges are Poisson arrivals, but the weights of the planted edges have to be treated separately. We call this the \emph{planted PWIT}, and define it as follows.

We label the vertices $V_\infty$ of the planted PWIT with sequences over $\mathbb{N}_0$, which we denote with bold letters. The root is labeled by the empty sequence $\root$. The children of a vertex $\bs{i} = (i_1,i_2,\cdots,i_t)$, are $\concat{\bs{i}}{j} \coloneqq (i_1,i_2,\cdots,i_t,j)$ for some $j \in \mathbb{N}_0$, and if $t > 0$ its parent is $\parent(\bs{i}) := (i_1,i_2,\cdots,i_{t-1})$. We say that $\bs{i}$ belongs to the $t$th generation of the tree, and write $\gen(\bs{i})=t$.

Appending $j \in \{1, 2, \ldots\}$ to $\bs{i}$ gives the $j$th non-planted child, i.e., the child with the $j$th smallest edge weight among the non-planted edges descending from the parent $\bs{i}$. However, appending $j=0$ indicates $\bs{i}$'s planted child if any, i.e., $\bs{i}$'s partner in the planted matching if its partner is one of its children instead of its parent. Since the planted partner of a planted child is its parent, these sequences never have two consecutive zeroes. (Note that the root has a planted child, so the first entry in the sequence is allowed to be $0$.) We denote the set of such sequences of length $t$ as $\Sigma^t$, and the set of all finite such sequences as $\Sigma = \bigcup_{t \in \mathbb{N}_0} \Sigma^t$. 
Thus the edge set is $E_\infty = \{ \{\bs{i}, \concat{\bs{i}}{j} \} \mid j \in \mathbb{N}_0 \text{ and } \bs{i}, \concat{\bs{i}}{j} \in \Sigma \}$, 
and the planted matching  $\Mplanted_\infty \subset E_\infty$ consists of the edges $\{ \{\bs{i}, \concat{\bs{i}}{0} \} \mid \bs{i}, \concat{\bs{i}}{0} \in \Sigma \}$. Let $T_\infty = (V_\infty = \Sigma, E_\infty, \Mplanted_\infty)$ be the resulting planted tree.


Next we define the random edge weights $\ell_\infty: E_\infty \to \rplus$. The weights of the un-planted edges are distributed just as in the PWIT: that is, for each vertex $\bs{i} \in V_\infty$, the sequence $\big( \ell_\infty(\bs{i},\concat{\bs{i}}{j} ) \big)_{j=1,2,\ldots}$ is distributed jointly as the arrivals $\zeta_1,\zeta_2,\ldots$ of a Poisson process with rate $1$. Then we have the planted edges: if $\concat{\bs{i}}{0} \in \Sigma$, then $\ell_\infty(\bs{i}, \concat{\bs{i}}{0} ) \sim \exp(\lambda)$ independent of everything else. Note that these random weights are independent for different parents $\bs{i} \in V_\infty$.

Finally, let $N_\infty = (T_\infty,\ell_\infty)$ denote the random planted tree and let $N_{\infty,\circ}$ denote the version of $N_\infty$ rooted at $\root$. We call $N_{\infty,\circ}$ the \emph{planted Poisson Weighted Infinite Tree} or the planted PWIT for short. Its structure is shown in Figure~\ref{fig:asympobject}.

As in Section~\ref{app:gen_setup}, let $[N_{\infty,\circ}]$ denote the equivalence class of $N_{\infty,\circ}$ up to rooted isomorphisms, and denote by $\mu_\infty\in\mathcal{P}(\Gstar)$ the probability distribution of $[N_{\infty,\circ}]$ in $\Gstar$. The following theorem shows that $\mu_n$ converges weakly to $\mu_\infty$.


\begin{theorem} \label{thm:lwc}
The planted PWIT is the random weak limit of the planted model on $K_{n,n}$, i.e,  $\mu_n \xrightarrow{w} \mu_\infty$.
\end{theorem}

\begin{proof}[Sketch of the proof]
Similar to the un-planted case \cite[Lemmas 10 and 11]{Aldous1992} the proof follows from the following steps: 
\begin{enumerate}
\item Recall that $\mu_n$ is the distribution of $[N_{n,\circ}(1)]$. We define an exploration process that explores the vertices of $N_{n,\circ}(1)$ starting from the root vertex $1$ in a series of stages. At stage $m$, this process reveals a tree of depth $m+1$ and maximum arity $m+1$, where the children of each vertex are its $m$ lightest un-planted neighbors (among the remaining vertices) and possibly its planted partner (if its planted partner is not its parent). 
%
\item In the limit $n \to \infty$, the tree explored at each stage is asymptotically the same as a truncated version of the planted PWIT, i.e., the analogous stage-$m$ neighborhood of the root $\root$. 
\item For large enough $m$ (independent of $n$),   the $\rho$-neighborhood $(G_{n,\circ}(1))_\rho$ of vertex $1$ in $N_{n,\circ}(1)$ is a subgraph of the explored tree at stage $m$ of the process with high probability. This is due to the fact that, while $K_{n,n}$ has plenty of cycles that are topologically short, it is very unlikely that any short cycle containing vertex $1$ consists entirely of low-weight edges. 
\item Finally, the result follows by using the Portmanteau Theorem, which enables us to extend the convergence of distributions on local neighborhoods in total variation distance to the desired local weak convergence.
\end{enumerate}
The complete proof is presented in Appendix~\ref{app:proofoflwc}. 
\end{proof}

Since the planted model on $K_{n,n}$ converges to the planted PWIT, we have every reason to believe that---just as Aldous showed for the un-planted problem---the minimum matching on the planted model converges locally weakly to the minimum involution invariant matching on the planted PWIT. We make this statement rigorous in the following sections, following and generalizing arguments in~\cite{Aldous1992,Aldous2001,Aldous2004}.

%% file: Sections/OptMatch.tex
In this section we define the optimal involution invariant random matching $\Minfopt$ on the planted PWIT---or more precisely, the joint distribution $(\ell_\infty,\Minfopt)$. We define it in terms of fixed points of a message-passing algorithm, construct it rigorously on the infinite tree, and prove that it is optimal and unique.

Since the planted PWIT is an infinite tree, the total weight of any matching is infinite. This makes it unclear whether there is a well-defined notion of a minimum-weight matching. But since we are ultimately interested in the cost per vertex of the minimum matching on $K_{n,n}$,  
we call a random matching $(\ell_{\infty},\M_{\infty})$ on the planted PWIT optimal if it minimizes the expected cost of the edge incident to the root, $\expect[\ell_{\infty}(\root,\M_{\infty}(\root))]$. 

However, since ${\mu}_n$ is involution invariant and involution invariance is preserved under weak limit, we need to restrict our search for minimum matching to involution invariant matchings. This restriction is crucial. For instance, if we simply want to minimize the expected cost at the root, we could construct a matching as follows, akin to a greedy algorithm: first match the root to its lightest child, i.e., the one with the lowest edge weight. Then match each of its other children with their lightest child, and so on. For this matching, $\expect[\ell_{\infty}(\root,\M_{\infty}(\root))]= \expect[\min(\eta,\zeta)] = 1/(1+\lambda)$ where $\eta \sim \exp(\lambda)$ is the weight of the root's planted edge and $\zeta \sim \exp(1)$ is the weight of its lightest un-planted edge. 

However, as pointed out by Aldous for the un-planted model~\cite[Section 5.1]{Aldous2001}, this matching is not involution invariant. For instance, suppose $1$ is $\root$'s lightest child, but that $1$ has a descending edge whose weight is even less. In this case, if we swap $\root$ and $1$, we won't include the edge $\{1,\root\}$ in the resulting matching. Indeed, in the un-planted case the optimal involution invariant matching has expected weight $\pi^2/6$ per vertex, while this greedy matching has expected weight $1$. The lesson here is that the only matchings on the PWIT (or the planted PWIT) that correspond to genuine matchings on $K_{n,n}$ are those that are involution invariant.

Before we proceed, we make a small increment to our formalism. For a network $N$ we define $\mathcal{M}[N]$ as the set of all matchings on $N$. Now, a \emph{random matching} $(\ell,\M)$ on $N$ is a joint distribution of edge weights and matchings, i.e., a probability measure on $\rplus^E \times \mathcal{M}[N]$ with marginal $\ell$ on $\rplus^E$. Intuitively, the reader would probably interpret the phrase ``random matching'' as a measurable function from $\rplus^E \to \mathcal{M}[N]$, assigning a distribution of matchings to each realization of the edge weights $\ell$. However, here we follow Aldous by using it to mean a distribution over both $\ell$ and $\M$. Note that $\M$ may have additional randomness even after conditioning on $\ell$; we will eventually learn, however, that $\Minfopt$ does not.

%% file: Sections/OptMatch-Heuristic.tex
We start by describing a message-passing algorithm on the planted PWIT that we will use to define $\Minfopt$. We have already discussed this, but we do it here in our notation for the infinite tree.

If $(\ell_\infty,\M)$ is involution invariant, $\expect[\ell_\infty(v,\M(v))]$ is independent of the choice of $v \in V_\infty$. 
Let us pretend for now that the total weight of the minimum involution invariant matching $\Minfopt$ is finite, and minimize it with a kind of message-passing algorithm.

For a vertex $v \in V_\infty$, let $T_\infty(v)$ denote the subtree consisting of $v$ and its descendants, rooted at $v$ (in particular, $T_\infty(\root) = T_\infty$). Let $\ell_\infty(T_\infty(v))$ and $\ell_\infty(T_\infty(v) \setminus \{v\})$ denote the total weight of the minimum involution invariant matching on $T_\infty(v)$ and $T_\infty(v) \setminus \{v\}$ respectively. The difference between these, which we denote
\begin{equation}
\label{eq:xv-def}
X_v = \ell_\infty(T_\infty(v)) - \ell_\infty(T_\infty(v) \setminus \{v\}) \, , 
\end{equation}
is the cost of matching $v$ with one of its children, as opposed to leaving it unmatched (or rather matching it with its parent, without including the cost of that edge). This is the difference between two infinite quantities, but as Aldous and Steele say~\cite{Aldous2004} we should ``continue in the brave tradition of physical scientists'' and see where it leads. While we have already seen the resulting RDEs in the proof of Theorem \ref{thm:main}, it will be helpful to restate them here in this more precise notation.

Suppose that in a realization of $\Minfopt$, $\root$ is matched with its child $i$. Then we have
\begin{align*}
\ell_\infty(T_\infty(\root)) &= \ell_\infty(\root,i) + \ell_\infty(T_\infty(i) \setminus \{i\}) + \sum_{ j \neq i}\ell_\infty(T_\infty(j) )\allowdisplaybreaks\\
&= \ell_\infty(\root,i) + \ell_\infty(T_\infty(i) \setminus \{i\}) + \ell_\infty(T_\infty(\root) \setminus\{\root\}) - \ell_\infty(T_\infty(i)).
\end{align*}
Rearranging and using~\eqref{eq:xv-def}, we have
\begin{align*}
X_\root = \ell_\infty(\root,i) - X_i \, . 
\end{align*}
We can read this as follows: by matching $\root$ with its child $i$, we pay the weight $\ell_\infty(\root,i)$ of the edge between them, but avoid the cost $X_i$ of having $i$ matching with one of its own children. But of course we want to match $\root$ with whichever child minimizes this cost, giving
\begin{equation}
\label{eq:xroot-min}
X_\root = \min_{j \ge 0} \left( \ell_\infty(\root,j) - X_j \right) \, . 
\end{equation}
Using the same argument, this relation holds for any vertex $v \in V_\infty$. Recalling that the children of $v$ are labeled $\concat{v}{j}$ (i.e., $v$'s label sequence with $j$ appended) for $j \ge 0$, we have
\begin{equation}
\label{eq:xvmin}
X_v = \min_{j \geq 0, \concat{v}{j} \in V_\infty} \left( \ell_\infty(v, \concat{v}{j} ) - X_{ \concat{v}{j} } \right) \, .
\end{equation}

Now recall that $v$'s planted partner is either its parent or its $0$th child. If the former, then this minimization ranges over $v$'s un-planted children $\concat{v}{j}$ for $j \ge 1$. If the latter, then it also includes $v$'s planted child $\concat{v}{0}$. Let us assume that $X_v$ is drawn from one of two distributions over $\mathbb{R}$, and denote this random variable $X$ in the first case and $Y$ in the second case. We expect these distributions to be fixed if we draw $X_{\concat{v}{j}}$ independently for each $j$, and obtain $X_v$ by applying~\eqref{eq:xvmin}. Since $v$'s un-planted children have planted children, but $v$'s planted child (if any) only has un-planted children, we get
the following recursive distributional equations (RDEs):
\begin{align}
X&\stackrel{d}{=} \min(\{\zeta_i - Y_i\}_{i=0}^\infty),\label{eq:disteqX} \\
Y&\stackrel{d}{=} \min(\eta - X,\{\zeta_i - Y_i\}_{i=0}^\infty)\label{eq:disteqY}
\end{align}
where $X$ is independent of everything else, $\{Y_i\}_{i=1}^\infty$ and $Y$ are  i.i.d.\, and $\{\zeta_i\}_{i=1}^\infty$ are the arrivals of a Poisson process with rate $1$, and $\eta \sim \exp(\lambda)$ is the weight $\ell(v,\concat{v}{0})$ of the planted edge---these are the edge weights of the planted PWIT described in Section~\ref{app:planted-PWIT}. 

As we saw in Section \ref{sec:main-infinite-tree}, the distributional equations~\eqref{eq:disteqX}--\eqref{eq:disteqY} have a unique fixed point supported on $\mathbb{R}$ whenever $\lambda < 4$. 
Our next task is to turn this heuristic derivation into a rigorous construction of random variables on the planted PWIT, and use them to construct the minimum involution invariant random matching $\Minfopt$.

%% file: Sections/OptMatch-Construct.tex
The construction is similar to the one in the un-planted model (see~\cite[Section 4.3]{Aldous2001} and \cite[Section 5.6]{Aldous2004}). We draw random variables $X$ from a fixed point of the system of recursive distributional equations \eqref{eq:disteqX}--\eqref{eq:disteqY}. Then we show that these random variables generate an involution invariant random matching, by constructing it (randomly) on finite neighborhoods, and then extending it to the infinite tree. In the next subsection, we analyze this matching and show that it is optimal.
 
Define the set of directed edges $\overleftrightarrow{E}_\infty = \overrightarrow{E}_\infty\cup \overleftarrow{E}_\infty$ of $T_\infty$ by assigning two directions to each edge $e\in E_\infty$: for an edge $e=\{\bs{i}, \concat{\bs{i}}{j}  \}$ let $\overrightarrow{e} = (\bs{i}, \concat{\bs{i}}{j} ) \in \overrightarrow{E}_\infty$ denote the edge directed downward, i.e., away from the root, and let $\overleftarrow{e} = ( \concat{\bs{i}}{j} ,\bs{i} ) \in \overleftarrow{E}_\infty$ denote the edge directed upward toward the root. We use $\overleftrightarrow{e}$ to denote a typical member of $\overleftrightarrow{E}_\infty$.
We extend the edge weights to $\overleftrightarrow{E}_\infty$, as $\ell_\infty(\overleftarrow{e})=\ell_\infty(\overrightarrow{e})=\ell_\infty(e)$. 

The following lemma shows how to define ``costs'' or ``messages'' on $\overleftrightarrow{E}_\infty$. It is essentially identical to~\cite[Lemma 14]{Aldous2001} and \cite[Lemma 5.8]{Aldous2004}, except that we have different distributions of messages on the planted and un-planted edges.

\begin{lemma}
\label{lem:const}
Let $(X_0,Y_0)$ be a solution of the system of recursive distributional equations~\eqref{eq:disteqX}--\eqref{eq:disteqY}. Jointly with the edge weights $\ell_\infty$, we can construct a random function $X: \overleftrightarrow{E}_\infty \to \mathbb{R}$ 
such that the following holds:
	\begin{enumerate}[label=(\roman*)]
		\item For every edge $(u,v)\in\overleftrightarrow{E}_\infty$ we have
		\begin{align}
		\label{eq:baleq}
		X(u,v) = \min_{(v,w) \in E_\infty, w \neq u} ( \ell_\infty(v,w) - X(v,w) ) \, .
		\end{align}
		\item For every planted edge $e \in \Mplanted_\infty$, $X(\overrightarrow{e})$ and $X(\overleftarrow{e})$ each have the same distribution as $X_0$. 
		\item For every un-planted edge $e \notin \Mplanted_\infty$, $X(\overrightarrow{e})$ and $X(\overleftarrow{e})$ each have the same distribution as $Y_0$.
		\item For every edge $e \in E_\infty$, $X(\overrightarrow{e})$ and $X(\overleftarrow{e})$ are independent.
	\end{enumerate}
\end{lemma}


\begin{proof}
The idea is to construct these random variables on the subtree consisting of all edges up to a given depth $h$. We do this by initially ``seeding'' them on the downward-pointing edges at that depth, drawing their $X$ independently from the appropriate fixed-point distribution. We then use the message-passing algorithm given by~\eqref{eq:baleq} to propagate them through this subtree. As with belief propagation on a tree, this propagation consists of one sweep upward to the root, and then one sweep back downward toward the leaves. Finally, we use the Kolmogorov consistency theorem~\cite[p. 115, Theorem 6.16]{Kallenberg2002} to take the limit $h \to \infty$, extending the distribution on these finite-depth subtrees to $T_\infty$.
	
Formally, let $h \in \mathbb{N}_+$. Let $\overrightarrow{E}_\infty(h)$ and $\overleftarrow{E}_\infty(h)$ respectively denote the set of downward- and upward-directed edges at depth $h-1$, and let $\overleftrightarrow{E}_\infty(\leq h)$ denote the set of all directed edges up to depth $h$:
	\begin{align*}
\overrightarrow{E}_\infty(h) 
	&= \{\overrightarrow{e}=(v, \concat{v}{j} ): \gen(v) = h-1,\, \{v, \concat{v}{j} \} \in E_\infty \} ,\\
\overleftarrow{E}_\infty(h) 
	&= \{\overleftarrow{e}=(\concat{v}{j} , v): \gen(v) = h-1,\, \{v, \concat{v}{j} \} \in E_\infty \} ,\\
\overleftrightarrow{E}_\infty(\leq h) 
	&= \{\overleftrightarrow{e}=(v,w): \gen(v),\,\gen(w)  \leq h,\, \{v,w\} \in E_\infty \}.
	\end{align*}
In particular, $\overrightarrow{E}_\infty(h)$ is the set of downward-pointing edges at the leaves of the subtree of depth $h$, and $\overleftrightarrow{E}_\infty(\leq h)$ is the set of all edges, pointing in both directions, within that subtree (see Figure~\ref{fig:Einf}). Our goal is to define $X$ on $\overleftrightarrow{E}_\infty(\leq h)$.
	
	\begin{figure}[!t]
		\centering
		\begin{subfigure}[!t]{.45\linewidth}
			\begin{center}
				\includegraphics[width=\linewidth]{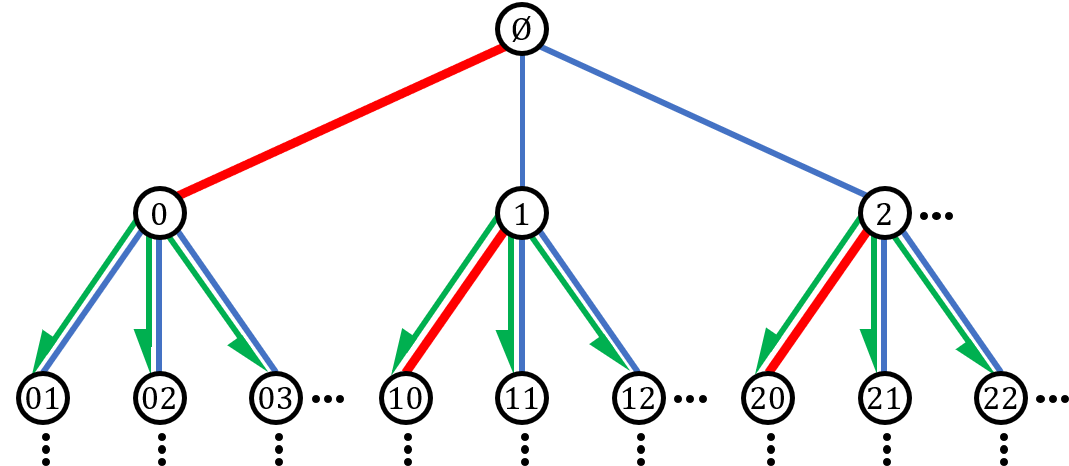}
				\caption*{Green arrows are the members of $\overrightarrow{E}_\infty(2)$.}
			\end{center}
		\end{subfigure}
		\hfill
		\begin{subfigure}[!t]{.45\linewidth}
			\begin{center}
				\includegraphics[width=\linewidth]{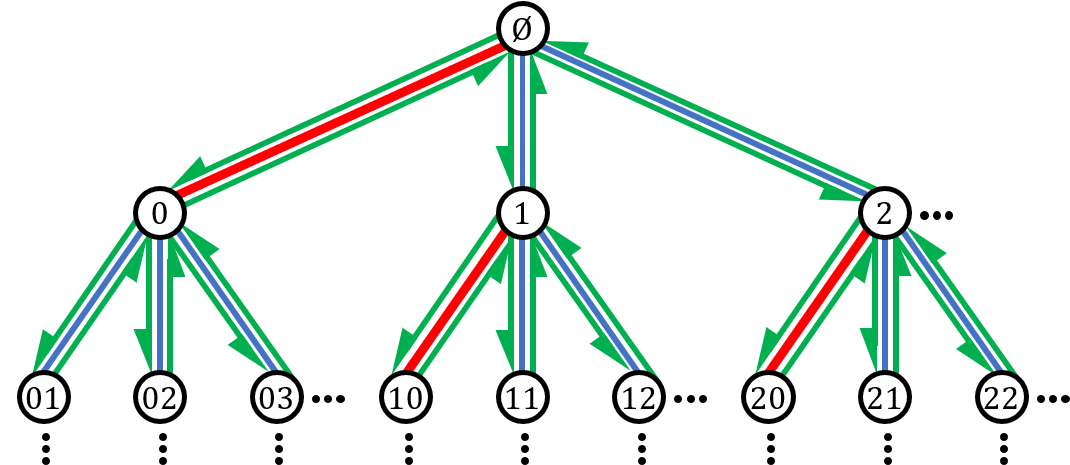}
				\caption*{Green arrows are the members of $\overleftrightarrow{E}_\infty(\leq 2)$. 
				 }
			\end{center}
		\end{subfigure}
		\caption{The sets $\protect\overrightarrow{E}_\infty(2)$ and $\protect\overleftrightarrow{E}_\infty(\leq 2)$. 
		Bold red edges are in $\Mplanted_\infty$ and solid blue edges are in $E_\infty \setminus \Mplanted_\infty$.}
		\label{fig:Einf}
	\end{figure}

To initialize the process, for each $\overrightarrow{e} \in \overrightarrow{E}_\infty(h)$ we assign the random variable $X(\overrightarrow{e})$ by drawing independently from $X_0$ if $e \in \Mplanted_\infty$ and from $Y_0$ if $e \notin \Mplanted_\infty$. 
We then use~\eqref{eq:baleq} recursively to define $\{X(\overrightarrow{e}):\overrightarrow{e} \in \overrightarrow{E}_\infty(k) \}$ for $k\in\{h-1,h-2,\ldots,1\}$. Once we have $X(\overrightarrow{e})$ for all edges incident to the root, we use~\eqref{eq:baleq} to obtain $X(\overleftarrow{e})$ for these edges, i.e., for $\overleftarrow{E}_\infty(1)$. We then move back down the tree, using~\eqref{eq:baleq} at each level to define $\{X(\overleftarrow{e}): \overleftarrow{e} \in \overleftarrow{E}_\infty(k) \}$ for $k\in\{1,2,\ldots,h\}$.

Parts (ii) and (iii) of the lemma follow from the fact that $(X_0,Y_0)$ are fixed points of~\eqref{eq:disteqX}--\eqref{eq:disteqY}. Part (iv) follows from the fact that, for all ${e} \in {E}_\infty(\leq h)$, $X(\overrightarrow{e})$ and $X(\overleftarrow{e})$ are determined by disjoint subsets of $\{X(\overrightarrow{e}): \overrightarrow{e} \in \overrightarrow{E}_\infty(h) \}$ and hence are independent. 

Finally, we extend these random variables to the entire planted PWIT. For each finite depth $h$, the above construction gives a collection of random variables 
	\begin{align*}
	\mathcal{X}_h = \{ \big( \ell_\infty(\overleftrightarrow{e}), X(\overleftrightarrow{e}) \big) : \overleftrightarrow{e}\in \overleftrightarrow{E}_\infty(\leq h) 
	\} \, , 
	\end{align*}
that satisfies (i), (ii), (iii), and (iv). Moreover, the marginal distribution of $\mathcal{X}_{h+1}$ restricted to 
depth $h$ is the same as the distribution of $\mathcal{X}_h$. Now, by the Kolmogorov consistency theorem, there exists a collection of random variables $\mathcal{X}_\infty$ that satisfies (i), (ii), (iii), and (iv), such that the marginal distribution of $\mathcal{X}_\infty$ restricted to depth $h$ is the same as the distribution of $\mathcal{X}_h$.
\end{proof}

One important implication of Lemma~\ref{lem:const} is the following corollary. 

\begin{corollary}
\label{cor:iidX} 
Consider the collection of random variables $\mathcal{X}_\infty$ given by Lemma \ref{lem:const}. 
\begin{enumerate}[label=(\roman*)]
\item Let $e = \{\root,0\}$ denote the planted edge incident to the root. Then $X(\overrightarrow{e})$ and $X(\overleftarrow{e})$ are independent and identically distributed as $X_0$, and are independent of $\ell_\infty(e)$.
\item Suppose we condition on the existence of an un-planted edge $e=\{\root, i\}$ incident to $\root$ with $\ell_\infty(e) = \zeta$. Then $X(\overrightarrow{e})$ and $X(\overleftarrow{e})$ are independent and identically distributed as $Y_0$. 
\end{enumerate}
\end{corollary}

\begin{proof}
Part (i) follows immediately from the construction in Lemma \ref{lem:const}. Part (ii) follows from the fact that if we condition on the existence of a Poisson arrival at time $\zeta$, the other arrivals are jointly distributed according to the same Poisson process. There is a subtlety here in that it is important to condition on $\zeta$ but not on $i$, since knowing where $\zeta$ is in the sorted order of the un-planted weights affects their distribution. On the other hand, if we fix an edge $e$ before doing this sorting, then $X(\overrightarrow{e})$ and $X(\overleftarrow{e})$ are independent of $\ell_\infty(e)$ for both planted and un-planted edges, and we will use this fact below.
\end{proof}

Our next task is to transform the above construction into a random matching $(\ell_\infty,\Minfopt)$. There are two ways we might do this. One would be to define a function on $V_\infty$ that yields a proposed partner $w$ for each vertex $v$. As in~\eqref{eq:baleq}, matching $v$ with $w$ would cost the weight of the edge between them, but remove the cost of having $w$ pair with one of its other neighbors. Minimizing this total cost over all neighbors $w$ (rather than over all but one as in the message-passing algorithm) gives
\begin{equation}
\label{eq:Minfdef}
\Minfopt(v) = \argmin_{w:\{v,w\} \in E_\infty} \left( \ell_\infty(v,w) - X(v,w) \right) \, .
\end{equation}
Since each edge weight $\ell_\infty(v,w)$ is drawn from a continuous distribution, and Corollary~\ref{cor:iidX} implies that it is independent of $X(v,w)$, with probability $1$ the elements of the set we are minimizing over are distinct and this $\argmin$ is well-defined.

Alternately, we could define a mark function on $E_\infty$ as described above, namely the indicator function for the event that an edge $e$ is in the matching. Including $e$ in the matching makes sense if $\ell_\infty(e)$ is less than the cost of matching each of its endpoints to one of their other neighbors. So (abusing notation) this suggests
\begin{align}
\label{eq:Minfalter}
\Minfopt(e) = \begin{cases} 
1 & \text{if } \ell_\infty(e) < X(\overrightarrow{e}) + X(\overleftarrow{e}) \\
0 & \text{otherwise.}
\end{cases}
\end{align}
A priori, there is no guarantee that either of these functions is a matching, or that they agree with each other. The following lemma (which is a reformulation of~\cite[Lemma 5.9]{Aldous2004}) gives the good news that they are, and they do. 


\begin{lemma}
\label{lem:match}
The following are equivalent:
\begin{enumerate}
\item $u = \argmin_{w:\{v,w\} \in E_\infty} \left( \ell_\infty(v,w) - X(v,w) \right)$
\item $v = \argmin_{w:\{u,w\} \in E_\infty} \left( \ell_\infty(u,w) - X(u,w) \right)$
\item $\ell_\infty(u,v) < X(u,v) + X(v,u)$.
\end{enumerate}
Therefore, $u=\Minfopt(v)$ if and only if $v=\Minfopt(u)$ (with $\Minfopt$ defined as in~\eqref{eq:Minfdef}), and these are equivalent to $\Minfopt(u,v)=1$ (with $\Minfopt$ defined as in~\eqref{eq:Minfalter}).
\end{lemma}

\begin{proof}
By~\eqref{eq:baleq}, condition (1) holds if and only if
\[
\ell_\infty(v,u) - X(v,u) < \argmin_{w:\{v,w\} \in E_\infty , w \ne u} \left( \ell_\infty(v,w) - X(v,w) \right) = X(u,v) \, . 
\]
Rearranging gives (3), so (1) and (3) are equivalent. Since (3) is symmetric with respect to swapping $u$ and $v$, (2) and (3) are also equivalent.
\end{proof}
%

Finally, given the symmetric dependency of $\Minfopt(e)$ on the values of $X(\overrightarrow{e})$ and $X(\overleftarrow{e})$, it is intuitive that the random matching $(\ell_\infty,\Minfopt)$ is involution invariant. The following lemma corresponds to~\cite[Lemma 24]{Aldous2001} in the un-planted case, but defining the involutions in a way that preserves the (un)planted edges takes a little more work. We give the proof in Appendix~\ref{app:proofofinvol}.
\begin{proposition}
\label{prop:invol}
	The random matching $(\ell_\infty,\Minfopt)$ is involution invariant.
\end{proposition}

%% file: Sections/OptMatch-Analysis.tex
Now that we have constructed $(\ell_\infty,\Minfopt)$, it is time to prove that $(\ell_\infty,\Minfopt)$ is {\em the} minimum involution invariant random matching. The steps we take to prove this claim are mostly the same as in~\cite[Sections 4.4 and 4.5]{Aldous2001}, but a few details differ in the planted model, so for the sake of completeness and consistency with our notation we give a self-contained proof.

As the first step, we are going to prove that $(\ell_\infty,\Minfopt)$ is a minimum involution invariant matching: that is, it achieves the minimum expected length at the root. We follow the discussion at the beginning of Section 4.5 in~\cite{Aldous2001}.

\begin{proposition}
\label{prop:ineq+eqq}
Let $(\ell_\infty,\M'_\infty)$ be an involution invariant random matching on the planted PWIT. Then $\expect[ \ell_\infty(\root,\M'_\infty(\root))]\geq \expect[ \ell_\infty(\root,\Minfopt(\root))]$.
\end{proposition}
\begin{proof}

Note that in addition to depending on the edge weights $\ell_\infty$, $\M'_\infty$ might also have additional randomness. However, we can always couple $(\ell_\infty,\M'_\infty)$ and $(\ell_\infty,\Minfopt)$ so that if we condition on $\ell_\infty$ then $\M'_\infty$ and $\Minfopt$ are independent. Let $A$ be the event that $\M'_\infty(\root) \neq \Minfopt(\root)$, and assume without loss of generality that $\prob[A] > 0$.
	
Conditioned on $A$, there is a doubly-infinite alternating path that passes through the root $\root$, alternating between edges in $\M'_\infty$ and $\Minfopt$. That is to say, there is a doubly-infinite sequence of distinct vertices $\cdots,v_{-2},v_{-1},v_0,v_1,v_2,\cdots$ where $v_0 = \root$, $v_1 = \Minfopt(\root)$, and $v_{-1} = \M'_\infty(\root)$, and where for all even integers $m$ we have $\Minfopt(v_m) = v_{m+1}$ and $\M'_\infty(v_m) = v_{m-1}$.

By the construction of $\Minfopt$, we know that $v_1$ achieves the minimum in Equation~\eqref{eq:Minfdef}:
\begin{equation}
\label{eq:cris1}
\ell_\infty(v_0,v_1) - X(v_0,v_1) = \min_{w:\{v_0,w\} \in E_\infty} \left( \ell_\infty(v_0,w) - X(v_0,w) \right) \, . 
\end{equation}
We also have the message-passing equation~\eqref{eq:baleq} for $X(v_{-1},v_0)$, 
\begin{equation}
\label{eq:cris2}
X(v_{-1},v_0) = \min_{w: \{v_0,w\} \in E_\infty, w \neq v_{-1}} ( \ell_\infty(v_0,w) - X(v_0,w) ) \, .
\end{equation}
The right-hand sides of~\eqref{eq:cris1} and~\eqref{eq:cris2} are the same except that $v_{-1}$ is excluded in~\eqref{eq:cris2}. But since the minimum is achieved by $v_1$, excluding $v_{-1}$ makes no difference, and the right-hand sides are equal. Rearranging gives
\begin{equation}
\label{eq:cris3}
\ell_\infty(v_0,v_1) = X(v_0,v_1) + X(v_{-1},v_0) \, . 
\end{equation}
On the other hand, \eqref{eq:baleq} also implies $X(u,v) \le \ell_\infty(v,w) - X(v,w)$ for any $u,v,w$ where $u$ and $w$ are distinct neighbors of $v$, and in particular
\begin{equation}
\label{eq:cris-ineq}
X(v_1,v_0) \le \ell_\infty(v_0,v_{-1}) - X(v_0,v_{-1}) \, . 
\end{equation}

Now, using~\eqref{eq:cris3}, the expected difference in the length at the root is 
\begin{align}
\expect[ \left( \ell_\infty(\root,\M'_\infty(\root)) - \ell_\infty(\root,\Minfopt(\root)) \right) ]
&= \expect[( \ell_\infty(v_0,v_{-1}) - \ell_\infty(v_0,v_1) ) \,\one_A ] \nonumber \\
&= \expect[( \ell_\infty(v_0,v_{-1}) - X(v_0,v_1) - X(v_{-1},v_0) ) \,\one_A ] \, . 
\label{eq:cris-align}
\end{align}
Now we use the fact that $\Minfopt$ and $\M'_\infty$ are both involution invariant. There is a subtlety here in that conditioning on $A$ breaks involution invariance, since it requires $\Minfopt$ and $\M'_\infty$ to differ at the root specifically. However, the involutions that swap $v_0$ with $v_1$ or with $v_{-1}$ maintain this conditioning, since $\Minfopt$ and $\M'_\infty$ differ at these vertices as well. It follows that $X(v_0,v_1)$ and $X(v_1,v_0)$ have the same conditional distribution and hence the same conditional expectation, and similarly for $X(v_0,v_{-1})$ and $X(v_{-1},v_0)$. 
Then~\eqref{eq:cris-align} becomes
\begin{equation}
\label{eq:cris-finish}
\expect[ \left( \ell_\infty(\root,\M'_\infty(\root)) - \ell_\infty(\root,\Minfopt(\root)) \right) ]
= \expect[ ( \ell_\infty(v_0,v_{-1}) - X(v_1,v_0) - X(v_0,v_{-1}) ) \,\one_A ]  \, ,
\end{equation}
which is greater than or equal to zero by~\eqref{eq:cris-ineq}.
\end{proof}

Even given Proposition~\ref{prop:ineq+eqq}, it is still possible a priori that there might be a random involution invariant matching $(\ell_\infty,\M'_\infty)$ with the same expected length at the root as $(\ell_\infty,\Minfopt)$. 
If we were simply trying to calculate the expected length of the minimum matching, this would not be an issue. But our object is the overlap, not the length. If there are two minimal matchings with the same length but different overlap, it would not be clear which is the weak limit of the minimum matching on $K_{n,n}$. 

Happily, we can follow a path similar to~\cite[Section 4.4 and 4.5]{Aldous2001} to show that $\Minfopt$ is unique, making the inequality in Proposition~\ref{prop:ineq+eqq} strict. The following is essentially Proposition 18 of~\cite{Aldous2001}.


\begin{proposition}
\label{prop:ineq}
Let $(\ell_\infty,\M'_\infty)$ be an involution invariant random matching on the planted PWIT. If $\prob[\M'_\infty(\root)\! \neq \!\Minfopt(\root)] \!> \!0$ then $\expect[ \ell_\infty(\root,\M'_\infty(\root))] \!>\! \expect[\ell_\infty(\root,\Minfopt(\root))]$.
\end{proposition}

\begin{proof}
For sake of contradiction, assume there is an involution invariant random matching $(\ell_\infty,\M'_\infty)$ such that $\expect[ \ell_\infty(\root,\M'_\infty(\root))] = \expect[ \ell_\infty(\root,\Minfopt(\root))]$. By the proof of Proposition \ref{prop:ineq+eqq}, we have $\expect[D\boldsymbol{1}_A] = 0$ where
\[
D = \ell_\infty(v_0,v_{-1}) - X(v_1,v_0) - X(v_0,v_{-1})  \ge 0 \, , 
\]
and where $A$ is again the event $\{ \M'_\infty(\root) \neq \Minfopt(\root)\}$, and where the inequality $D \ge 0$ is given by Equation~\eqref{eq:cris-ineq}. Therefore, conditioned on $A$, almost surely
\begin{equation}
\label{eq:blek}
X(v_1,v_0) = \ell_\infty(v_0,v_{-1}) - X(v_0,v_{-1}) \, . 
\end{equation}
Now recall that $v_1$ achieves the minimum, over all $w$ in $v_0$'s neighborhood, of $\ell_\infty(v_0,w) - X(v_0,w)$. By Equation~\eqref{eq:baleq}, $X(v_1,v_0)$ is the minimum of this same quantity over all $w \ne v_1$. But this is the second minimum, i.e., the second-smallest value, and~\eqref{eq:blek} implies
\begin{equation}
\label{eq:blik}
	v_{-1} = {\argmin_{i}}^{[2]} ( \ell_\infty(\root,i) - X(\root,i) ) \, ,
\end{equation}
where ${\min}^{[2]}$ denotes the second minimum. Thus the following holds almost surely: either $\M'_\infty$ agrees with $\Minfopt$ at the root, or it matches the root with the second minimum of $\ell_\infty(\root,i) - X(\root,i)$ rather than the minimum. That is, without conditioning on $A$, 
\[
\prob\!\left( \M'_\infty(\root) \in \left\{ {\argmin_{i}} (\ell_\infty(\root,i) - X(\root,i)) \text{  or  } \argmint_{i} (\ell_\infty(\root,i) - X(\root,i)) \right\} \right) = 1 \, . 
\]
Since $(\ell_\infty,\M'_\infty)$ is involution invariant, the same relation holds for each vertex $v \in V_\infty$, i.e.,
\begin{equation}
\label{eq:argmin}
\prob\!\left(\! \M'_\infty(v) \in \!\left\{\! \argmin_{w:\{w,v\} \in E_\infty}\!\! (\ell_\infty(v,w) - X(v,w)) \text{ or } \argmint_{w:\{w,v\} \in E_\infty}\!\!(\ell_\infty(v,w) - X(v,w)) \!\right\} \!\right)\! = 1 \, .
\end{equation}
Thus any matching with the same expected length as $\Minfopt$ must, almost surely at almost all vertices $v$, match $v$ with its best or second-best partner according to $\ell_\infty(v,w) - X(v,w)$.
	
Surprisingly, no involution-invariant matching can choose the second-best partner with nonzero probability. The following proposition shows that~\eqref{eq:argmin} cannot hold unless $\M'_\infty= \Minfopt$ almost surely.

\begin{proposition}[Proposition 20 of~\cite{Aldous2001}]
\label{prop:unique}
The only involution invariant random matching that satisfies~\eqref{eq:argmin} is $\Minfopt$.
\end{proposition}

\begin{proof}
The reader might be wondering why we can't simply assign everyone to their second-best partner. But recall the key fact from Lemma~\ref{lem:match} that if $$\Minfopt(v) = \argmin_{w:\{v,w\}\in E_\infty} (\ell_\infty(v,w) - X(v,w)),$$ then $\Minfopt(\Minfopt(v)) = v$ and $\{ \{ v, \Minfopt(v) \} : v \in V_\infty \}$ is indeed a matching. The problem is that this fact does not generally hold if we replace $\argmin$ with $\argmint$.  

If $\M'_\infty$ and $\Minfopt$ differ anywhere with positive probability, then by involution invariance they differ at the root with positive probability. In that case, as before, there is a doubly-infinite alternating path from the root to infinity. Thus once $\M'_\infty$ matches the root with its second-best partner, it must keep doing this forever on that path. But in order for $\M'_\infty$ to be involution invariant, it must make the same choices if we follow the path in reverse, and so each vertex on this path must be the second-best partner of its second-best partner. We will see that the probability that this is true on every step of the path, all the way to infinity, is zero.

Let $\cdots$$,v_{-2}$$,v_{-1}$$,v_0$$,v_1$$,v_2$$,\cdots$ be the alternating path defined as follows. First let $v_0=\root$. To define $v_t$ for $t > 0$, we extend the path by alternately apply the best and second-best rules,
\[
v_{t+1} = \begin{cases}
\argmin_{u: \{ u,v_t \} \in E_\infty} (\ell_\infty(v_t,u) - X(v_t,u))  & \text{if $t$ is even} \\
\argmint_{u: \{ u,v_t \} \in E_\infty} (\ell_\infty(v_t,u) - X(v_t,u))  & \text{if $t$ is odd} 
\end{cases}
\]
Similarly, for $t < 0$ we extend the path backwards, 
\[
v_{t-1} = \begin{cases}
\argmint_{u: \{ u,v_t \} \in E_\infty} (\ell_\infty(v_t,u) - X(v_t,u))  & \text{if $t$ is even} \\
\argmin_{u: \{ u,v_t \} \in E_\infty} (\ell_\infty(v_t,u) - X(v_t,u))  & \text{if $t$ is odd} 
\end{cases}
\]
In particular, $v_{1} = \Minfopt(\root)$ and (if $A$ holds) $v_{-1} = \M'_\infty(\root)$.

		
Now for each odd integer $t$, define the event $B_t$ that $v_t$ and $v_{t+1}$ are the second-best partners of each other. For odd $t > 0$ we can write
\[
	B_t = \left\{ v_t = \argmint_{u: \{ u,v_{t+1} \} \in E_\infty} ( \ell_\infty(v_{t+1},u) - X(v_{t+1},u) ) \right \} \, ,
\]
and for odd $t < 0$,  
\[
	B_t = \left\{ v_{t+1} = \argmint_{u: \{ u,v_t \} \in E_\infty} ( \ell_\infty(v_t,u) - X(v_t,u) ) \right \} \, .
\]
As discussed above, since $\M'_\infty$ is involution invariant $A$ implies $B_t$, in particular, for all $t=1,3,5,\ldots$. Thus
\[
	A \subset \xbar{B}_\infty \coloneqq \bigcap_{t=1,3,5,\ldots}^\infty B_t \, . 
\]
Writing $\xbar{B}_t = \bigcap_{t'=1,3,5,\ldots}^t B_{t'}$, this implies
\[
\prob[A] 
\le \prob[\xbar{B}_\infty]
= \prod_{t=1,3,5,\ldots} \prob[B_{t+2} \mid \xbar{B}_t]
= \prod_{t=1,3,5,\ldots} \frac{\prob[\xbar{B}_{t+2}]}{\prob[\xbar{B}_t]} \, . 
\]
and so
\begin{equation}
\label{eq:Bqcriteria}
\text{if $\prob[A] > 0$ then $\lim_{t \to \infty, \,\text{$t$ odd}} \frac{\prob[\xbar{B}_{t+2}]}{\prob[\xbar{B}_t]} = 1$.}
\end{equation}

Now we use involution invariance again. If we root the planted PWIT at $v_2$ instead of $v_0$, sliding the alternating path two steps to the left, the event $\xbar{B}_{t+2}$ becomes the event $B_{-1} \cap \xbar{B}_t$ (and $A$ still holds). By involution invariance the probability of these two events is the same, so
\[
\frac{\prob[\xbar{B}_{t+2}]}{\prob[\xbar{B}_t]}
= \frac{\prob[B_{-1} \cap \xbar{B}_t]}{\prob[\xbar{B}_t]}
= \prob[B_{-1} \mid \xbar{B}_t] \, .
\]
By continuity of probability measure, if $\prob[\xbar{B}_\infty] > 0$ --- which holds if $\prob[A] > 0$ --- we also have
\[
\lim_{t \to \infty} \prob[B_{-1} \cap \xbar{B}_t] 
= \prob[B_{-1} \cap \xbar{B}_\infty]
\quad \text{and} \quad 
\lim_{t \to \infty} \prob[\xbar{B}_t] = \prob[\xbar{B}_\infty] \, , 
\]
in which case 
\[
\lim_{t \to \infty} \prob[B_{-1} \mid \xbar{B}_t]
= \prob[B_{-1} \mid \xbar{B}_\infty] \, . 
\]
Thus~\eqref{eq:Bqcriteria} demands that this conditional probability is $1$. But the following lemma, which generalizes Lemma 22 of~\cite{Aldous2001} to the planted case, shows that this is not so.

\begin{lemma}
\label{lem:BB}
	If $\prob[\xbar{B}_\infty] > 0$ then $\prob[B_{-1} \mid \xbar{B}_\infty] < 1$.
\end{lemma}

\begin{proof}
As in~\cite[Remark on p. 402]{Aldous2001}, the idea is that $B_{-1}$ only depends on what happens on the ``leftward'' branch of the alternating path, $v_0, v_{-1}, v_{-2}, \ldots$, while $\xbar{B}_\infty$ depends only on the ``rightward'' branch $v_0, v_1, v_2, \ldots$  For the details, see Appendix~\ref{app:proofofBB}.
\end{proof}

\noindent
Lemma~\ref{lem:BB} implies that $P(A)=0$, and by the discussion above that $\M'_\infty=\Minfopt$ almost surely. This completes the proof of Proposition~\ref{prop:unique}\ldots
\end{proof}

\noindent \ldots which completes the proof of Proposition~\ref{prop:ineq}.
\end{proof}

An immediate corollary of Proposition~\ref{prop:ineq} is the following.

\begin{corollary}\label{cor:Mopuniq}
In the minimum involution invariant random matching $(\ell_\infty,\Minfopt)$, $\Minfopt$ is a function of the edge lengths $\ell_\infty$. That is to say, given a realization of $(\ell_\infty(e),e\in E_\infty)$, $\Minfopt$ is a fixed matching on the planted PWIT.
\end{corollary}

\begin{proof}
Consider a coupling $(\ell_\infty,\Minfopt,\M'_{\infty,\opt})$ such that conditioned on $(\ell_\infty(e),e\in E_\infty)$, $\Minfopt$ and $\M'_{\infty,\opt}$ are i.i.d..  Then, by Proposition~\ref{prop:ineq} we have $\Minfopt=\M'_{\infty,\opt}$ almost surely.
\end{proof}

In other words, $\Minfopt$ does not have any additional randomness besides its dependence on $\ell_\infty$. This was left as an open question for the un-planted case in~\cite[Remark (d)]{Aldous2001}, although we claim that that paper in fact resolved it! As later stated in~\cite{Aldous2004}, this implies that if we use the construction of Section~\ref{app:OptMatch-Construct} to define random variables $X$ on neighborhoods of depth $h$, then (conditioning on $\ell_\infty$) the random matching defined by these variables becomes concentrated around a single matching as $h \to \infty$. 

This does not quite imply that the messages $X$ on the directed edges of the planted PWIT are determined by $\ell_\infty$. This was shown for the un-planted case by Bandyopadhyay using the concept of endogeny~\cite{Bandyopadhyay2002}. We believe endogeny holds for the planted case, but we leave this as an open question. In any case, as long as the system of recursive distributional equations~\eqref{eq:disteqX}--\eqref{eq:disteqY} has a solution supported on $\mathbb{R}$, whether it is unique or not, the minimum involution invariant random matching $\Minfopt$ is uniquely defined. Therefore, whenever we focus on a realization of $\ell_\infty$, there is no need to call $\Minfopt$ a random matching.


%% file: Sections/OptMatch-RDEUnique.tex
Recall from Section~\ref{app:OptMatch-Construct} that $\Minfopt$ is defined by drawing messages at the boundary of neighborhoods of increasing size from a fixed point of the RDEs \eqref{eq:disteqX}--\eqref{eq:disteqY}, propagating these messages throughout the neighborhood, and then including edges $(u,v)$ whose weights $\ell_\infty(u,v)$ are less than the sum of their messages $X(u,v)+X(v,u)$. 

However, Corollary \ref{cor:Mopuniq} shows that $\Minfopt$ is a function of the weights $\ell_\infty$. As we commented there, this doesn't quite imply that the messages $X$ are also functions of $\ell_\infty$. However, Corollary \ref{cor:Mopuniq} imposes strong conditions on the possible solutions of the RDEs. Specifically, if the RDEs have more than one solution, then each one must somehow result in the same matching $\Minfopt$ given the edge weights. In this section, we show that this implies that the fixed point is indeed unique. This provides an interesting counterpart to the dynamical proof of uniqueness given in~Theorem \ref{thm:ode_unique_solution}.

First we show that any solution has a well-defined moment generating function in a neighborhood of the origin.

\begin{lemma}
\label{lem:mgf}
Let $(X,Y)$ be a solution of the system of recursive distributional equations~\eqref{eq:disteqX}--\eqref{eq:disteqY} supported on $\mathbb{R}$. Then the random variable $X$ has a finite moment generating function $\expect[\e^{\mu X}]$ for $\mu$ in an open neighborhood of $0$.
\end{lemma}

\begin{proof}
Recall that $F_X$ and $F_Y$ denote the cumulative distribution functions of $X$ and $Y$ respectively, and $\bar{F}_X$ and $\bar{F}_Y$ denote their complements.  
On the one hand, by~\eqref{eq:fbarx}, for all $x>0$ we have
\[
	\bar{F}_X(x) = \exp\!\left( - \int_{z=-x}^\infty \bar{F}_Y(z) \,\dz \right) \leq \exp\!\left(-x \bar{F}_Y(0) \right).
\]
On the other hand, for every $x_0 > 0$, Lemma \ref{lem:rde-ode} gives
\[
	f_X(x_0) 
	= \bar{F}_X(x_0) \bar{F}_X(-x_0) \,\expect[F_X(\eta+x_0)] 
	\geq \bar{F}_X(-x_0) \bar{F}_X(x_0) \,\expect[F_X(\eta-x_0)] = f_X(-x_0)\, ,
\]
where the inequality follows by the fact that $F_X(\eta-x_0) \leq F_X(\eta+x_0)$ for all $\eta$. Hence,
\[
	 \prob[X < -x]  \leq \prob[X > x] \leq \exp\!\left(-x \bar{F}_Y(0) \right) \, ,
\]
and $\bar{F}_X(0) \ge 1/2$. Then~\eqref{eq:fbary} implies
\[
\bar{F}_Y(0) = \bar{F}_X(0) \,\expect[F_X(\eta)] \ge \frac{1}{2} \,\expect[F_X(\eta)] > 0 \, .
\]
where the last inequality holds because $\eta$ can be arbitrarily large and $X$ is supported on $\mathbb{R}$. 
The result now follows by simple algebra. If $0 \le \mu < \bar{F}_Y(0)$ we have
\begin{align*}
	\expect[\e^{\mu X}] 
	&= \int_0^\infty \prob[\e^{\mu X} > s] \,\ds
	\leq 1 + \int_1^\infty \prob[\e^{\mu X} > s] \,\ds \\
	&= 1 + \int_1^\infty \prob\!\left[ X > \frac{\ln s}{\mu} \right] \,\ds 
	\le 1 + \int_1^\infty s^{-\bar{F}_Y(0) / \mu} \,\ds 
	< \infty \, , 
\end{align*}
and the proof for $-\bar{F}_Y(0) < \mu \le 0$ is similar. Hence $\expect[\e^{\mu X}] < \infty$ for $\mu \in (-\bar{F}_Y(0), \bar{F}_Y(0))$.
\end{proof}

Now recall that by Lemmas~\ref{lem:const} and Lemma~\ref{lem:match}, given $\ell(u,v) = x$, the probability that $(u,v) \in \Minfopt$ equals $\prob[X+X' > x]$ where $X=X(u,v)$ and $X'=X(v,u)$ are  i.i.d.\ copies of the random variable $X$. If the RDEs have two distinct solutions $(X_1,Y_1)$ and $(X_2,Y_2)$, Corollary \ref{cor:Mopuniq} implies that $\prob[X_1+X'_1 > x] = \prob[X_2+X'_2 > x]$ for all $x$, so that $X_1+X'_1$ and $X_2+X'_2$ have the same distribution. But since 
\[
\expect[\e^{\mu(X_1+X'_1)}] = \left( \expect[\e^{\mu X_1}] \right)^2 \, , 
\]
and similarly for $X_2$, this implies that $X_1$ and $X_2$ have the same moment generating function, which by Lemm \ref{lem:mgf} is well-defined in a neighborhood of the origin. It follows that $X_1$ and $X_2$ have the same distribution~\cite[Theorem 1]{Curtiss1942}. Using \eqref{eq:RDE1}, $Y_1$ and $Y_2$ are equidistributed as well, and we have proved the following theorem:


\begin{theorem}
\label{thm:unique-from-pwit}
Assume the system of recursive distributional equations~\eqref{eq:disteqX}--\eqref{eq:disteqY} has a solution supported on $\mathbb{R}$. Then such solution is unique. 
\end{theorem}

%% file: Sections/JointConv.tex
At this point we have constructed $(\ell_\infty,\Minfopt)$ and shown that it is the unique involution invariant matching on the planted PWIT that minimizes the weight at the root. It is finally time to show that the minimum matching $(\ell_n,\Mnopt)$ on our original planted model on $K_{n,n}$ converges to $(\ell_\infty,\Minfopt)$ in the local weak sense. This implies that these two objects have the same joint distribution of edge weights, and which edges they include in the matching, on neighborhoods of any finite radius. In particular, they have the same expected overlap---which is the overlap we computed in Section \ref{sec:main-infinite-tree}. Thus we finally complete the proof of Theorem~\ref{thm:main}.

To use the framework of local weak convergence to study minimum matchings, we append $\{0,1\}$ to the edges of planted networks in $\Gstar$. In a slight abuse of terminology, we add a $\widehat{\text{hat}}$ and also call $\widehat{N} = (G,\ell,\M)$ a planted network where $\ell$ is the weight function and $\M: E \to \{0,1\}$ is a mark function (which may or may not be a matching). This lets us discuss the joint distribution of edge weights, the planted matching, and the minimum matching of vertex neighborhoods in either model.

In particular, let $\widehat{\mu}_n$ be this distribution in the vicinity of a uniformly random vertex in the finite model,  
\[
\widehat{\mu}_n = \mathbb{E} U(\widehat{N}_{n}) 
\quad \text{is the law of} \quad 
[\widehat{N}_{n,\circ}(1)]
\quad \text{where} \quad
\widehat{N}_{n} = (K_{n,n}, \ell_n, \Mnopt) \, , 
\]
and let $\widehat{\mu}_\infty$ be the analogous distribution at the root of the planted PWIT, 
\[
\widehat{\mu}_\infty 
\quad \text{is the law of} \quad 
[\widehat{N}_{\infty,\circ}]
\quad \text{where} \quad 
\widehat{N}_{\infty} = (T_{\infty}, \ell_\infty,\Minfopt) \, . 
\]
We will show that $\widehat{\mu}_\infty$ is the weak limit of $\widehat{\mu}_n$. Thus the two models have all the same local statistical properties, including their expected weight and overlap.


The proof consists of two main steps, namely, the easier half and the harder half. 
In the easier half, using a simple compactness argument we prove that any subsequence of probability measures $\widehat{\mu}_n$ has a subsequence that converges to an involution invariant random matching on planted PWIT. Using Skorokhod's theorem this shows that the weight of the minimum matching on $K_{n,n}$ is at least that of $\Minfopt$:
\begin{equation}
\label{eq:easyhalf}
\liminf_{n\to\infty}\expect[\ell_n(1,\Mnopt(1))] \geq \expect[\ell_\infty(\root,\Minfopt(\root))] \, .
\end{equation}
Informally, this follows by contradiction. If $\liminf_n \ell_n(1,\Mnopt(1))$ were smaller than this, then the subsequence of sizes $n$ on which it converges to that smaller value would itself have a subsequence that convergences to an involution invariant matching on the planted PWIT with that weight\ldots but this would contradict the optimality of $\Minfopt$.

In the harder half, using $(\ell_\infty,\Minfopt)$ we follow the strategy of~\cite{Aldous1992}. First we construct an almost-perfect matching on $(K_{n,n},\ell_n)$ with weight close to $\expect[\ell_\infty(\root,\Minfopt(\root))]$. Then, we fix this almost-perfect matching matching to make a perfect matching on $(K_{n,n},\ell_n)$ without changing the weight too much. This proves that
\begin{equation}
\label{eq:hardhalf}
\limsup_{n\to\infty}\expect[\ell_n(1,\Mnopt(1))] \leq \expect[\ell_\infty(\root,\Minfopt(\root))] \, .
\end{equation}
Combining~\eqref{eq:easyhalf} and~\eqref{eq:hardhalf}, we have 
\begin{equation}
\label{eq:bothhalves}
\lim_{n \to \infty}\expect[\ell_n(1,\Mnopt(1))] = \expect[\ell_\infty(\root,\Minfopt(\root))] \, .
\end{equation}
As in Aldous' proof of the $\zeta(2)$ conjecture for the un-planted model, this establishes the expected weight of the minimum matching in the planted model. But much more is true. Since $\Minfopt$ is unique, we get the following theorem.

\begin{theorem}
\label{thm:itallworks}
The random weak limit of $(K_{n,n}, \ell_n, \Mnopt)$ is $(T_\infty, \ell_\infty, \Minfopt)$, i.e., $\widehat{\mu}_n \xrightarrow{w} \widehat{\mu}_\infty$. In particular, their expected overlap is equal to 
	\begin{align*}
	\alpha(\lambda) &\coloneqq \lim_{n\to\infty} \frac{1}{n} \,\expect[|\Mnopt \cap \Mplanted_n|]\\
	&= \lim_{n\to\infty} \prob( \{1,1'\} \in \Mplanted_n ) = \prob( \{ \root, 0 \} \in \Minfopt ) \, .
	\end{align*}
\end{theorem}
\begin{proof}
In the easy-half proof, we show that every subsequence of $\widehat{\mu}_n$ has a further subsequence that converges to an involution random matching on the planted PWIT (see next subsection). Now, by \eqref{eq:bothhalves} and Proposition \ref{prop:ineq}, every subsequence of $\widehat{\mu}_n$ converges to $\widehat{\mu}_\infty$; so does the whole sequence. Hence, the random weak limit of $(K_{n,n}, \ell_n, \Mnopt)$ is $(T_\infty, \ell_\infty, \Minfopt)$. 
\end{proof}

Finally, using Corollary \ref{cor:iidX} we have for the expected overlap 
\begin{equation}
\label{eq:mininf-overlap}
\prob( \{ \root, 0 \} \in \Minfopt ) = \prob( X + \widehat{X} > \ell_{\infty}(\root,0) ) \, ,
\end{equation}
where $X$ and $\widehat{X}$ are independent copies of $X_0$ and $(X_0,Y_0)$ is the unique solution of the system of recursive distributional equations~\eqref{eq:disteqX}--\eqref{eq:disteqY}. But this computation is exactly what we have done in Section \ref{sec:main-infinite-tree} by transforming these distributional equations into a system of ordinary differential equations. This completes the proof of Theorem~\ref{thm:main}.


%% file: Sections/EasyHalf_CompArg.tex
As the first step toward the proof of the local weak convergence of $\widehat{\mu}_{n}$ to $\widehat{\mu}_{\infty}$, we show that for any sequence of $n$ that tends to infinity, there is a subsequence that converges weakly to some involution invariant random matching $\widehat{\mu}'_{\infty}$ on the planted PWIT. Saying this again in symbols, for any sequence $(n_k)$ there is an involution invariant $\widehat{\mu}'_{\infty}$ and a subsequence $(n_j) \subseteq (n_k)$ such that $\widehat{\mu}_{n_j} \xrightarrow{w} \widehat{\mu}'_{\infty}$. Our argument is somewhat simplified from~\cite[Section 5.8, pp. 53-54]{Aldous2004}.

By Theorem \ref{thm:lwc}, we already know $\mu_{n_k} \xrightarrow{w} \mu_{\infty}$: that is, the two models agree on their  local distributions of weighted neighborhoods. Since $\Gstar$ is a Polish space, the Prokhorov theorem~\cite[p. 309, Thm. 16.3]{Kallenberg2002} implies that the sequence $\mu_{n_k}$ is tight, i.e., for every $\epsilon > 0$ there is a compact set $\mathcal{K} \subset \Gstar$ such that $P([N_{{n_k},\circ}(1)] \in \mathcal{K}) > 1-\epsilon$ (recall that $\mu_{n_k}$ is the law of $[N_{{n_k},\circ}(1)]$). 

Define $\widehat{\mathcal{K}}$ by appending $\{0,1\}$ to the edges of each member of $\mathcal{K}$. Since $\mathcal{K}$ is compact, so is $\widehat{\mathcal{K}}$. Moreover, since $P([N_{{n_k},\circ}(1)] \in \mathcal{K}) > 1-\epsilon$, so is $P([\widehat{N}_{{n_k},\circ}(1)] \in \widehat{\mathcal{K}}) > 1-\epsilon$. Hence, the sequence $\widehat{\mu}_{n_k}$ is also tight, and by the Prokhorov theorem there is a further subsequence $n_j$ such that $\widehat{\mu}_{n_j}\xrightarrow{w}\widehat{\mu}'_{\infty}$ where $\widehat{\mu}'_{\infty}$ is some random matching on the planted PWIT. Since involution invariance passes through limit, $\widehat{\mu}'_{\infty}$ is involution invariant. Note that we cannot invoke the Portmanteau theorem to show that $\expect[\ell_{n_j}(1,M_{n_j,\min}(1))] \to \expect[\ell_{\infty}(\root,\M'_{\infty}(\root))]$, since the weight of the minimum matching is not a bounded continuous function with respect to the topology of (local) weak convergence.

By Skorokhod's theorem we can assume that $[\widehat{N}_{n_j,\circ}(1)] \lwc [\widehat{N}'_{\infty,\circ}]$ almost surely, where $[\widehat{N}'_{\infty,\circ}]\sim \widehat{\mu}'_{\infty}$, and $\widehat{N}'_{\infty} = (T_\infty,\ell_{\infty},\M'_\infty)$. By the definition of local convergence,
\[
\ell_{n_j}(1,M_{n_j,\min}(1)) \rightarrow \ell_{\infty}(\root,\M'_{\infty}(\root)) \text{ as } n\to\infty, \text{ almost surely}.
\]
Using Fatou's Lemma, we have
\[
\liminf_{n_j\to\infty}\expect[\ell_{n_j}(1,M_{n_j,\min}(1))] \geq \expect[\ell_{\infty}(\root,\M'_{\infty}(\root))].
\]
The lower bound~\eqref{eq:easyhalf} follows by assuming $n_k$ is a subsequence of $n$ that achieves $\liminf_{n \to \infty} \expect[\ell_{n}(1,\allowbreak\Mnopt(1))]$.


%% file: Sections/HardHalf_Intro.tex
In the ``easy half'' above we proved the inequality~\eqref{eq:easyhalf}, namely that the average weight of the minimum matching on $(K_{n,n},\ell_n)$ is bounded below by $\expect[\ell_{\infty}(\root,\Minfopt(\root))]$. Now, we are going to prove the inequality~\eqref{eq:hardhalf} in the opposite direction, and therefore that inequality can be replaced by equality within arbitrarily small $\epsilon$. The key idea is to construct a low weight matching on $(K_{n,n},\ell_n)$ using $(\ell_{\infty},\Minfopt)$. In particular, we want the average weight of this matching to be arbitrarily close to $\expect[\ell_{\infty}(\root,\Minfopt(\root))]$ for large enough $n$.

Recall that the weight of planted and un-planted edges in $(K_{n,n},\ell_n)$ are distributed as $\exp(\lambda)$ and $\exp(1/n)$ respectively, independent of everything else. Intuitively speaking, $K_{n,n}$ viewed from the planted directed edge $(1,1')$ corresponds to the doubly-rooted planted PWIT viewed from the root's planted edge $(\root,0)$, and $K_{n,n}$ viewed from an un-planted directed edge such as $(1,2)$ corresponds to the doubly-rooted planted PWIT viewed from an un-planted edge incident to the root, namely $(\root,i)$ where $i \in \nplus$ is arbitrary. 

Now, using this ``edge-centric'' viewpoint, and following the approach of Aldous in~\cite{Aldous1992,Aldous2001}, we will assign possibly fractional values to the edges of $(K_{n,n},\ell_n)$ such that the value assigned to the edge $e = \{i,j'\} \in E_n$ corresponds to the probability that $\Minfopt(e) = 1$, assuming its local neighborhood in $(K_{n,n},\ell_n)$ is a realization of the planted PWIT. 

Consider the $n \times n$ matrix $Q_n$ where $q_{i,j'}$ is the value assigned to the edge $(i,j'$). If these values corresponded exactly to probabilities, then $Q_n$ would be doubly-stochastic. That is, $Q_n$ would be a fractional matching, i.e., an element of the matching polytope, with weight 
\[
\ell_n(Q_n) = \sum_{i \in [n], j \in [n']} q_{i,j'} \ell_n(i,j') \, . 
\]
Since minimizing the weight is a linear programming problem, there is a vertex of this polytope, i.e., an honest matching  $M(\pi)_{i,j'} = 1$ if $j' = \pi(i)$ and $0$ otherwise for some permutation $\pi$, whose weight is less than or equal to $\ell_n(Q_n)$. Alternatively, we could use the Birkhoff theorem to write $Q$ as a convex combination of permutation matrices, 
\[
Q = \sum_{\pi} c_\pi M(\pi) \, .
\]
Then if we choose a random matching $\pi$ with probability $c_\pi$, the expected weight would be $\ell_n(Q_n)$. Finally, if $(1/n) \ell_n(Q_n)$ also converges to the expected weight of $\Minfopt$, we would be done.

All this is almost true. As we will see, we will define $Q_n$ by looking at a bi-infinite version of the planted PWIT, extending it in either direction from an edge~\cite[Section 5.2]{Aldous2001}. By looking at $(\ell_\infty, \Minfopt)$ on large neighborhoods of this edge, we will obtain probabilities that almost, but not quite, sum to $1$, since the true partner of a vertex in $\Minfopt$ might be outside this neighborhood. As a result, $Q_n$ is almost doubly-stochastic in a certain sense. Following~\cite{Aldous1992}, we then build an almost-perfect matching with weight close to that of $\Minfopt$, and then---by swapping a small fraction of edges---convert this into a perfect matching within increasing the weight very much. At this point in the proof, we can use lemmas in~\cite{Aldous1992} virtually unchanged.



\begin{remark}
\label{rem:simplif}
For sake of notational simplicity, for the remainder of the section, we will drop the subscripts $\infty$ and $\circ\circ$. Thus objects without the subscript $n$ live on the planted PWIT, while those with $n$ live on the finite model.
\end{remark}

%% file: Sections/HardHalf_BiInf.tex
\begin{figure}[!t]
	\begin{subfigure}[t!]{.45\linewidth}
		\begin{center}
			\includegraphics[width=\linewidth]{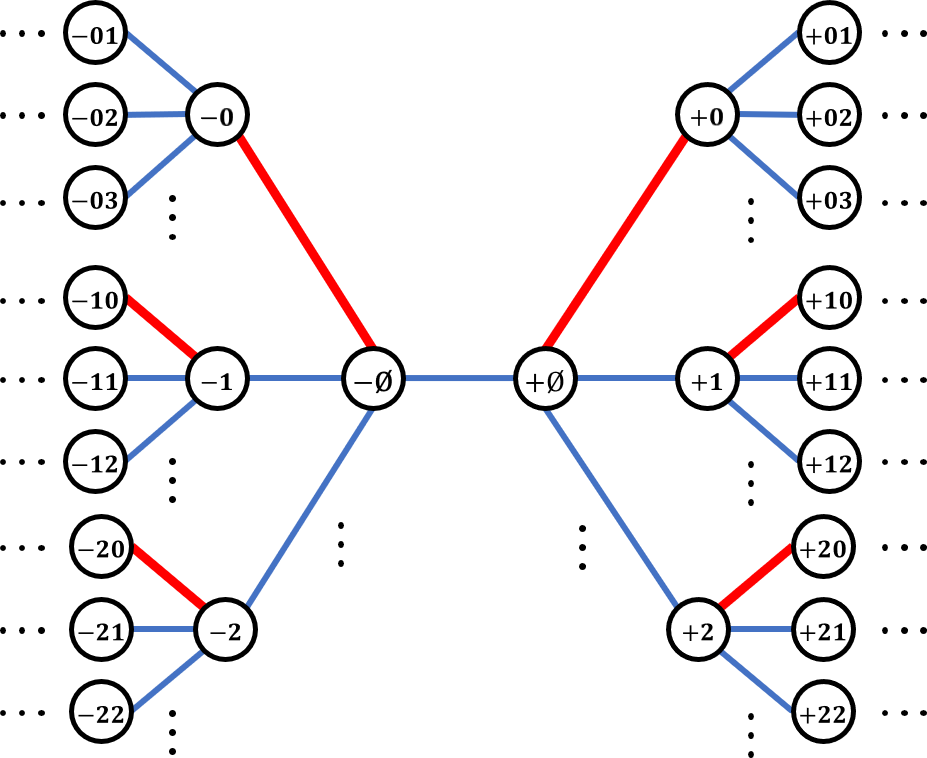}
			\caption{}\label{fig:Tu}
		\end{center}
	\end{subfigure}
	\hfill
	\begin{subfigure}[t!]{.45\linewidth}
		\begin{center}
			\includegraphics[width=\linewidth]{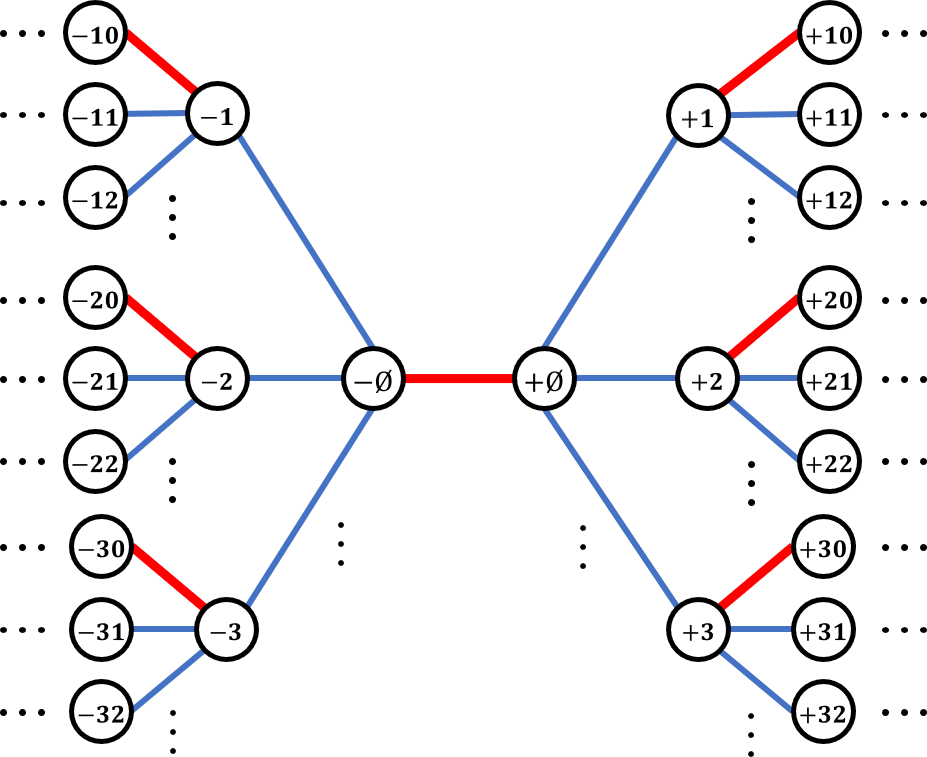}
			\caption{}\label{fig:Tp}
		\end{center}
	\end{subfigure}
	\caption{The structure of (a) $T^{\lra}_{u}$ and (b) $T^{\lra}_{p}$. The bold red edges are planted edges and the solid blue edges are un-planted.}
	\label{fig:biinf}
\end{figure}

As the first step toward the proof of the ``harder half'', we change the viewpoint from a vertex to an edge. There are two doubly-rooted infinite versions of the planted PWIT that we need to study. One is rooted at a planted edge $(\root,0)$, and the other is rooted at an un-planted edge $(\root,i)$ for some $i \in \nplus$. We illustrate these in Figure \ref{fig:biinf}. 

The measure on $i \in \nplus$ will be the uniform counting measure---that is, every position in the order of $\root$'s un-planted edges has the same measure. Although this measure is not normalizable, we will speak informally of $i$ as ``uniformly chosen.''  One intuition for this uniformity comes from the finite model $K_{n,n}$. If we choose uniformly from the $n(n-1)$ un-planted edges, or (by symmetry) from the $n-1$ un-planted edges of the root vertex $1$, then its order in the sorted list of edge weights is uniform on the set $\{1,\ldots,n-1\}$. 

The other intuition is as follows. The edge weight of an un-planted edge is distributed as $\exp(1/n)$, which for weights of constant size is asymptotically $1/n$ times the Lebesgue measure on $\rplus$. If the weight is $x$, the probability that the Poisson process of weight $1$ generates $i-1$ arrivals in the interval $[0,x]$ is $\e^{-x} x^{i-1} / (i-1)!$. The total Lebesgue measure of this event is then
\begin{equation}
\label{eq:uniform-i}
\int_0^\infty \frac{\e^{-x} x^{i-1}}{(i-1)!} \dx = \frac{\Gamma(i)}{(i-1)!} = 1 
\quad \text{for all $i$} \, ,
\end{equation}
so every $i$ has equal measure.

We make this intuition rigorous below, showing how the appropriate measure on bi-infinite PWITs around both types of edge is related to the planted PWIT by extending the strategy of~\cite[Section 5.2]{Aldous2001} to the planted case.


In the case of the planted edge $\{\root,0\}$, the corresponding bi-infinite tree is just a relabeling of the vertices: relabel $\root$ as $-\root$ and relabel $0$ as $+\root$, and then relabel all the other vertices ``below'' these two roots as we did before (see Figure \ref{fig:biinf}). However, for a planted edge $\{\root,i\}$ for a ``uniformly'' selected $i \in \nplus$, things are a bit more complicated. Since $i \in \nplus$ is ``uniformly'' selected, the cost of the edge $\{\root,i\}$ is ``uniformly distributed'' over $\mathbb{R}_+$, i.e., Lebesgue measure on $\mathbb{R}_+$. On the other hand, as we pointed out in the proof of Corollary \ref{cor:iidX}, the other un-planted children of $\root$ are still arrivals of a Poisson process. Specifically, if we remove $\{\root,i\}$, the remaining connected component of $\root$ is still a planted PWIT. Hence, the corresponding bi-infinite tree is obtained by gluing two independent copies of the planted PWIT, using an edge with edge weight ``distributed'' as Lebesgue measure, and then relabeling the vertices: $\root$ as $-\root$, $i$ as $+\root$, and the others accordingly (See Figure \ref{fig:biinf}). In the rest of the section, we are going to give detailed construction of the bi-infinite trees, and then we show that they are equivalent to the corresponding doubly-rooted planted PWIT.

\begin{remark}
	To distinguish the edge-centric viewpoint from the singly-rooted vertex viewpoint, we use the superscript $\lra$.
\end{remark}

Let $T_{u}$ denote the planted PWIT rooted at $\root$, and let $T_{p}$ denote the subtree, rooted at $0$, obtained by removing the edge $\{\root,0\}$. Relabel the root of $T_{p}$ to be $\root$, and then relabel all the vertices of $T_p$ using the same rule as in the planted PWIT.
Let $T_{u}^-$ and $T_{u}^+$ be two independent copies of $T_{u}$. Similarly, let $T_{p}^-$ and $T_{p}^+$ be two independent copies of $T_{p}$. Relabel the vertices of $T_{u}^+$ and $T_{p}^+$ by adding ``$\,+\,$'' sign to the original labels, and relabel the vertices of $T_{u}^-$ and $T_{p}^-$ by adding ``$\,-\,$'' sign to the original labels. 

Now, let $T^{\lra}_{u}$ denote a bi-infinite tree, rooted at $(-\root,+\root)$, obtained by joining the roots of $T_{u}^-$ and $T_{u}^+$. Let $V^{\lra}_{u}$ and $E^{\lra}_{u}$ denote the vertices and the edges of $T^{\lra}_{u}$ respectively. Let $M^{\lra}_{u}$ denote the set of planted edges of $T^{\lra}_{u}$, which is the union of the planted edges in $T_{u}^-$ and $T_{u}^+$. Let $\ell^{\lra}_u$ denote the function that assigns weight to the edges of $T^{\lra}_{u}$ by using the weight of the edges in $T_{u}^-$ and $T_{u}^+$, and specifying the weight of $\{-\root,+\root\}$ to be uniformly ``distributed'' on $[0,\infty)$, i.e., Lebesgue measure on $\mathbb{R}_+$, independent of everything else. Write $T^{\lra}_{u} = (G^{\lra}_{u},\ell^{\lra}_u)$, where $G^{\lra}_{u} = (V^{\lra}_{u},E^{\lra}_{u},M^{\lra}_{u},(-\root,+\root))$, and let $\mu^{\lra}_{u}$ denote the $\sigma$-finite measure associated with $T^{\lra}_{u}$.

Similarly, define $T^{\lra}_{p}$ by joining the roots of $T_{p}^-$ and $T_{p}^+$. However, this time include $\{-\root,+\root\}$ as a planted edge in $M^{\lra}_{p}$, and specify the weight of $\{-\root,+\root\}$ to be an exponentially distributed random variable with parameter $\lambda$, independent of everything else. Write $T^{\lra}_{p} = (G^{\lra}_{p},\ell^{\lra}_p)$, where $G^{\lra}_{p} = (V^{\lra}_{p},E^{\lra}_{p},M^{\lra}_{p},(-\root,+\root))$, and let $\mu^{\lra}_{p}$ denote the probability distribution of $T^{\lra}_{p}$. Figure \ref{fig:biinf} illustrates a realization of $T^{\lra}_{u}$ and $T^{\lra}_{p}$.

Recall that the doubly-rooted planted PWIT is the product measure $\mu\times \texttt{count} \text{ on }\{0,1,2,\allowbreak3,\cdots\}$, where $\texttt{count}$ is the counting measure. We can think of it as a product measure on $[0,\infty)^E \times \{0,1,2,\cdots\}$. Similarly, we can think of $\mu^{\lra}_{u}$ and $\mu^{\lra}_{p}$ as a $\sigma$-finite measure on $[0,\infty)^{E^{\lra}_u}$ and a probability measure on $[0,\infty)^{E^{\lra}_p}$, respectively.
Now, depending on whether the second root is $0$ or not, there is a natural map from the doubly-rooted planted PWIT to $T^{\lra}_{p}$ or $T^{\lra}_{u}$: (1) If the second root is $0$, relabel the vertices as 
\begin{enumerate}[label=\roman*.]
	\item relabel $0$ as $+\root$ and $\root$ as $-\root$,
	\item relabel any sequence $0i_1i_2\cdots i_l$ for $l\geq 0$ as $+i_1i_2\cdots i_l$,
	\item relabel any sequence $i_1i_2\cdots i_l$ for $l\geq 1$ such that $i_1 \neq 0$ as $-i_1i_2\cdots i_l$.
\end{enumerate}
This relabeling induces a bijection $\psi\big|_p : [0,\infty)^E \times \{0\} \to [0,\infty)^{E^{\lra}_p}$; (2) If the second root is $k\in \{1,2,3,\cdots\}$, relabel the vertices as
\begin{enumerate}[label=\roman*.]
	\item relabel $k$ as $+\root$ and $\root$ as $-\root$,
	\item relabel any sequence $ki_1i_2\cdots i_l$ for $l\geq 0$ as $+i_1i_2\cdots i_l$,
	\item relabel any sequence $i_1i_2\cdots i_l$ for $l\geq 1$ such that $i_1 > k$ as $-(i_1-1)i_2\cdots i_l$,
	\item relabel any sequence $i_1i_2\cdots i_l$ for $l\geq 1$ such that $i_1 < k$ as $-i_1i_2\cdots i_l$.
\end{enumerate}
This relabeling induces a bijection $\psi\big|_u : [0,\infty)^E \times \{1,2,3,\cdots \} \to [0,\infty)^{E^{\lra}_u}$. Clearly, $\psi\big|_p$ maps $\mu\times \delta_{0}$ to $\mu^{\lra}_{p}$, where $\delta_{0}$ is the delta measure on the second root. It is also easy to see that $\psi\big|_u$ maps $\mu\times \texttt{count} \text{ on } \{1,2,3,\cdots \}$ to $\mu^{\lra}_{u}$ using the following lemma. Finally, note that
\begin{align*}
\mu\times \texttt{count} \text{ on } \{0,1,2,3,\cdots \} = \mu\times \delta_{0} + \mu\times \texttt{count} \text{ on } \{1,2,3,\cdots \}.
\end{align*}

\begin{lemma}\cite[Lemma 25]{Aldous2001}
	Write $\Delta \coloneqq \{(x_i): 0< x_1 < x_2 < \cdots, x_i\to\infty \}$. Write \texttt{Pois} for the probability measure on $\Delta$ which is the distribution of the Poisson process of rate $1$. Consider the map $\mathcal{X}:\Delta \times \{1,2,3,\cdots\} \to \Delta \times [0,\infty)$ which takes $((x_i),k)$ to $((x_i,i\neq k),x_k)$. Then $\mathcal{X}$ maps $\texttt{Pois} \times \texttt{count}$ to $\texttt{Pois} \times \texttt{Leb}$, where \texttt{Leb} is the Lebesgue measure on $[0,\infty)$.
\end{lemma}
\begin{remark}
	For sake of brevity, we use $E^\lra_{\cdot}$ instead of $E^\lra_u$ or $E^\lra_p$ whenever the subscript does not affect the discussion. Similar changes applies to other symbols. More specifically, all the following arguments are true if we replace all the ``$\,_{\cdot}\,$'' with``$\,_{u}\,$'' or ``$\,_{p}\,$''.
\end{remark}
We introduce two additional pieces of notation that we will use in the following subsections. Recall that $G^{\lra}_{\cdot} = (V^{\lra}_{\cdot},E^{\lra}_{\cdot},M^{\lra}_{\cdot},(-\root,+\root))$ is a bi-infinite tree (without the edge weights). We define $G^{\lra}_{\cdot,B}$ to be a subtree of $G^{\lra}_{\cdot}$ induced by $V^{\lra}_{\cdot,B}\coloneqq \{\pm\root \} \cup \{\pm i_1i_2\cdots i_l \in V^{\lra}_{\cdot} : i_s \in \{0,1,\cdots, B\} \}$, i.e., the subtree obtained by restricting the number of un-planted children of each vertex to $B$. We also define $G^{\lra}_{\cdot,B,H}$ to be a subtree of $G^{\lra}_{\cdot,B}$ induced by $V^{\lra}_{\cdot,B,H}\coloneqq \{\pm\root \} \cup \{\pm i_1i_2\cdots i_l \in V^{\lra}_{\cdot,B} : l \leq H+1 \text{ and } l = H+1 \text{ iff }i_l = 0\}$, i.e., the subtree obtained by restricting the depth of vertices to $H$ or $H+1$, depending on whether the vertex is a planted pair of its parent or not. Define $\ell^{\lra}_{\cdot,B}$ to be the restriction of $\ell^{\lra}_{\cdot}$ to $\ell^{\lra}_{\cdot,B}$, and $\ell^{\lra}_{\cdot,B,H}$ to be the restriction of $\ell^{\lra}_{\cdot,B}$ to $\ell^{\lra}_{\cdot,B,H}$. Now, define $T^{\lra}_{\cdot,B} \coloneqq (G^{\lra}_{\cdot,B},\ell^{\lra}_{\cdot,B})$ and $T^{\lra}_{\cdot,B,H} \coloneqq (G^{\lra}_{\cdot,B,H},\ell^{\lra}_{\cdot,B,H})$, and let $\mu^{\lra}_{\cdot,B}$ and $\mu^{\lra}_{\cdot,B,H}$ to be the associated measures. Note that there is a natural restriction ${\rho}_{\cdot,H}$ that maps $T^{\lra}_{\cdot,B}$ to $T^{\lra}_{\cdot,B,H}$.

Thus far, we show that the doubly-rooted planted PWITs and bi-infinite trees are equivalent. Next, we are going establish the connection between the planted model $(K_{n,n},\ell_n)$ and the bi-infinite trees. 

%% file: Sections/HardHalf_Unfold.tex
Now that we have discussed how to view the planted PWIT from an edge (a planted edge or a ``uniformly'' selected un-planted edge), we are going to discuss the similar viewpoint in $K_{n,n} = (V_n,E_n,\Mplanted_n)$. This is done via an unfolding map that unfolds $(K_{n,n},\ell_n)$ viewed from a planted or un-planted edge. This unfolding map is similar to the one discussed in \cite[Section 3.2]{Aldous1992} with two additional properties: the number of un-planted children of every vertex is the same (for a fixed $B$ when $n$ is large enough),
and the set of un-planted edges $\Mplanted_n$ is a matching on the unfolded graph. Eventually, we want to show that the local neighborhood of an edge in $K_{n,n}$ is the same as the local neighborhood of the edge in the bi-infinite trees. This should not come as a surprise, given Theorem \ref{thm:lwc}.

Next, we are going to give detailed construction of the unfolding map. Fix $B \in \mathbb{N}_+$, and fix some edge $\{i,j'\}$ from $K_{n,n}$. We are going to unfold $(K_{n,n},\ell_n)$ from the viewpoint of the edge $\{i,j'\}$ and construct the doubly-rooted tree $U^{\lra}_{B,n}(i,j')$, rooted at $(i,j')$, with maximum $B+1$ arity. This unfolding map resembles the exploration process discussed in the proof of Theorem \ref{thm:lwc}, with one key difference: this map unfolds $(K_{n,n},\ell_n)$ from an edge viewpoint.

Include the edge $\{i,j'\}$ in $U^{\lra}_{B,n}(i,j')$. 
The unfolding map proceeds as follows.  In this process, all vertices will be ``live'', ``dead'', or ``neutral''. The live
vertices will be contained in a queue. Initially, $i$ and $j'$ are live and the queue consists of only
$i$, $j'$ in order and all the other vertices are neutral.  
At each time step, a live vertex $v$ is popped from the head of the queue. 
If $v$'s planted neighbor $M_n^*(v)$ is neutral, then include the edge $\{v, M_n^*(v)\}$ in 
the doubly-rooted tree 
$U^{\lra}_{B,n}(i,j')$, add $M_n^*(v)$ to the end of the queue and let $M_n^*(v)$ live. 
Let $v_1, v_2, \ldots, v_B$ denote the 
$B$ closest un-planted, neutral neighbors of vertex $v$ in a non-decreasing order of distance to $v$. 
Include all the edges $\{v,v_k\}$ in the doubly-rooted tree 
$U^{\lra}_{B,n}(i,j')$ for $1 \le k \le B$. Also, include all the edges $\{v_k, M_n^*(v_k)\}$ 
in $U^{\lra}_{B,n}(i,j')$ for $ 1 \le k \le B$. The popped vertex $v$ is dead. Add those neutral vertices $v_1, v_2, \ldots, v_B$ to the end of the queue in order and they are live. Also, add  those neutral vertices 
$M_n^*(v_1), M_n^*(v_2), \ldots, M_n^*(v_B)$ to the end of the queue in order and they are live.
The process ends when the queue is empty.  According to this rule, 
the order of the vertices that are selected after $j'$ is
 the planted neighbor $i_0$ of $i$ (if $\{i,j'\}$ is not planted), $i_1,\, i_2,\,\cdots,\,i_B$, 
$\Mplanted_n(i_1),\,\Mplanted_n(i_2),\,\cdots,\,\Mplanted_n(i_B)$,
the planted neighbor $j'_0$ of $j'$ (if $\{i,j'\}$ is not planted), $j'_1,\,j'_2,\,\cdots,\,j'_B$,  $\Mplanted_n(j'_1),\,\Mplanted_n(j'_2),\,\cdots,\,\Mplanted_n(j'_B)$, etc.
This unfolding process stops when all vertices of $K_{n,n}$ are included in $U^{\lra}_{B,n}(i,j')$. Figure \ref{fig:unfold} illustrates the $U^{\lra}_{1,4}(1,2')$ 
and $U^{\lra}_{1,4}(1,1')$ for the planted network given in Figure \ref{fig:pmod}. Note that in both cases, $\Mplanted_n$ is a matching.

\begin{figure}[!t]
	\centering
	\begin{subfigure}[t!]{.65\linewidth}
		\begin{center}
			\includegraphics[width=\linewidth]{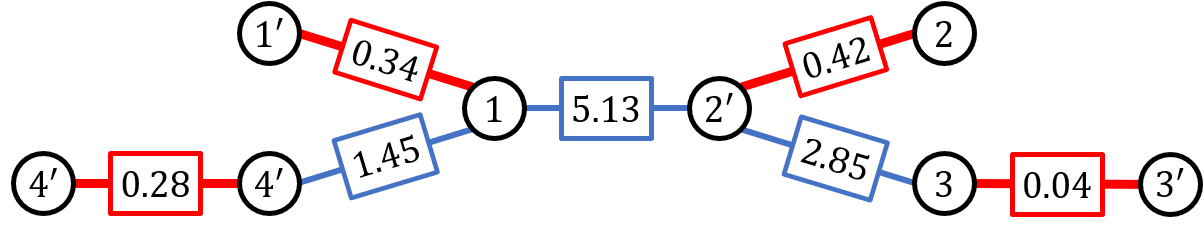}
			\caption{}
		\end{center}
	\end{subfigure}
	\hfill
	\begin{subfigure}[t!]{.65\linewidth}
		\begin{center}
			\includegraphics[width=\linewidth]{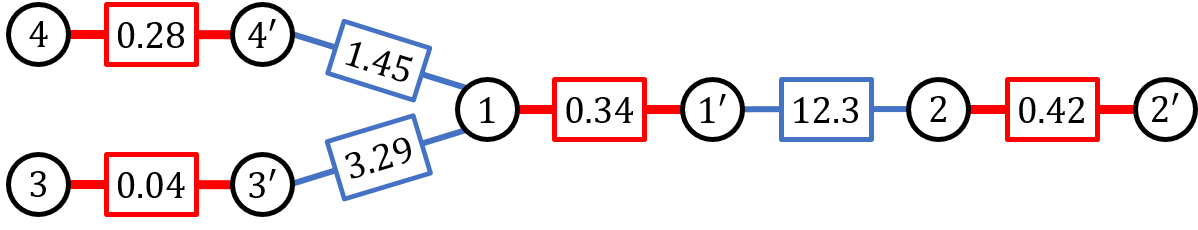}
			\caption{}
		\end{center}
	\end{subfigure}
	\caption{The doubly-rooted trees (a) $U^{\lra}_{1,4}(1,2')$ and (b) $U^{\lra}_{1,4}(1,1')$ obtained by unfolding the planted network in Figure \ref{fig:pmod}.}\label{fig:unfold}
\end{figure}

Let us define a relabeling bijection $\phi^{i,j'}_{B,n}$ (not to be confused with $\phi\big|_p$ and $\phi\big|_u$) from $V_n$ to a subset of $+\Sigma\, \cup -\Sigma$, 
where $+\Sigma:= \{+\bs{i}:\bs{i}\in\Sigma\}$ and  $-\Sigma:= \{-\bs{i}:\bs{i}\in\Sigma\}$ denote the set of vertex labelings of the bi-infinite planted PWIT. 
 Define $\phi^{i,j'}_{B,n}(i) = -\root$ and $\phi^{i,j'}_{B,n}(j') = +\root$. At each step of the unfolding process, when we pick vertex $v$ and add $\{v,v_k\}$ and $\{v_k,\Mplanted_n(v_k)\}$ to $U^{\lra}_{B,n}(i,j')$, set $\phi^{i,j'}_{B,n}(v_k) = \concat{\phi^{i,j'}_{B,n}(v)k}$ and $\phi^{i,j'}_{B,n}(\Mplanted_n(v_k)) = \concat{\phi^{i,j'}_{B,n}(v_k)0}$. 

Here, using the bijection $\phi^{i,j'}_{B,n}$, we first map $U^{\lra}_{B,n}(i,j')$ to a subtree of $T^{\lra}_{\cdot,B}$ and then apply  the restriction map $\rho_{\cdot,H}$ on this subtree. For ease of notation, we abbreviate $\rho_{\cdot,H} \left( \phi^{i,j'}_{B,n} \left( U^{\lra}_{B,n}(i,j') \right) \right) $ as $\rho_{\cdot,H} \left( U^{\lra}_{B,n}(i,j') \right)$ whenever the context is clear.
It is easy to see that the graph structure of $\rho_{u,H}\left(U^{\lra}_{B,n}(i,j')\right)$ for $j'\neq i'$ and all sufficiently large $n$ is isomorphic to $G^{\lra}_{u,B,H}$, and the graph structure of $\rho_{p,H}\left(U^{\lra}_{B,n}(i,i')\right)$ for all sufficiently large $n$ is isomorphic to $G^{\lra}_{p,B,H}$. Let $\mu^{\lra}_{n,u,B,H}$ denote the probability measure associated with $\rho_{u,H}\left(U^{\lra}_{B,n}(i,j')\right)$ for $j'\neq i'$, and let $\mu^{\lra}_{n,p,B,H}$ denote the probability measure associated with $\rho_{u,H}\left(U^{\lra}_{B,n}(i,i')\right)$. Note that if $D_{u,x} = \{T^{\lra}_{u,B,H}: \ell^{\lra}_{u,B,H}( \{-\root,\root\} )< x \}$, then $\mu^{\lra}_{n,u,B,H}(D_{u,x}) = 1 - \exp(-x/n) \approx x/n$ as $n\to\infty$. Also, note that $\mu^{\lra}_{u,B,H}(D_{u,x}) = x$. We now generalize Lemma 10 in \cite{Aldous1992} to the planted case.
\begin{lemma}\label{lem:unfold}
	For any fixed $B, H\in \mathbb{N}_+$, and $x>0$ we have
	\begin{align*}
	n \mu^{\lra}_{n,u,B,H}(D_{u,x} \cap  \cdot \,) &\xrightarrow{TV} \mu^{\lra}_{u,B,H}(D_{u,x} \cap  \cdot \,)\\
	\mu^{\lra}_{n,p,B,H}&\xrightarrow{TV} \mu^{\lra}_{p,B,H}, 
	\end{align*}
	where the total variation convergence of positive measures is defined as:
	\begin{align*}
	\mu_n\xrightarrow{TV} \mu \text{ iff } \sup_A|\mu_n(A) - \mu_(A)| \to 0 .
	\end{align*}
\end{lemma}
\begin{proof}
	The proof is almost identical to the proof of Lemma \ref{lem:tvconv}, and has been omitted. The only subtle difference is the restriction to $D_{u,x}$, which causes no problem since $n\times \frac{1}{n}\e^{-x/n} \to 1$ as $n\to\infty$.
\end{proof}
Next, we are going to use the involution invariant random matching $(\ell,\M_{\opt})$ 
on the planted PWIT to assign values to the edges of $(K_{n,n},\ell_n)$. Ideally, the value assigned to an edge $e = \{i,j'\}$ corresponds to the probability of $e$ being in the matching that we want to construct.

%% file: Sections/HardHalf_AlMat.tex
Now it's time to assign fractional values to the edges of $(K_{n,n},\ell_n)$. This is done by pretending that the local neighborhood of an edge $e$ in $(K_{n,n},\ell_n)$, is a realization of the local neighborhood of the corresponding bi-infinite tree. So, as the first step toward assigning values to the edges of $(K_{n,n},\ell_n)$, we need to know how to assign value to $\{-\root,+\root\}$ in the bi-infinite tree using the minimum matching on the planted PWIT. The idea is to use the inverse image of $\psi\big|_\cdot$ and map the edge $\{-\root,+\root\}$ to the corresponding edge on the planted PWIT. This gives us a function $g_{\cdot}:[0,\infty)^{E^{\lra}_{\cdot}}\to[0,1]$, which we then use to assign fractional values to the edges of $(K_{n,n},\ell_n)$, by conditioning on its neighborhood. We are going to discuss these steps in detail. We follow the discussion in \cite[Section 5.5]{Aldous2001} and then \cite[Introduction]{Aldous1992}.

Recall that $\widetilde{\mu} = \mu \times \texttt{count} \text{ on }\{0,1,2,3,\cdots\}$ is the measure associated with the doubly-rooted planted PWIT. Also, recall that since $\mu$ is involution invariant, the product measure $\widetilde{\mu}$ is invariant under the involution map $\iota$ that swaps the roots. Following the discussion of \cite[Section 5.5]{Aldous2001}, we define the function $\gamma_{\opt}:[0,\infty)^E \times \{0,1,2,\cdots\}\to[0,1]$ by
\begin{align*}
\gamma_{\opt}(\bs{w},i) = \prob[\M_{\opt}(\root) = i | \ell(e) = w(e)\, \forall e \in E ]
\end{align*}
As Aldous points out, the function $\gamma_{\opt}$ satisfies certain consistency properties:
\begin{enumerate}[label=(\roman*)]
	\item $\sum_{i=0}^\infty \gamma_{\opt}(\bs{w},i) = 1$ since $\M_{\opt}$ is a matching.
	\item $\gamma_{\opt}(\iota(\bs{w},i)) = \gamma_{\opt}(\bs{w},i)$ since the random matching $(\ell,\M_\opt)$ is involution invariance, where $\iota(\bs{w},i)$ swaps $i$ and $\root$ given $\{\ell(e) = w(e)\, \forall e\in E\}$.
\end{enumerate}
Also, we have (iii) $\expect[\ell(\M_{\opt}(\root),\root)] = \int_{\bs{w}} \sum_{i=0}^\infty \gamma_{\opt}(\bs{w},i)w(\root,i)\, \mu(d\bs{w})$. 
Now, define $g_{\cdot}:[0,\infty)^{E^{\lra}_{\cdot} }\to[0,1]$ as follows:
\begin{align*}
g_.(\bs{w}^{\lra}_.) = \gamma_{\opt}(\psi\big|_\cdot^{-1}(\bs{w}^{\lra}_.)),
\end{align*}
where $\psi\big|_\cdot$ is the bijection map that we defined earlier in Section \ref{app:BiinfTree}, and $\psi\big|_\cdot^{-1}$ is its inverse.
We can think of the function $g_\cdot$ as the probability that the edge $\{ -\root,+\root \}$ is in the matching. The function $g_\cdot$ satisfies similar consistency properties corresponding to (i) and (ii):
\begin{enumerate}[label=(\roman*$^\prime$)]
	\item The function $g_\cdot$ assigns honest probabilities to the neighbors of $-\root$ as well as $+\root$. To be more specific,
	let $\iota_{p,i}^+,\iota_{p,i}^-:[0,\infty)^{E^{\lra}_{p}}\to[0,\infty)^{E^{\lra}_{u}}$ denote the root swapping maps, that change the root from $(-\root,+\root)$ to $(-\root,-i)$, and from $(-\root,+\root)$ to $(+i,+\root)$ respectively (and then relabel all the vertices). Similarly, define $\iota_{u,i}^+$ and $\iota_{u,i}^-$. We have:
	\begin{enumerate}[label=(\alph*)]
		\item In $T^{\lra}_p$, the values assigned to the neighbors of $-\root$ as well as $+\root$ sums to one:
			\begin{align*}
			&g_p(\bs{w}^{\lra}_p) + \sum_{i=1}^\infty g_u(\iota_{p,i}^{\pm}(\bs{w}^{\lra}_p))= 1.
			\end{align*}
		\item	In $T^{\lra}_u$, the values assigned to the neighbors of $-\root$ as well as $+\root$ sums to one:
			\begin{align*}
			&g_u(\bs{w}^{\lra}_u) + g_p(\iota_{u,0}^\pm(\bs{w}^{\lra}_u)) + \sum_{i=1}^\infty g_u(\iota_{u,i}^\pm(\bs{w}^{\lra}_u))= 1.
			\end{align*}
		\end{enumerate}
	\item Let $\iota^\lra_\cdot$ denote the root swamping map, that swaps the root $(-\root,+\root)$ to $(+\root,-\root)$ (and then relabels all the vertices). We have $g_.(\bs{w}^{\lra}_.) = g_.(\iota^\lra_\cdot(\bs{w}^{\lra}_.))$.
\end{enumerate}
Also note that the function $g_\cdot$ is measurable with respect to the product $\sigma$-algebra on $[0,\infty)^{E^\lra_\cdot}$. Hence, (iii) becomes (iii$^\prime$)
\begin{align*}
\expect[\ell(\M_{\opt}(\root),\root)] = &  \int_{\bs{w}^{\lra}_p} \bs{w}^{\lra}_p(-\root,+\root) \,g_p(\bs{w}^{\lra}_p) \, \mu^{\lra}_p(d\bs{w}^{\lra}_p) \\ & + \int_{\bs{w}^{\lra}_u} \bs{w}^{\lra}_u(-\root,+\root)\, g_u(\bs{w}^{\lra}_u) \, \mu^{\lra}_u(d\bs{w}^{\lra}_u).
\end{align*}

Now, we use the function $g_\cdot$ to assign fractional values to the edges of the planted model. We store the values assigned to the edges of $(K_{n,n},\ell_n)$ in an $n\times n$ matrix $Q_n=[q_{i,j'}]_{i,j'}$, where $q_{i,j'}$ is the value assigned to the edge $\{i,j'\}$.
Let $g_{\cdot,B,H}$ to be the conditional expectation of $g_\cdot$ given $\sigma(\ell_{\cdot}^{\lra}(e),e\in E^{\lra}_{\cdot,B,H})$, with respect to the measure $\mu^{\lra}_{\cdot}$. Also, let $g_{u,x,B,H} = g_{u,B,H}\bs{1}_{D_{u,x}}$. Now, $q_{i,j'}$ is defined as follows:
\begin{align}
q_{i,j'} \coloneqq 
\begin{cases*}
g_{u,x,B,H}\left(\rho_{u,H}\left(U^{\lra}_{B,n}(i,j')\right)\right)&if $j'\neq i'$\\
g_{p,B,H}\left(\rho_{p,H}\left(U^{\lra}_{B,n}(i,i')\right)\right)&if $j'= i'$
\end{cases*}
\label{eq:construction_Qn}
\end{align}
As we mentioned before, if $Q_n$ were doubly-stochastic, we could have used it to construct a matching on $K_{n,n}$. However, it is not hard to see that $\prob[Q_n \text{ is doubly-stochastic}] = 0$. Nevertheless, we expect $Q_n$ to become almost doubly-stochastic, as a larger neighborhood is revealed. 

Let us define a discrimination factor, as Aldous does in~\cite[Introduction]{Aldous1992}:
\begin{align*}
\X(Q_n) = \frac{1}{n}\sum_{i=1}^{n}\left|1 - \sum_{j'=1'}^{n'} q_{i,j'}\right| + \frac{1}{n}\sum_{j'=1'}^{n'}\left|1 - \sum_{i=1}^n q_{i,j'}\right|
\end{align*}
Note that if $Q_n$ is doubly-stochastic, then $\X(Q_n) =0$. Naturally, we should expect that $\expect[\X(Q_n)] \approx 0$, for large values $x$, $B$, $H$, and $n$. We should also expect the average expected cost of $Q_n$ to be close to the expected cost of $\M_{\opt}$, i.e.,
\begin{align*}
\frac{1}{n}\expect[\sum_{i,j'} q_{i,j'} \ell_n(i,j')] \approx \expect[\ell(\M_{\opt}(\root),\root)],
\end{align*}
for large enough values of $x$, $B$, $H$, and $n$. Using the same set of inequalities as in \cite[Section 3.4]{Aldous1992}, it follows that the both intuitions are correct.
\begin{lemma}\label{lem:2limits}
	\begin{enumerate}[label=(\roman*)]
		\item For any $\epsilon > 0$, there is an $x_0$, $B_0$, $H_0$ and $n_0$ such that for all $x>x_0$, $B>B_0$, $H>H_0$, and $n>n_0$ we have $\expect[\X(Q_n)] < \epsilon$. In the other words,
		\begin{align*}
				\lim_{x\to\infty}\limsup_{B\to\infty}\limsup_{H\to\infty}\limsup_{n\to\infty}\expect[\X(Q_n)]=0.
		\end{align*}
		\item For any $\delta > 0$, there is an $x_0$ and $n_0$ such that for all $x>x_0$ and $n>n_0$, the average expected cost of $Q_n$ is in the $\delta$ neighborhood of the cost of $\M_{\opt}$ on the planted PWIT, for all $B$ and $H$. In the other words,
		\begin{align*}
		\lim_{x\to\infty}\limsup_{n\to\infty}\frac{1}{n}\expect[\sum_{i,j'} q_{i,j'} \ell_n(i,j')]=\expect[\ell(\M_{\opt}(\root),\root)] \qquad \text{for all $B$ and $H$}.
		\end{align*}
	\end{enumerate}
\end{lemma}	
\begin{proof}

The proof is almost identical to the proof presented in \cite[Section 3.4]{Aldous1992}. The proof of part (ii) is a direct consequence of Lemma \ref{lem:unfold} and linearity of expectation. The proof of part (i) needs more work, but the inequalities are the same as the one presented in \cite[Section 3.4]{Aldous1992}. The key factor is the consistency properties of functions $g_u$ and $g_p$, Lemma \ref{lem:unfold}, and the fact that there is no short cycle containing the root consists entirely of low-weight edges (as mentioned in the proof of Theorem~\ref{thm:lwc}).
\end{proof}

Now that we know $Q_n$ eventually becomes a doubly-stochastic matrix with weight close to that of $\M_{\opt}$, we will construct a perfect matching on $(K_{n,n},\ell_n)$ with (within $\epsilon$) the same weight. First, by invoking \cite[Proposition 7]{Aldous1992}, we construct a partial matching with the cost close to the cost of $Q_n$. Next, using~\cite[Proposition 9]{Aldous1992}, we construct a perfect matching by swapping operation while keeping the cost almost the same. The changes required to extend this analysis to the planted case are very minor, but we present the strategy in the next section for completeness.

%% file: Sections/MatchConst.tex
Finally, we are going to give the precise construction of the low cost matching on $(K_{n,n},\ell_n)$ and prove~\eqref{eq:hardhalf}. By Lemma \ref{lem:2limits}, we know that $Q_n$ becomes arbitrary close to a doubly-stochastic matrix, i.e., $\X(Q_n)$ becomes arbitrary close to $0$ with high probability. Now, using $Q_n$ for sufficiently large $n$, we construct a low cost partial matching which matches most of the vertices and leaves a small fraction of vertices isolated. This is done by invoking Proposition 7 in~\cite{Aldous1992}. Then we use Proposition 9 in~\cite{Aldous1992} and swap some of the edges to obtain a perfect matching on $(K_{n,n},\ell_n)$ while keeping the cost almost the same. On the rare occasions  that we fail to either construct the partial matching or swap the edges, we use $\Mplanted_n$. Since these events are rare, the use of planted matching won't affect the total cost. We follow the discussion in \cite[Section 2]{Aldous1992}.

Let us begin with describing how to construct a partial matching using an almost doubly-stochastic matrix. We say $Q_n$ is an almost doubly-stochastic matrix if its discrimination factor is close enough to $0$, or more precisely, if $\X(Q_n) < 1/200$. We say $\nu_n:E_n\to\{0,1,\root\}$ is a $(1-\theta)$ partial matching if vertices in $U(\nu_n) = \{i \in [n] : \nu_n(i, j') = 1 \text{ for some } j' \in [n'] \}$ are matched to different vertices in $[n']$ and $|U(\nu_n)| \geq (1-\theta)n$.

The first step is to convert $Q_n$ to a doubly-stochastic matrix. Define an $n\times n$ matrix $A_n = [a_{i,j'}]$ as follows:
\begin{align*}
a_{i,j'} \coloneqq \frac{q_{i,j'}}{\max(1, q_{i,:})\max(1, q_{:,j'})},
\end{align*}
where $q_{i,:} \coloneqq \sum_{j'\in[n']} q_{i,j'} \text{ and }q_{:,j'} \coloneqq \sum_{i\in[n]} q_{i,j'}$. Similarly, define $a_{i,:}$ and $a_{:,j'}$ for all $i\in[n]$ and $j'\in[n']$. Note that $a_{i,:} \leq 1$ and $a_{:,j'}\leq 1$. Define an $n\times n$ matrix $B_n  = [b_{i,j'}]$ as follows:
\begin{align*}
b_{i,j'} = \frac{(1 - a_{i,:})(1 - a_{:,j'})}{n - \sum_{i\in[n]}\sum_{j'\in[n']}{a_{i,j'}}}.
\end{align*}
It is easy to check that $A_n  + B_n$ is doubly-stochastic. Hence, by Birkhoff--von Neumann theorem $A_n+B_n$ can be written as a convex combination of permutations. Hence, there exists a random matching $\M_n$ on $(K_{n},\ell_n)$ such that $\forall i\in[n],\,j'\in[n']:\, P(\M_n(i) = j') = a_{i,j'}  + b_{i,j'}$.
Note that for all $i\in [n]$ and $j'\in[n']$ we have $a_{i,j'}\leq q_{i,j'}$, however, there is no such a bound for $b_{i,j'}$. As a result, we may end up assigning high probabilities to undesired edges which can affect the expected cost of the matching $\M_n$. Now, the idea is to use some part of the matching $\M_n$ that is behaving well enough.
\begin{proposition}\cite[Proposition 7]{Aldous1992}\label{prop:nonrand}
		Let $Q_n = [q_{i,j'}]$ and $\ell_n=[\ell_n(i,j')]$ be given non-random $n\times n$ matrices. Suppose $200 \X(Q_n) \leq \theta < 1$. Consider the random matching $\M_n$ given as above and define the random set $D(\M_n) \coloneqq \{i\in[n]:\, b_{i,\M_n(i)} \leq \eta a_{i,\M_n(i)}\}$ where $\eta = \sqrt{3\X(Q_n)/\theta}$. Then
		\begin{align*}
		&\expect[\sum_{i\in D(\M_n)} q_{i,\M_n(i)} \ell_n(i,\M_n(i))] \leq (1 + \eta) \sum_{i\in[n]} \,\sum_{j'\in[n']} q_{i,j'}\ell_n(i,j'),\\
		&\prob[|D(\M_n)| \geq (1-\theta)n ] \geq 1 - \frac{3(1+\eta^{-1})\X(Q_n)}{\theta},
		\end{align*}
		that is, the random matching $\M_n$ restricted to the random set $D(\M_n)$ is a $(1-\theta)$ matching with high probability with cost close to the cost of $Q_n$.
		Specifically, there is a $(1-\theta)$ partial matching $\nu_n$(non-random) such that
		\begin{align*}
		\sum_{i \in U(\nu_n)} q_{i,\nu_n(i)} \ell_n(i,\nu_n(i)) \leq (1 + 4\sqrt{\X(Q_n)/\theta}) \sum_{i\in[n]} \,\sum_{j'\in[n']} q_{i,j'} \ell_n(i,j') .
		\end{align*}
\end{proposition}
Next, we patch the partial matching $\nu_n$ given by the Proposition \ref{prop:nonrand} (which exists for almost all realization of the edge weights and $Q_n$),  to construct a perfect matching without distorting the expected cost and the given partial matching too much. Aldous suggests using the greedy algorithm to do that. In the planted setting, the idea is to simply remove all the planted edges and then use the greedy algorithm. 
\begin{proposition}\cite[Proposition 9]{Aldous1992}\label{prop:swap}
	Fix $0< \theta < 1/10$, and let $ k  = \floor{\theta n}$. Let $\ell_n = [\ell_n(i,j')]$ denote the matrix of the edge weights. Let $\nu_n$ denote a
	 $1-\theta$ partial matching. Let $N_0 = [n]\setminus U(\nu_n)$ and $N_0' = [n']\setminus \nu_n(U(\nu_n))$ denote the set of unpaired vertices in either side. (Here, $N_0$, $N_0'$ and $\nu_n$ may depend on $\ell_n$ in an arbitrary way.) Then there exists a random subset $S\subset [n]\setminus N_0$ of size $k$, random bijections $\nu_1 : S \to N_0'$, $\nu_2 : N_0 \to \nu_n(S)$, and events $\Omega_n$ with $P(\Omega_n)\to 0$ such that
	\begin{align*}
	\limsup_{n\to\infty} \frac{1}{n}\expect\left[1_{\Omega_n^c} \left(\sum_{ i \in N_0} \ell_n({i,\nu_2(i)}) + \sum_{i\in S}\ell_n({i,\nu_1(i)}) \right)\right] \leq 24\, \theta^{1/2}.
	\end{align*}
\end{proposition}
It remains to combine Proposition \ref{prop:nonrand} and Proposition \ref{prop:swap} to construct a matching on $K_{n,n}$. The key idea is to rewrite the edge weights of $K_{n,n}$ as the minimum of two independent exponential random variables.
\begin{lemma}\label{lem:elem}
	If $X\sim \exp(\mu_1)$ and $Y\sim\exp(\mu_2)$ are independent, then $\min(X,Y)\sim \exp(\mu_1 + \mu_2)$.
\end{lemma}
Now, for any fixed $0<\alpha< 1$, we can write
\begin{align}
\ell_n(i,j') = \min\left(\frac{\ell^{1}_n(i,j')}{1-\alpha},\frac{\ell^{2}_n(i,j')}{\alpha} \right), \label{eq:cost_split}
\end{align}
where $\ell^{1}_n(i,j')$ and $\ell^{2}_n(i,j')$ are independent copies of $\ell_n(i,j')$. We use $\ell^{1}_n(i,j')$ to construct the partial matching, and then $\ell^{2}_n(i,j')$ to patch the partial matching and construct a complete matching. On the event $\Omega^*_n$ that the construction is not possible, that is either Proposition \ref{prop:nonrand} or Proposition \ref{prop:swap} failed, we can always use the planted matching to construct a matching on $K_{n,n}$. Since the probability of this failure goes to zero, this does not affect the cost of the matching that much. Specifically, we need to show that $L_n \coloneqq \frac{1}{n}\sum_{i\in[n]}\ell_n(i,i')$ is uniformly integrable.
\begin{lemma}
	There exists a function $\delta(\cdot)$ with $\delta(x)\to 0$ as $x \to 0$ such that for arbitrary events $\Omega^*_n$,
	\begin{align}
	\limsup_{n\to\infty} \expect[L_n \bs{1}_{\Omega^*_n}] \leq \delta(\epsilon) ,  \label{eq:planted_cost_integrable}
	\end{align}
	where $\epsilon = \limsup_{n\to\infty} \prob[\Omega^*_n]$.
	\begin{proof}
		Since $\expect[\left(L_n\right)^2] < \infty$, $L_n$ is uniformly integrable. In particular, by Cauchy-Schwarz inequality, we have $ \expect[L_n \bs{1}_{\Omega^*_n}] \le  \sqrt{\expect[\left(L_n\right)^2]  \prob[\Omega^*_n]} = 
		\sqrt{\expect[\left(L_n\right)^2] \epsilon}.$
	\end{proof}
\end{lemma}

Now, we are ready to present the formal proof of~\eqref{eq:hardhalf}, which closely follows the proof of Proposition 2 in \cite{Aldous1992}. Fix $0 <\theta <1$ and let $\epsilon=\theta^3/200$. Construct $Q_n=[q_{i,j'}]_{i,j'}$ as per \eqref{eq:construction_Qn} using the edge cost $\ell_n^{1}$. It follows from Lemma \ref{lem:2limits} that for all sufficiently large $x, B, H, n$, $\expect[\X(Q_n)] < \epsilon$ and 
\begin{align}
\frac{1}{n}\expect[\sum_{i,j'} q_{i,j'} \ell_n^1 (i,j')] \le \expect[\ell(\M_{\opt}(\root),\root)]  +\epsilon. 
\label{eq:Q_n_cost}
\end{align}
Define event $\Omega_{n}^1=\{200 \X(Q_n ) > \theta^2\}$. Then by Markov's inequality, 
$$
\prob[\Omega_n^1] = \prob[ 200 \X(Q_n ) > \theta^2 ] \le \frac{200 \expect[ \X(Q_n) ] }{\theta^2} \le \frac{200 \epsilon}{\theta^2} = \theta.
$$
Outside event $\Omega_{n}^1$, we have $200 \X(Q_n) \le \theta^2<\theta$. 
It follows from Proposition \ref{prop:nonrand} that outside event $\Omega_{n}^1$, 
 there exists a $(1-\theta)$ partial matching $\nu_n$  such that
		\begin{align}
		\sum_{i  \in U(\nu_n)} q_{i,\nu_n(i)} \ell_n^1(i,\nu_n(i)) \leq (1 + \theta ) \sum_{i\in[n]} \,\sum_{j'\in[n']} q_{i,j'} \ell_n^1(i,j'),
\label{eq:partial_cost}
		\end{align}
where $U(\nu_n) $ is the set of vertices $i$ matched under $\nu_n$.
Now, condition on a realization (outside of $\Omega_{n}^1$) of $\ell_n^1$, and apply Proposition \ref{prop:swap} to $\ell_n^2$.
Since $\ell_n^1$ and $\ell_n^2$ are independent, it follows from Proposition \ref{prop:swap} that there exist an event $\Omega_n^2$ and bijections $\nu_1$ and $\nu_2$ such that for all sufficiently large $n$, $\prob[\Omega_n^2] \le \theta$ and 
\begin{align}
 \frac{1}{n}\expect\left[1_{ (\Omega_n^1)^c} \left(\sum_{ i \in [n] \setminus U(\nu_n)
 } \ell^2_n({i,\nu_2(i)}) + \sum_{i\in  \nu_n^{-1} ( \nu_2 ([n] \setminus U(\nu_n) ) } 
 \ell^2_n({i,\nu_1(i)}) \right)\right] \leq 24\, \theta^{1/2}.
 \label{eq:patch_cost}
\end{align}
Outside event $\Omega_n^*=\Omega_n^1 \cup \Omega_n^2$, we can construct a complete matching 
$\pi_n: [n] \to [n']$ such that $\pi_n(i) = \nu_2(i)$ if $ i \in [n] \setminus U(\nu_n)$; and
$\pi_n(i) = \nu_1(i)$ if $i\in  \nu_n^{-1} ( \nu_2 ([n] \setminus U(\nu_n) )$; and 
$\pi_n(i) =\nu_n(i)$ otherwise.  On event $\Omega_n^*$, we just let $\pi_n$ to be the planted matching. 
Combining \eqref{eq:cost_split}, \eqref{eq:Q_n_cost}, \eqref{eq:partial_cost}, \eqref{eq:patch_cost}, and \eqref{eq:planted_cost_integrable} yields that
\begin{align*}
\limsup_{n\to \infty} \frac{1}{n} \expect[\sum_{i,j'} \pi_n( i, j')  \ell_n (i,j')] \le \frac{(1+\theta) \left( \expect[\ell(\M_{\opt}(\root),\root)]  +\epsilon \right) }{1-\alpha} + \frac{24 \theta^{1/2}}{\alpha} + \delta(2\theta). 
\end{align*}
Letting $\theta$ (and hence $\epsilon$) $\to 0$, then letting $\alpha \to 0$, we establish~\eqref{eq:hardhalf}.

%% file: Sections/Proofs/Lwc.tex
\noindent {\bf Step 1: The Exploration Process}

As we pointed out in the sketch of the proof, the first step is to define an exploration process that explores vertices of $N_{n,\circ}(1)$ in a series of stages. The stage $m$ of the exploration process reveals a rooted subtree of $(K_{n,n},\ell_n)$, denoted by $N_{n,\circ}[m]$. The root of $N_{n,\circ}[m]$ is vertex $1$, the number of un-planted children of every vertex is $m$, and the set of planted edges restricted to $N_{n,\circ}[m]$ is a matching. Next, we provide a formal construction of $N_{n,\circ}[m]$.

The construction begins with vertex $1$. Include the edge $\{1,\Mplanted_n(1)=1'\}$ in $N_{n,\circ}[m]$. Let $\{v_1,v_2,\cdots,v_m\}$ denote the $m$ closest un-planted neighbors of vertex $1$. Add all the edges $\{1,v_k\}$ and then $\{v_k,\Mplanted_n(v_k)\}$ to $N_{n,\circ}[m]$. Next, continue with the vertex $1'$. Let $\{w_1,w_2,\cdots,w_m\}$ denote the $m$ closest un-planted neighbors of vertex $1'$, among all the vertices that has not been added to $N_{n,\circ}[m]$. Include all the edges $\{1',w_k\}$ and then $\{w_k,\Mplanted_n(w_k)\}$ to $N_{n,\circ}[m]$. At each step of the construction, we follow two simple rules: $1)$ the next vertex to pick is the oldest one in $N_{n,\circ}[m]$; $2)$ when we add $m$ closest un-planted children of this vertex and their planted pairs, we avoid all the vertices that has already been added to $N_{n,\circ}[m]$.

The construction continues until we are about to pick a vertex at depth $m$, at which point it stops. Note that the only vertices at depth $m+1$ are the planted partners of the vertices at depth $m$. Let  $V_n^{(m)}$ and $E_n^{(m)}$ denote the set of the vertices and the edges of $N_{n,\circ}[m] $ respectively. Note that for all sufficiently large $n$, $|V_n^{(m)}|$ and $|E_n^{(m)}|$ are independent of $n$. Let $\mu_n^{(m)} \in \mathcal{P}(\Gstar)$ denote the law of $[N_{n,\circ}[m]]$. Figure \ref{fig:expproc} demonstrates the construction of $N_{n,\circ}[2]$ for the graph given by Figure \ref{fig:pmod}.

\begin{figure}[!t]
	\centering
	\begin{subfigure}[t!]{.45\linewidth}
		\begin{center}
			\includegraphics[scale=0.5]{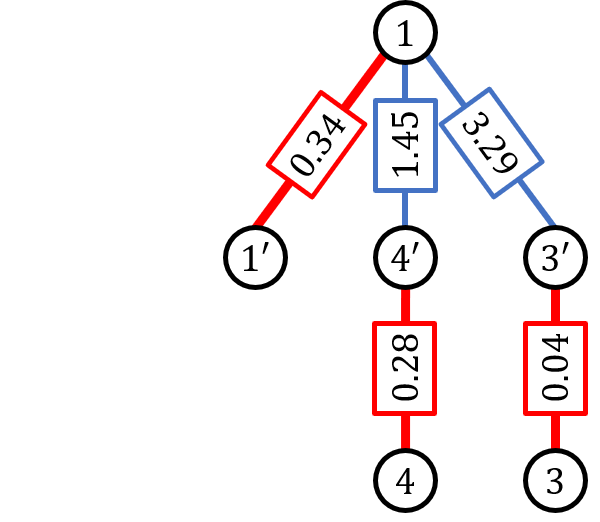}
			\caption*{Step 1 of the construction of $N_{n,\circ}[2]$.}
		\end{center}
	\end{subfigure}
	\hfill
	\begin{subfigure}[t!]{.45\linewidth}
		\begin{center}
			\includegraphics[scale=0.5]{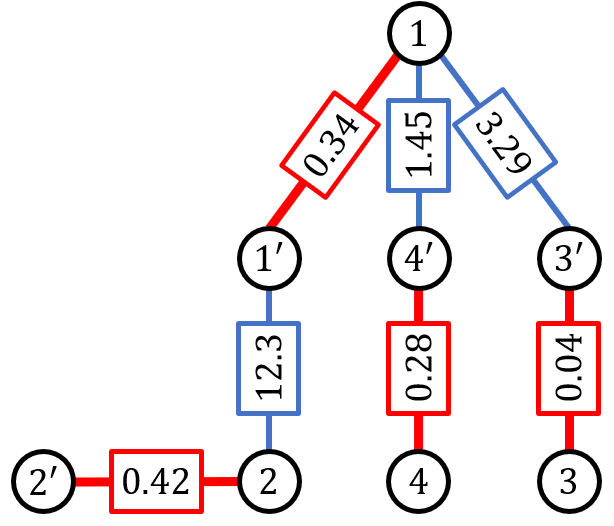}
			\caption*{Step 2 of the construction of $N_{n,\circ}[2]$.}
		\end{center}
	\end{subfigure}
	\caption{Stage $2$ of the exploration process on $(K_{4,4},\ell_4)$ given by Figure \ref{fig:pmod}.}\label{fig:expproc}
\end{figure}

\medskip

\noindent {\bf Step 2: A Total Variation Convergence}

The rooted planted tree $N_{n,\circ}[m]$ has the same graph structure as a truncated version of the planted PWIT: remove all vertices $\bs{i} = (i_1,i_2,\cdots i_l)$ such that either 
$(1)$ $i_s > m$ for some $s\in\{1,2,\cdots,l\}$, or 
$(2)$ $l>m+1$, or 
$(3)$ $l = m+1$ and $i_{l} \neq 0$. In particular, the number of un-planted children of every vertex in the truncated version is $m$, the depth of vertices are bounded by $m+1$, and the only vertices at depth $m+1$ are the planted pairs of the vertices at depth $m$. Let $N_{\infty}[m]$ denote the truncated planted PWIT, and let $\mu_{\infty}^{(m)}$ denote the law of $[N_{\infty}[m]]$. Now, using the same approach as in \cite[Lemma 10]{Aldous1992}, we show that $\mu_{n}^{(m)}$ converges to $\mu_{\infty}^{(m)}$ in total variation norm.
\begin{lemma}\label{lem:tvconv}
	For any fixed $m$, $\mu_n^{(m)} \xrightarrow{TV} \mu_\infty^{(m)}$ where the total variation convergence of positive measures is defined as follows:
	\begin{align*}
		\mu_n\xrightarrow{TV} \mu \text{ iff } \sup_A|\mu_n(A) - \mu(A)| \to 0 
	\end{align*}
\end{lemma}
\begin{proof}
	It is easy to see that $\mu_n^{(m)}$ is absolutely continuous with respect to $\mu_\infty^{(m)}$. Moreover, the Radon-Nikodym derivative of $\mu_n^{(m)}$ with respect to $\mu_\infty^{(m)}$ equals the likelihood ratio.
	
	Consider similar steps on the planted PWIT to construct $N_{\infty}[m]$. 
	Conditioned on the first $t-1$ steps of the construction of $N_{\infty}[m]$ and $N_{n,\circ}[m]$, we will calculate the ratio of the conditional densities for the next step of the construction. Since planted edges have the same $\exp(\lambda)$ distribution in both cases, we are only interested in the corresponding ratio of un-planted edges. 
	
	At $t$th step of the construction of $N_{\infty}[m]$, the conditional density of $(x_1,x_1+x_2,\cdots,x_1+\cdots+x_m)$ is 
	$\exp\left(-(x_1+x_2+\cdots+x_m)\right)$. At $t$th step of the construction of $N_{n,\circ}[m]$, using the memoryless property of exponential random variables, the conditional density of $(x_1,x_1+x_2,\cdots,x_1+\cdots+x_m)$ is
	\begin{align*}
	 \prod_{i=0}^{m-1} \frac{|I_{t-1}|-i}{n} \exp\left(-\frac{x_{i+1}(|I_{k-1}|-i)}{n}\right).
	\end{align*}
	where $I_{t-1}$ is the set of vertices that has not been added to $N_{\infty}[m]$ yet up to the $t$-th step.
	Hence, the ratio of conditional densities is at least
	\begin{align*}
	\prod_{i=0}^{m-1} \frac{|I_{k-1}|-i}{n} \geq \left(\frac{n - \left|V^{(m)}\right|}{n}\right)^{m},
	\end{align*}
	and we have
	\begin{align*}
	\frac{d \mu_n^{(m)}}{d \mu_\infty^{(m)}} \geq \left(1 - \frac{ \left|V^{(m)}\right|}{n}\right)^{\left|E^{(m)}\right| -  \left|V^{(m)}\right|/2} \, , 
	\end{align*}
	where the exponent $\left|E^{(m)}\right| -  \left|V^{(m)}\right|/2$ is the number of un-planted edges explored. 

Note that $\left|V^{(m)}\right|$ and $\left|E^{(m)}\right|$ do not depend on $n$, for all sufficiently large $n$. Hence, as $n\to\infty$ the right-hand side of the above inequality goes to $1$. Now, the result follows by the fact that $\mu_n^{(m)}$ and $\mu_\infty^{(m)}$ are probability measures.
\end{proof}

\noindent {\bf Step 3: Locally Tree-Like Property}

Fix some $\rho > 0$. Recall that $(G_{n,\circ}(1))_{\rho}$ denotes the $\rho$-neighborhood of vertex $1$ in $N_{n,\circ}(1)$ as is defined in Section~\ref{app:gen_setup}. Similarly, $(G_{n,\circ}[m])_\rho$ denotes the neighborhood $\rho$ of node $1$ in $N_{n,\circ}[m]$. The question is, whether these two neighborhoods are the same. Note that $(G_{n,\circ}[m])_\rho$ is a tree but $(G_{n,\circ}(1))_{\rho}$ is not necessary. However, it becomes a tree with high probability.
\begin{lemma} \label{lem:neigh}
	Fix $\epsilon > 0 $ and $\rho>0$. Then there exists large enough $m_0(\epsilon, \rho)$ such that for all fixed $m > m_0(\epsilon,\rho)$,
\begin{align*}
\prob[ (G_{n,\circ}(1))_{\rho}\neq (G_{n,\circ}[m])_\rho ]  \leq \epsilon, \quad  \text{ as }n\to\infty.
\end{align*}
\end{lemma}
\begin{proof}
	Let $m_0$ to be large enough such that 
	$$
	\prob[ \text{number of vertices in }(G_{n,\circ}(1))_{\rho} > m_0 ] < \epsilon/2.
	$$ 
	Fix $m> m_0$. Consider the event $\{(G_{n,\circ}(1))_{\rho} \neq (G_{n,\circ}[m])_\rho\}$. This event may happens if either the number of vertices in $(G_{n,\circ}(1))_{\rho}$ is greater than $m$ or there are two vertices $v, w \in V_n^{(m)}$ such that $\{v,w\}\notin E_n^{(m)}$ but $\ell_n(v,w) \leq \rho$. The probability of the first event is bounded by $\epsilon/2$. For the other event, note that if $v,w\in V_n^{(m)}$ and $\{v,w\}\in E_n \setminus E_n^{(m)}$, then $\ell_n(v,w)$ is dominated by an exponentially distributed random variable with mean $n$. 
	(To see this, assume that $v$ was revealed earlier than $w$. Then we know that the cost of the edge $\{v,w\}$ is larger than the cost $c_m$ of the edge $\{v,v_m\}$, where $v_m$ is the $m$th closest un-planted neighbor of $v$. Hence, the probability distribution of the weight of $\{v,w\}$, using the memoryless property, is $1/n\exp((x-c_m)/n)$ which is stochastically larger than a random variable distributed as $\exp(1/n)$. Hence,
	\begin{align*}
	\prob[\exists v,w \in V_n^{(m)} \text{ such that } \{v,w\}\notin E_n^{(m)} \text{ and }\ell_n(v,w)\leq \rho] \\
	\quad \leq 
	{|V_n^{(m)}| \choose 2} (1 - \exp( -t/n)) \to 0 \text{ as } n\to\infty.
	\end{align*}
	As we mentioned before, $|V_n^{(m)}|$ is independent of $n$ for all sufficiently large $n$. Now, the result follows by combining the last two displayed inequalities.
\end{proof}
Now, combining Lemma \ref{lem:tvconv} and Lemma \ref{lem:neigh}, we get the following corollary.
\begin{corollary}\label{cor:tvconv}
	Fix $\rho>0$. Let $\mu_{n,\rho}$ denote the law of $[((G_{n,\circ}(1))_{\rho},\ell_n)]$, and let $\mu_{\infty,\rho}$ denote the law of $[((G_\infty)_{\rho},\ell_\infty)]$. Then $\mu_{n,\rho}\xrightarrow{TV}\mu_{\infty,\rho}$.
\end{corollary}

\medskip

\noindent {\bf Step 4: Portmanteau Theorem}

For a fixed $R>0$, since the condition $d([N_{\circ}],[T_{\circ}]) < (R+1)^{-1}$ is equivalent to $d([((G_{\circ})_R,\ell)],[T_{\circ}])\allowbreak < (R+1)^{-1}$, 
Corollary~\ref{cor:tvconv} implies that, for all finite rooted planted trees $[T_{\circ}] \in \Gstar$, we have
\begin{align}\label{eq:totalvarconv}
\left|\mu_n (A_{T_{\circ}}) - \mu_{\infty}(A_{T_{\circ}})\right| \to 0 \text{ as }n\to\infty \, ,
\end{align}
where $A_{T_{\circ}}$ is defined as
\begin{align*}
A_{T_{\circ}} \coloneqq \{[N_{\circ}]\in\Gstar: d([N_{\circ}],[T_{\circ}]) < (R+1)^{-1} \}.
\end{align*}
Note that the support of $\mu_{\infty}$ is rooted planted trees. Moreover, recall that $\Gstar$ is separable, hence, the restriction of $\Gstar$ to the rooted planted trees is also separable. Since $\mu_{\infty}$ is a probability measure, for any $R>0$ and any $\epsilon > 0$, there exists a finite set  $S(R,\epsilon)$ consisting of 
rooted planted trees $T_{\circ} = (G_{\circ},\ell)$ with $(G_{\circ})_R = G_{\circ}$ such that
\begin{align*}
\mu_{\infty} \left(\bigcup_{T_{\circ} \in S(R,\epsilon)} A_{T_{\circ}} \right) > 1-\epsilon.
\end{align*}
Using Corollary \ref{cor:tvconv}, there exists $n_0(\epsilon)\in\mathbb{N}_+$ such that for all $n > n_0(\epsilon)$,
\begin{align*}
\mu_{n} \left(\bigcup_{T_{\circ} \in S(R,\epsilon)} A_{T_{\circ}} \right) > 1-2\epsilon .
\end{align*}

Now, we are going to prove that $\mu_n \xrightarrow{w} \mu_{\infty}$. By definition $\mu_n \xrightarrow{w} \mu_{\infty}$, if for any continuous bounded function $f:\Gstar \to \mathbb{R}$,
\begin{align*}
\int_{\Gstar} f\,d\mu_n \to \int_{\Gstar} f\,d\mu_{\infty}.
\end{align*}
Using the Portmanteau Theorem, we can restrict our attention to the uniformly continuous bounded functions. Let $f:\Gstar \to \mathbb{R}$ be a uniformly continuous bounded function. Now, for any $\epsilon > 0$, there is a $\delta > 0$ such that if $d([\Ns],[\Nsp]) < \delta$ then $|f([\Ns]) - f([\Nsp]) | \leq \epsilon$. Fix the value of $\epsilon$ and let $R$ to be large enough such that $(R+1)^{-1} < \delta$. We have,
\begin{align*}
\left|\int_{\Gstar} f d\mu_n - \int_{\Gstar} fd\mu_\infty\right| \leq 3\epsilon |f|_{\infty} + \sum_{T_{\circ} \in S(R,\epsilon)} f(T_{\circ}) \left|\mu_n(A_{T_{\circ}}) - \mu_\infty(A_{T_{\circ}}) \right| + 2\epsilon.
\end{align*}
where $|f|_{\infty} \coloneqq \sup_{N\in\Gstar} |f(N)|$. The result follows by arbitrary choice of $\epsilon$, the fact that $|S(R,\epsilon)| < \infty$, and \eqref{eq:totalvarconv}.

%% file: Sections/Proofs/Involinvar.tex
By Lemma~\ref{lem:match}, $\Minfopt$ is a deterministic function of the collection of random variables $$\mathcal{C}_\infty = \left\{\ell_{\infty}(e),X(\overleftrightarrow{e});e\in E_{\infty}, \text{ and }\overleftrightarrow{e}\text{ is directed}\right\}.$$Also note that $X(\overleftrightarrow{e})$ satisfies \eqref{eq:baleq}, which does not depend on the relabeling of the vertices. Now, by construction, we only need to show that the distribution of $(X(\overrightarrow{e});e\in \overrightarrow{E}_{\infty}(h))$ for $h \geq  1$ is invariant with respect to the involution map $\iota$.

Abusing the notation, let $\widehat{\mu}_{\infty,\opt}$ denote the law of $\mathcal{C}_\infty$ and define $\widetilde{\mu}_{\infty,\opt}$ to be $\widehat{\mu}_{\infty,\opt} \times \texttt{count} \text{ on }\mathbb{N}_+$, similar to \eqref{eq:involinvardef}. Fix the second root $k \geq 0$ and let
\begin{align}
B = \left\{(X(\overrightarrow{e}),\overrightarrow{e}\in \overrightarrow{E}_{\infty}(h))\in \cdot \, \text{, $k$ is distinguished}\right\}
\end{align}
denote a measurable subset on $\Gmadstarstar$, 
where $\Gmadstarstar$ is defined similar to $\Gstarstar$ (the set all isomorphism classes of connected locally finite doubly-rooted planted networks) with an additional mark on the edges $\overleftrightarrow{e}$ representing $X(\overleftrightarrow{e})$.
Let $(X_0,Y_0)$ denote a solution of the system of recursive distributional equations \eqref{eq:disteqX}--\eqref{eq:disteqY} as in Lemma~\ref{lem:match}. Recall that if $\overrightarrow{e}$ is a planted edge then $X(\overrightarrow{e})$ and $X_0$ have the same distribution; otherwise, $X(\overrightarrow{e})$ and $Y_0$ have the same distribution. Note that the collection of random variables $(X(\overrightarrow{e}),\overrightarrow{e} \in \overrightarrow{E}_{\infty}(h))$ are independent since they depend on messages received from disjoint subtrees.

We need to show that 
\begin{align*}
\widetilde{\mu}_{\infty,\opt}(\iota^{-1}(B)) = \widetilde{\mu}_{\infty,\opt}(B).
\end{align*}
We treat the cases $k=0$ or $k > 0$ separately.
\begin{enumerate}[label=(\roman*)]
	\item If $k = 0$, then we have
	\begin{align*}
	\iota^{-1}(B) = \big\{\text{the double root is $(\root,0)$} \big\} \cap \big\{(X(\overrightarrow{e}),\overrightarrow{e}\in \overrightarrow{E}_{\infty}(h,0))\in \cdot \,\big\} ,
	\end{align*}
	where $\overrightarrow{E}_{\infty}(h,0) \coloneqq \{(v,vj): \{v,vj\} \in E_{\infty} \text{ s.t. } v=i_1 i_2 i_3 \cdots i_{h-2} \text{ with } i_1\neq 0 \text{ or } v=0i_2i_3\cdots i_{h} \}$. 
	Now to complete the proof, it suffices to show that 
$(X(\overrightarrow{e}),\overrightarrow{e}\in \overrightarrow{E}_{\infty}(h,0))$ has the same distribution as $(X(\overrightarrow{e}),\overrightarrow{e}\in \overrightarrow{E}_{\infty}(h))$.
Clearly, the collection of random variables $(X(\overrightarrow{e}),\overrightarrow{e}\in \overrightarrow{E}_{\infty}(h,0))$ are independent, and 
$X(\overrightarrow{e})$ has the same distribution as $X_0$ or $Y_0$ depending on whether $\overrightarrow{e}$ is planted or un-planted. Thus it remains to prove that 
there is a one-to-one map from $\overrightarrow{E}_{\infty}(h,0)$ to $\overrightarrow{E}_{\infty}(h)$ that maps (un)planted edges to (un)planted ones. 
Consider the relabeling function $\phi$ defined as follows
	\begin{align*}
	&\phi(0i_2i_3\cdots i_{h}) = i_2i_3\cdots i_{h}&&\forall\, 0i_2i_3\cdots i_{h} \in V_\infty, \\
	&\phi(i_1i_2\cdots i_{h-2}) = 0i_1i_2\cdots i_{h-2}&&\forall\, i_1i_2\cdots i_{h-2} \text{ with }i_1\neq 0,
	\end{align*}
	and define $\gamma:\overrightarrow{E}_{\infty}(h,0) \to \overrightarrow{E}_{\infty}(h)$ by $\gamma(\{v,vj\}) \coloneqq \{\phi(v),\phi(v)j \}$.
	\item If $k > 0 $, then we have
	\begin{align*}
	\iota^{-1}(B) = \bigcup_{l>0} \bigg\{ \big\{\text{the double root is $(\root,l)$} \big\} \cap A_l \cap \big\{(X(\overrightarrow{e}),\overrightarrow{e}\in \overrightarrow{E}_{\infty}(h,l))\in \cdot \,\big\} \bigg\},
	\end{align*}
	where
	\begin{align*}
	&A_l \coloneqq \big\{ \ell_{\infty}(l,l(k-1)) < \ell_{\infty}(\root,l) < \ell_{\infty}(l,lk) \big\},
	\end{align*}
	and $\overrightarrow{E}_{\infty}(h,l)$ is defined similar to $\overrightarrow{E}_{\infty}(h,0)$. Note that the events $\{(X(\overrightarrow{e}),\overrightarrow{e}\in \overrightarrow{E}_{\infty}(h,l))\in \cdot \,\} $ and $A_l$ are independent, and the distribution of $(X(\overrightarrow{e}),\overrightarrow{e}\in \overrightarrow{E}_{\infty}(h,l))$ does not depend on $l$. Also, note that by \eqref{eq:uniform-i}, we have
\[
	\sum_{l>0} \prob[A_l]
	= \int_{0}^{\infty} \prob[\text{exactly $k-1$ arrivals before $x$}] \,\dx 
	= 1.
\]
The result then follows using the same argument as in the previous case.
\end{enumerate}

%% file: Sections/Proofs/Ineqnoteq.tex
There are two branches that we are interested in: the alternating path from $\root$ through $v_{-1}$, and the alternating path from $\root$ through $v_1$. It is more convenient to study these two branches on the doubly-rooted planted PWIT, rooted at $(\root,v_1)$. 
The proof uses the discussion of bi-infinite planted PWITs $T^{\lra}_u$ and $T^{\lra}_p$ in Section~\ref{app:BiinfTree}. We follow the same notation and simplification (Remark \ref{rem:simplif} in Section~\ref{app:BiinfTree}) here.

Using the relabeling maps $\psi\big|_{p}$ and $\psi\big|_{u}$, we already know that $\mu_{\infty}\times \delta_0$ is equivalent to $\mu^{\lra}_{p}$ and $\mu_{\infty}\times\texttt{count} \text{ on } \{1,2,3,\cdots \}$ is equivalent to $\mu^{\lra}_{u}$.
We can use the relabeling map $\psi\big|_{\cdot}$ to define $\{X_{\cdot}(\overrightarrow{e}),\overrightarrow{e} \text{ is a directed edge in } E^\lra_{\cdot} \}$ jointly with $\{\ell^{\lra}_{\cdot}(e),e\in E^\lra_{\cdot} \}$. Note that the joint distribution of $\{\ell^{\lra}_{\cdot}(e), X_{\cdot}(\overrightarrow{e}); e \in E^\lra_{\cdot} \text{ and }\overrightarrow{e}\text{ is directed} \}$ is exactly the same as if we use the construction of Lemma \ref{lem:const} by redefining $\overrightarrow{E}(h)$ as
\begin{align}\label{eq:Mconst_lra}
	\overrightarrow{bE}^{\lra}_{\cdot}(h) \coloneqq \{\overrightarrow{e}=(-v,-vj): \text{gen}(v) = h-1\} \cup  \{\overrightarrow{e}=(+v,+vj): \text{gen}(v) = h-1\}.
\end{align}
Now, given $\{X_{\cdot}(\overrightarrow{e}),\overrightarrow{e} \text{ is a directed edge in } E^\lra_{\cdot} \}$ we can define a minimum matching $\M^{\lra}_{\cdot,\opt}$ on $T^{\lra}_{\cdot}$, same as in Lemma \ref{lem:match}, i.e., 
\begin{align}
&\forall v\in V^{\lra}_{\cdot}: \M^{\lra}_{\cdot,\opt}(v) = \argmin_{w:\{v,w\}\in E^{\lra}_{\cdot}} \left(\ell^{\lra}_{\cdot}(v,w) - X_{\cdot}(v,w) \right)\label{eq:Mdef_lra}\\
&\forall e\in E^{\lra}_{\cdot}: \M^{\lra}_{\cdot,\opt}(e) = 1 \text{ if and only if } \ell^{\lra}_{\cdot}(e) < X_{\cdot}(\overrightarrow{e}) + X_{\cdot}(\overleftarrow{e}).\label{eq:Mdefalt_lra}
\end{align}
Next, we are going to show that the bi-infinite tree $T^{\lra}_{u}$ ($T^{\lra}_{p}$) restricted to $\M^{\lra}_{u,\opt}(-\root,+\root) = 1$ ($\M^{\lra}_{p,\opt}(-\root,+\root) = 1$) is equivalent to the doubly-rooted planted PWIT, rooted at $\{\root,\M_{\opt}(\root)\}$, restricted to $\M_{\opt}(\root) \neq 0$ ($\M_{\opt}(\root) = 0$).

On the planted PWIT $T_u$ (recall that $T_u=T(\root)$ as defined in Section~\ref{app:BiinfTree}), define $X_{u}^{\downarrow} = \min_{i \geq 0} (\ell(\root,i) - X(\root,i))$. Let $\nu_u(x)$ denote the conditional distribution of the set $$\{\ell(e), X(\overrightarrow{e}); e \in E_{h} \text{ and }\overrightarrow{e}\text{ is directed away from }\root \}$$ given $X_{u}^{\downarrow} = x$. 
Similarly, on the subtree $T_p$ (recall that $T_p$ is a relabeling of $T(0)$ as defined in Section~\ref{app:BiinfTree}), define $X_{p}^{\downarrow} = \min_{i \geq 1} (\ell(\root,i) - X(\root,i))$. Let $\nu_p(x)$ denote the conditional distribution of the set $$\{\ell(e), X(\overrightarrow{e}); e \in E_p \text{ and }\overrightarrow{e} \text{ is directed away from } \root\}$$ given $X_{p}^{\downarrow} = x$. 
On the bi-infinite tree $T^{\lra}_{\cdot}$, define $\mu^{1}_{\cdot}$ to be the measure obtained by restricting $\mu^{\lra}_{\cdot}$ to the set $\{\ell^{\lra}_{\cdot}(-\root, +\root) < X_{\cdot}(-\root, +\root) + X_{\cdot}(+\root, -\root)\}$, i.e., $\M^{\lra}_{\cdot,\opt}(-\root,+\root) = 1$. Let $E^{\lra +}_{\cdot}$ and $E^{\lra -}_{\cdot}$ denote all edges of form $\{+v,+vj\}$ and $\{-v,-vj\}$ respectively. Clearly, $E^{\lra}_\cdot = E^{\lra +}_{\cdot} \cup E^{\lra -}_{\cdot} \cup \{-\root,+\root \}$.
Let $(X_0,Y_0)$ be a solution of the system of recursive distributional equations~\eqref{eq:disteqX}--\eqref{eq:disteqY}.
\begin{lemma}\label{lem:changepres}
	The measures $\mu^{1}_{u}$ and $ \mu^{1}_{p}$ are finite positive measures. The total mass of $\mu^{1}_{u}$ equals $\prob[\M_{\opt}(\root) \neq 0]$ and the total mass of $\mu^{1}_{p}$ is $\prob[\M_{\opt}(\root) = 0]$.
	Under $\mu^{1}_{\cdot}$ we have:
	\begin{enumerate}[label=(\roman*)]
		\item \sloppy The joint density of $\left(\ell^{\lra}_{u}(-\root,+\root),X_{u}(+\root,-\root),X_{u}(-\root,+\root)\right)$ at point $(l,x_1,x_2)$ is $f_{u}(x_1) f_{u}(x_2)\bs{1}_{(0<l<x_1+x_2)}$, where $f_u(\cdot)$ is the density of $Y_0$ and $\bs{1}_{(0<l<x_1+x_2)}$ is the indicator function; the joint density of $\left(\ell^{\lra}_{p}(-\root,+\root),X_{p}(+\root,-\root),X_{p}(-\root,+\root)\right)$ at point $(l,x_1,x_2)$ is $f_{p}(x_1)f_{p}(x_2)\lambda\exp(-\lambda l)$, where $f_p(\cdot)$ is the density of $X_0$. 
		\item Conditioned on $\left(\ell^{\lra}_{\cdot}(-\root,+\root),X_{\cdot}(+\root,-\root),X_{\cdot}(-\root,+\root)\right) = (l,x_1,x_2)$ with $x_1+x_2 > l$, the distribution of the family
		\begin{align*}
		\{\ell^{\lra}_{\cdot}(e), X_{\cdot}(\overrightarrow{e}); e \in E^{\lra -}_{\cdot} \text{ and }\overrightarrow{e} \text{ is directed away from } -\root \}
		\end{align*}
		is the image of $\nu_{\cdot}(x_1)$ under the natural embedding $T_{\cdot}\to T_{\cdot}^{-} \subset  T_{\cdot}^\lra$; the distribution of the family
		\begin{align*}
		\{\ell^{\lra}_{\cdot}(e), X_{\cdot}(\overrightarrow{e}); e \in E^{\lra +}_{\cdot} \text{ and }\overrightarrow{e} \text{ is directed away from } +\root \}
		\end{align*}
		is the image of $\nu_{\cdot}(x_2)$ under the natural embedding $T_{\cdot}\to T_{\cdot}^+ \subset  T_{\cdot}^\lra$ and these two families are conditionally independent. 
	\end{enumerate}
\end{lemma}
\begin{remark}\label{rem:onedirection}
	Conditioned on $\left(\ell^{\lra}_{\cdot}(-\root,+\root),X_{\cdot}(+\root,-\root),X_{\cdot}(-\root,+\root)\right) = (l,x_1,x_2)$ with $x_1+x_2 > l$, we have $ \M^{\lra}_{\cdot,\opt}(-\root,+\root) = 1$. Now, by \eqref{eq:Mdef_lra} and the construction of $X_{\cdot}$ on $T_{\cdot}^\lra$, we have $X_{\cdot}(-i,-\root) = \ell^{\lra}_{\cdot}(-\root,+\root) - X_{\cdot}(-\root,+\root)$ for all $i$. Similarly, $X_{\cdot}(+i,+\root) = \ell^{\lra}_{\cdot}(+\root,-\root) - X_{\cdot}(+\root,-\root)$ for all $i$. This combined with the families in part (ii) of Lemma \ref{lem:changepres}, specifies $X_{\cdot}$ on $T_{\cdot}^\lra$ under $\mu^{1}_{\cdot}$.
\end{remark}
\begin{proof}
	By construction of $X_{\cdot}$ on $T_{\cdot}^\lra$, we already know that $X_{\cdot}(+\root,-\root)$ and $X_{\cdot}(-\root,+\root)$ are independent with density $f_\cdot(\cdot)$. Moreover, $\ell^{\lra}_{u}(-\root,+\root)$ has uniform ``distribution'' on $[0,\infty)$, and $\ell^{\lra}_{p}(-\root,+\root)$ is an exponentially distributed random variable with parameter $\lambda$.
	Hence, the joint density has the form mentioned in (i). Moreover, the total mass of $\mu^{1}_{u}$ is
	\begin{align}\label{eq:eq1Y0}
	\int_{x_1= -\infty}^{\infty}\int_{x_2 = -\infty}^{\infty} (x_1+x_2)^+ f_{u}(x_1)f_{u}(x_2)\,dx_2\,dx_1,
	\end{align}
	and the total mass of $\mu^{1}_{p}$ is
	\begin{align}\label{eq:eq1X0}
	\int_{x_1= -\infty}^{\infty}\int_{x_2 = -\infty}^{\infty} (1 - \exp( -\lambda (x_1+x_2)^+)) f_{p}(x_1)f_{p}(x_2)\,dx_2\,dx_1,
	\end{align}
	where $(x_1+x_2)^+ = \max(x_1+x_2,0)$. Now, using the joint density above, we can calculate the total mass of $\mu^{1}_{u}$ and $\mu^{1}_{p}$ as shown by the following lemma.
	\begin{lemma}\label{lem:eq1}
	The equation \eqref{eq:eq1Y0} equals $\prob[\M_{\opt}(\root) \neq 0]$, and the equation \eqref{eq:eq1X0} equals $\prob[\M_{\opt}(\root) = 0]$.
	\end{lemma}
	\begin{proof}[proof of Lemma \ref{lem:eq1}]
		Let $X_1$ and $X_2$ denote two independent copies of $X_0$, and $Y_1$ and $Y_2$ denote two independent copies of $Y_0$. Let $\eta$ denote an exponentially distributed random variable with parameter $\lambda$. 
		
		Using Corollary \ref{cor:iidX}, we have
		\begin{align*}
		\text{Equation \eqref{eq:eq1Y0}} &= \expect[(Y_1+Y_2)^+]\\
		&= \int_{x=0}^{\infty} \prob[Y_1+Y_2 > x] \,dx  \\
		&= \int_{\zeta=0}^{\infty} \prob[X(\overrightarrow{e})+X(\overleftarrow{e}) > \ell(e) \,|\, \exists e = \{\root,i\geq 1\}, \ell(e) = \zeta] \,d\zeta\\
		&= \prob[\M_{\opt}(\root) \neq 0], \\
				\text{Equation \eqref{eq:eq1X0}} &= \prob[X_1 + X_2 > \eta] = \prob[X(\root,0) + X(0,\root) > \ell(\root,0)]=\prob[\M_{\opt}(\root) = 0].
		\end{align*}
	\end{proof}
	Next, by construction based on $(X_\cdot(\overrightarrow{e}), \overrightarrow{e} \in \overrightarrow{bE}^{\lra}_{\cdot}(h))$ as in Lemma \ref{lem:const}, under $\mu^{\lra}_{\cdot}$ the families 
	\begin{align*}
	\{\ell^{\lra}_{\cdot}(e), X_{\cdot}(\overrightarrow{e}); e \in E^{\lra -}_{\cdot} \text{ and }\overrightarrow{e} \text{ is directed away from } -\root \} \cup X_{\cdot}(+\root,-\root),
	\end{align*} 
	and
	\begin{align*}
	\{\ell^{\lra}_{\cdot}(e), X_{\cdot}(\overrightarrow{e}); e \in E^{\lra +}_{\cdot} \text{ and }\overrightarrow{e} \text{ is directed away from } +\root \} \cup X_{\cdot}(-\root,+\root),
	\end{align*} 
	are independent of each other and $\ell^{\lra}_{\cdot}(-\root,+\root)$. Therefore, 
	the desired conditional independence in part (ii) follows,
	when conditioned on $\left(\ell^{\lra}_{\cdot}(-\root,+\root),X_{\cdot}(+\root,-\root),X_{\cdot}(-\root,+\root)\right) = (l,x_1,x_2)$. 
	 
	Finally, note that each families under $\mu^{\lra}_{\cdot}$ is distributed as the image of corresponding family on $T_{\cdot}$, where $X_\cdot(+\root,-\root)$ (or $X_\cdot(-\root,+\root)$) corresponds to $X^{\downarrow}_{\cdot}$. Now, the independence of these two families under $\mu^{\lra}_{\cdot}$ implies that the conditional distribution of families under $\mu^{1}_{\cdot}$ depends only on the corresponding value of $X^{\downarrow}_{\cdot}$, i.e., $x_1$ for the first family and $x_2$ for the second one. 
\end{proof}
Recall that $\psi\big|_p : [0,\infty)^E \times \{0\} \to [0,\infty)^{E^{\lra}_p}$ maps $\mu\times \delta_{0}$ to $\mu^{\lra}_{p}$, and $\psi\big|_u : [0,\infty)^E \times \{1,2,3,\cdots \} \to [0,\infty)^{E^{\lra}_u}$ maps $\mu\times \texttt{count} \text{ on } \{1,2,3,\cdots \}$ to $\mu^{\lra}_{u}$. Note that the inverse image of the event $\{\ell^{\lra}_{\cdot}(-\root, +\root) < X_{\cdot}(+\root, -\root) + X_{\cdot}(-\root, +\root)\}$ under $\psi\big|_\cdot$ is the event $$\{(\root,\M_{\opt}(\root))\text{ is the double root}\}.$$ Hence, $\psi\big|_p^{-1}$ maps the measure $\mu^{1}_{p}$ to $\mu\times \delta_{0}$ restricted to $\{\text{the second root is }\M_{\opt}(\root) = 0\}$, and $\psi\big|_u^{-1}$ maps the measure $\mu^{1}_{u}$ to $\mu\times  \texttt{count} \text{ on } \{1,2,3,\cdots \}$ restricted to $$\{\text{the second root is } \M_{\opt}(\root)  \text{ and }\M_{\opt}(\root) \neq 0\}.$$

Hence, to study the events $\xbar{B}_\infty$ and $B_{-1}$ on the planted PWIT, we can relabel the vertices by setting $\root$ to be $-\root$, $\M_{\opt}(\root)$ to be $+\root$, map the doubly-rooted planted PWIT, rooted at $(\root,\M_{\opt}(\root))$, to the corresponding bi-infinite tree ($T_{u}^\lra$ or $T_{p}^\lra$ depending on whether $\M_{\opt}(\root) = 0$ or not), and study the image of these events under $\mu^{1}_{p}$ and $\mu^{1}_{u}$.

\subsection{A Lemma on $T_u$ and $T_p$}
Before analyzing the image of $B_{-1}$ on $T_{u}^\lra$ (or $T_{p}^\lra$ depending on whether $\M_{\opt}(\root) = 0$ or not), let us present a technical lemma which generalizes Lemma 23 in \cite{Aldous1992} to the planted case.
\begin{lemma}\label{lem:g}
	On the planted PWIT, define
	\begin{align*}
	&X_{u}^{\downarrow} = \min_{i \geq 0} (\ell(\root,i) - X(\root,i)),\\
	&X_{p}^{\downarrow} = \min_{i \geq 1} (\ell(\root,i) - X(\root,i)),\\
	&I_u = \arg{\min_{i\geq 0}} (\ell(\root,i) - X(\root,i)),\\
	&I_p = \arg{\min_{i\geq 1}} (\ell(\root,i) - X(\root,i)).
	\end{align*}
	For $-\infty < b < a < \infty$ define 
	\begin{align*}
	&g_u(a,b) \coloneqq \prob[\ell(\root,I_u) - b > {\mint_{j: \{I_u,I_uj\} \in E}} (\ell(I_u,I_uj) - X(I_u,I_uj) )\,|\,X_u^\downarrow = a], \\
	&g_p(a,b) \coloneqq \prob[\ell(\root,I_p) - b > {\mint_{j: \{I_p,I_pj\} \in E}} (\ell(I_p,I_pj) - X(I_p,I_pj) )\,|\,X_p^\downarrow = a].
	\end{align*}
	Then $g_u(a,b),g_p(a,b) > 0$.
\end{lemma} 
\begin{proof}
	Let us begin with an observation, which is the continuous analogous of the splitting property of a Poisson process.
	\begin{observation}
		Let $\{X_i\}_{i=1}^\infty$ be independent real-valued continuous random variables with common distribution $\alpha_X$. Let $\{\zeta_i\}_{i=1}^\infty$ denote the arrivals of a Poisson process with parameter 1. Then $\{(\zeta_i,X_i)\}_{i=1}^\infty$ forms a Poisson point process on $[0,\infty)\times (-\infty,+\infty)$ with mean intensity $\beta(z,x)dzdx = dz\alpha_X(dx)$. 
		
		Now, let $Y_i = \zeta_i - X_i$. The set of points $\{Y_i\}_{i=1}^\infty$ forms a certain inhomogeneous Poisson process on $(-\infty,\infty)$ with mean intensity $\gamma(y)dy = \alpha_X([-y,\infty))dy$. Finally, it is easy to see that conditioned on the time of the first arrival to be $y_0$, the other points in $\{Y_i \}_{i=1}^\infty$ are the points of a certain inhomogeneous Poisson process on $(y_0,\infty)$. Similar statement holds if we condition on no arrival before $y_0$.
	\end{observation}
Note that by \eqref{eq:baleq}, we have
\begin{align*}
X(\root,I_\cdot) = {\min_{j: \{I_\cdot,I_\cdot j\} \in E}} (\ell(I_\cdot ,I_\cdot j) - X(I_\cdot ,I_\cdot j) ).
\end{align*}
Now, by the above observation, the set of points $\{\ell(I_\cdot ,I_\cdot j) - X(I_\cdot ,I_\cdot j), j\geq1 \}$ conditioned on $X(\root,I_\cdot ) = x$, are the points of a certain inhomogeneous Poisson process on $(x,\infty)$ (note that the claim is true regardless of whether ${\argmin_{j: \{I_\cdot,I_\cdot j\} \in E}} (\ell(I_\cdot ,I_\cdot j) - X(I_\cdot ,I_\cdot j) ) = 0$ or not). Hence,
\begin{align*}
\prob[{\mint_{j: \{I_\cdot,I_\cdot j\} \in E}} (\ell(I_\cdot,I_\cdot j) - X(I_\cdot,I_\cdot j) ) \in [y,y+dy] \,|\,X(\root,I_\cdot) = x] \geq \widetilde{\beta}_x(y)dy,
\end{align*}
where $\widetilde{\beta}_x(y) > 0$ for all $y > x$. Since the above term does not depend on the value of $\ell(\root,I_\cdot )$, we have
\begin{align*}
&\widetilde{g}_\cdot(a,b,x) \coloneqq \\
&\qquad \prob[\ell(\root,I_\cdot) - b > {\mint_{j: \{I_\cdot,I_\cdot j\} \in E}} (\ell(I_\cdot,I_\cdot j) - X(I_\cdot,I_\cdot j) ) \,|\,X(\root,I_\cdot) = x, \ell(\root,I_\cdot) = a+x] >0,
\end{align*}
for all $-\infty < b < a < \infty $, and $-\infty < x < \infty$. Now, the result follows by
\begin{align*}
g_\cdot (a,b) = \expect[\widetilde{g}_\cdot (a,b,X(\root,I_\cdot))\,|\,X_\cdot ^\downarrow = a],
\end{align*}
since $X_\cdot ^\downarrow = \ell(\root,I_\cdot) - X(\root,I_\cdot)$.
\end{proof}
\subsection{Calculating with the Bi-infinite Tree}
On the bi-infinite tree $T_{\cdot}^\lra$, define the event $C_{\cdot,-1}$ as
\begin{align*}
C_{\cdot,-1} \coloneqq \left\{ -\root = \argmint_{y:\{y,-I_{\cdot}\}\in E^{\lra}_\cdot} \left(\ell^{\lra}_\cdot(-I_{\cdot},y) - X_{\cdot}(-I_{\cdot},y)\right) \right\},
\end{align*}
where 
\begin{align*}
-I_{\cdot} = \argmin_{-i:\{-\root,-i\}\in E^{\lra}_\cdot}(\ell^{\lra}_\cdot(-\root,-i) - X_{\cdot}(-\root,-i)).
\end{align*}
The event $C_{\cdot,-1}$ under $\mu^{1}_{\cdot}$ corresponds to the event $B_{-1}$ on the doubly-rooted planted PWIT, rooted at $(\root,v_1)$, where $v_1 = \M_\opt(\root)$. Define the following $\sigma$-algebras:
\begin{align*}
&\mathcal{F}_{\cdot}^- = \sigma\left(X_\cdot(\overrightarrow{e}),\ell_{\cdot}^{\lra}(e): e\in E^{\lra-}_{\cdot} \text{ and $\overrightarrow{e}$ is directed} \right),\\
&\mathcal{F}_{\cdot}^+ = \sigma\left(X_\cdot(\overrightarrow{e}),\ell_{\cdot}^{\lra}(e): e\in E^{\lra+}_{\cdot} \text{ and $\overrightarrow{e}$ is directed} \right),\\
&\mathcal{F}_{\cdot}^{\root} = \sigma\left(\ell_{\cdot}^{\lra}(-\root,+\root),X_\cdot(+\root,-\root),X_\cdot(-\root,+\root) \right).
\end{align*}
\begin{lemma}\label{lem:mu1calc}
	$\mu^{1}_{\cdot}(C_{\cdot,-1}^c\big|\mathcal{F}_{\cdot}^+,\mathcal{F}_{\cdot}^{\root}) = g_\cdot\left(X_{\cdot}(+\root,-\root),\ell^{\lra}_{\cdot}(-\root,+\root) - X_{\cdot}(-\root,+\root)\right)$, where $C_{\cdot,-1}^c$ is the complement of the event $C_{\cdot,-1}$.
\end{lemma}
\begin{proof}
	By Remark \ref{rem:onedirection} right after Lemma \ref{lem:changepres},
	\begin{align*}
	\mathcal{F}_{\cdot}^- \cap  \mathcal{F}_{\cdot}^{\root} = \mathcal{F}_{\cdot}^{\root} \cap \sigma\left(X_\cdot(\overrightarrow{e}),\ell_{\cdot}^{\lra}(e): e\in E^{\lra-}_{\cdot} \text{ and $\overrightarrow{e}$ is directed away from }-\root \right) , \\
	\mathcal{F}_{\cdot}^+ \cap  \mathcal{F}_{\cdot}^{\root} = \mathcal{F}_{\cdot}^{\root} \cap \sigma\left(X_\cdot(\overrightarrow{e}),\ell_{\cdot}^{\lra}(e): e\in E^{\lra+}_{\cdot} \text{ and $\overrightarrow{e}$ is directed away from }+\root \right). 
	\end{align*}
	Now, since $C_{\cdot,-1}$ is $\mathcal{F}_{\cdot}^-$ measurable, by conditional independence of Lemma \ref{lem:changepres} part (ii) we have
	\begin{align*}
	\mu^{1}_{\cdot}(C_{\cdot,-1}^c\big|\mathcal{F}_{\cdot}^+,\mathcal{F}_{\cdot}^{\root}) = \mu^{1}_{\cdot}(C_{\cdot,-1}^c\big|\mathcal{F}_{\cdot}^{\root}).
	\end{align*}
	Hence, we need to show that for all $(l,x_1,x_2)$,
	\begin{align*}
	\mu^{1}_{\cdot}\left\{C_{\cdot,-1}^c\big|(\ell^{\lra}_{\cdot}(-\root,+\root),X_{\cdot}(+\root,-\root),X_{\cdot}(-\root,+\root)) = (l,x_1,x_2)\right\} = g_\cdot\left(x_1,l-x_2\right).
	\end{align*}
	By Lemma \ref{lem:changepres}, conditioned on the event $(\ell^{\lra}_{\cdot}(-\root,+\root),X_{\cdot}(+\root,-\root),X_{\cdot}(-\root,+\root))\!= \!(l,x_1,x_2)$, the distribution of the family
	\begin{align*}
	\{\ell^{\lra}_{\cdot}(e), X_{\cdot}(\overrightarrow{e}); e \in E^{\lra -}_{\cdot} \text{ and }\overrightarrow{e} \text{ is directed away from } -\root \},
	\end{align*}
	is the image of $\nu_{\cdot}(x_1)$ under the natural embedding $T_{\cdot}\to T_{\cdot}^- \subset  T_{\cdot}^\lra$. Recall that $\nu_{\cdot}(x_1)$ is the distribution of $\{\ell(e), X(\overrightarrow{e}); e \in E_\cdot \text{ and }\overrightarrow{e} \text{ is directed away from } \root\}$ on $T_{\cdot}$, given $X_{\cdot}^{\downarrow} = x_1$. Hence,
	\begin{align*}
	g_\cdot\left(x_1,l-x_2\right) &= \mu^{1}_{\cdot}\bigg\{ \ell^{\lra}_{\cdot}(-\root,-I_{\cdot}) - (l-x_2) > \\
	&\qquad\qquad{\mint_{j: \{-I_\cdot,-I_\cdot j\} \in E^{\lra-}_{\cdot}}} (\ell^{\lra}_{\cdot}(-I_\cdot,-I_\cdot j) - X_\cdot(-I_\cdot,-I_\cdot j) )\\ &\qquad\qquad\qquad\qquad\,\big|\,(\ell^{\lra}_{\cdot}(-\root,+\root),X_{\cdot}(+\root,-\root),X_{\cdot}(-\root,+\root)) = (l,x_1,x_2)\bigg\} .
	\end{align*}
	However, under this conditioning 
	\begin{align*}
	 \ell^{\lra}_{\cdot}(-\root,-I_{\cdot}) - (l-x_2) &= \ell^{\lra}_{\cdot}(-\root,-I_{\cdot}) - (\ell^{\lra}_{\cdot}(-\root,+\root) - X_{\cdot}(-\root,+\root))\\
	 &=  \ell^{\lra}_{\cdot}(-I_{\cdot},-\root) - X_{\cdot}(-I_{\cdot},-\root),
	\end{align*}
	where the last equality follows by \eqref{eq:baleq} and the fact that under $\mu^{1}_{\cdot}$, $\M^{\lra}_{\cdot,\opt}(-\root,\root) = 1$ (see Remark \ref{rem:onedirection}).
	Finally, note that under $\mu^{1}_{\cdot}$, 
	\begin{align*}
	C_{\cdot,-1}^c = \bigg\{\ell^{\lra}_{\cdot}(-I_{\cdot},-\root) - X_{\cdot}(-I_{\cdot},-\root) > 
	{\mint_{j: \{-I_\cdot,-I_\cdot j\} \in E^{\lra-}_{\cdot}}} (\ell^{\lra}_{\cdot}(-I_\cdot,-I_\cdot j) - X_\cdot(-I_\cdot,-I_\cdot j) )\bigg\}.
	\end{align*}
\end{proof}
Now, we have all the machinery to finish the proof of Lemma \ref{lem:BB}. By using the relabeling bijections, $\prob[B_{-1}\,|\,\xbar{B}_\infty, \{\M_{\opt}(\root) = 0\}]$ equals $\mu^{1}_{p}\left(C_{p,-1}\,|\, C_p\right)$ where $C_p$ is a certain event which is measurable with respect to $\mathcal{F}_{p}^+ \cap \mathcal{F}_{p}^{\root}$ such that $\mu^{1}_{p}(C_p) = \prob[\xbar{B}_\infty\cap \{\M_{\opt}(\root) = 0\}]$.
Similarly, $\prob[B_{-1}\,|\,\xbar{B}_\infty, \{\M_{\opt}(\root) \neq 0\}]$ equals $\mu^{1}_{u}\left(C_{u,-1}\,|\, C_u\right)$ for a certain event $C_u$ that is defined similar to $C_p$. Now, by Lemma \ref{lem:g} and Lemma \ref{lem:mu1calc}, if $\mu^{1}_{\cdot }(C_\cdot) > 0$, then we have
\begin{align*}
\mu^{1}_{\cdot}\left(C_{\cdot,-1}^c \cap C_\cdot \right) &= \mathbb{E}_{\mu^{1}_{\cdot}} \left[ \bs{1}_{C_\cdot}\, \mu^{1}_{\cdot}(C_{\cdot,-1}^c\big|\mathcal{F}_{\cdot}^+,\mathcal{F}_{\cdot}^{\root}) \right]\\
&= \mathbb{E}_{\mu^{1}_{\cdot}} \left[ \bs{1}_{C_\cdot} \, g_\cdot\left(X_{\cdot}(+\root,-\root),\ell^{\lra}_{\cdot}(-\root,+\root) - X_{\cdot}(-\root,+\root)\right) \right] > 0.
\end{align*}